\documentclass[preprint,11pt]{elsarticle}

\usepackage{amsfonts, amsmath, amssymb, amsthm, pbox, graphicx}
\usepackage{longtable}
\usepackage{hyperref}
\newtheorem{theorem}{Theorem}[section]
\newtheorem{lemma}[theorem]{Lemma}
\theoremstyle{definition}
\newtheorem{definition}{Definition}[section]
\newtheorem{example}{Example}[section]
\theoremstyle{remark}
\newtheorem{remark}{Remark}

\journal{Logic Journal of the IGPL}

\begin{document}

\begin{frontmatter}

\title{A Complete Logic for Database Abstract State Machines}

\author[1]{Qing Wang}
\author[2]{Flavio Ferrarotti}
\author[2]{Klaus-Dieter Schewe}
\author[2]{Loredana Tec}
 
\address[1]{Research School of Computer Science, The Australian National University, Australia, \textrm{qing.wang@anu.edu.au}}
\address[2]{Software Competence Center Hagenberg, Austria, \textrm{[flavio.ferrarotti$\mid$loredana.tec]@scch.at},\textrm{kdschewe@acm.org}}

%\ShortAuthor{Q. Wang, F. Ferrarotti, K.-D. Schewe, L. Tec}
%\LongAuthor{%
%\author{Qing Wang}
%\address{Research School of Computer Science,
%The Australian National University, Canberra, Australia\\
%Email: qing.wang@anu.edu.au}
%\author{Flavio Ferrarotti}
%\address{Software Competence Center Hagenberg, Austria \\
%Email: flavio.ferrarotti@scch.at}
%\author{Klaus-Dieter Schewe}
%\address{Software Competence Center Hagenberg, Austria \\
%and: Johannes-Kepler-University Linz, Research
%Institute for Applied Knowledge Processing, Austria \\
%Email: kd.schewe@scch.at}
%\author{Loredana Tec}
%\address{Software Competence Center Hagenberg, Austria \\
%Email: loredana.tec@scch.at}
%}

%\Received{}

\begin{abstract}
In database theory, the term \emph{database
transformation} was used to refer to a unifying treatment for
computable queries and updates.
Recently, it was shown that non-deterministic
database transformations can be captured exactly by a variant of
ASMs, the so-called Database Abstract State Machines (DB-ASMs). In
this article we present a logic for DB-ASMs,  extending the logic of Nanchen and St\"ark for ASMs. In particular, we
develop a rigorous proof system for the logic for DB-ASMs, which is
proven to be sound and complete. The most
difficult challenge to be handled by the extension is a proper
formalisation capturing non-determinism of database transformations
and all its related features such as consistency,
update sets or multisets associated with DB-ASM rules. As the database part of a state of database transformations is a
finite structure and DB-ASMs are restricted by allowing quantifiers only over
the database part of a state, we resolve this problem by taking update sets explicitly into the logic, i.e. by using an additional modal operator $[X]$, where $X$ is interpreted as an update set $\Delta$
generated by a DB-ASM rule. The DB-ASM logic provides a powerful
verification tool to study properties of database transformations.\\[1ex]

\noindent
{\normalsize\textbf{Acknowledgements.}} 
%We would like to thank the anonymous reviewers for their insightful comments which have significantly influenced the development of the final version of this paper. 
The research reported in this paper results
    from the project \textit{Behavioural Theory and Logics for
      Distributed Adaptive Systems} supported by the \textbf{Austrian
      Science Fund (FWF): [P26452-N15]}. It was further supported by the Austrian Research Promotion Agency (FFG) through the COMET funding for the Software Competence Center Hagenberg. 

\end{abstract}
\end{frontmatter}

\section{Introduction}

Queries and updates are two basic types of computations in databases, which
capture the capability to retrieve and update data, respectively. In database theory, database transformations refer to transforming database instances into other database instances, including both queries and updates. More formally, a database transformation
was defined as a binary relation on database instances over an input
schema and database instances over an output schema which must satisfy the four criteria:
well-typedness, effective computability, genericity and
functionality \cite{AbiteboulIQL89,AbiteboulDatalogExtension}.

In the past 50 years, a main research topic of database transformations was to characterise different subclasses of database transformations in terms of their logical or algebraical properties. Although prior
studies have yielded fruitful results for queries, e.g., computable queries
\cite{ChandraCompleteness}, determinate transformations
\cite{AbiteboulIQL89}, semi-deterministic transformations
\cite{VandenBusscheSemideterminism}, constructive transformations
\cite{VandenBusscheThesis,VandenBusscheCompleteness}, etc., extending these results to updates is by no means
straightforward. In \cite{VandenBusscheNondeterministic} {Van den Bussche} and {Van Gucht}
conjectured that queries and updates may have some fundamental
distinctions, and raised the question of whether there exists a theoretical
framework that can unify both queries and updates.

In many real-life database applications, database queries and
updates often turn out to have intimate connections. For instance, a relation may be updated by using matched tuples from another
(possibly the same) relation on a join operation, tuples of a relation may be deleted based on
some selection criteria, etc. On one side, such connections between
queries and updates justify the fact that the embedding of queries
in updates is supported as a fundamental feature in all major
commercial relational database systems. On the other side, it brings
up considerable concerns on the theoretical foundations for database
updates. In a sharp contrast to the elegant and fruitful theory for
database queries, the theoretical foundations for database updates,
or more generally, for a unifying framework encompassing
both queries and updates are still lacking.

Recently, the rising trend of NoSQL database applications has further increased the importance of studying database transformations as a unifying framework that encompasses both queries and updates. This is because
these application areas are often complicated and
their data is schema-less, heterogeneous, redundant,
inconsistent and frequently modified. As a result, the query mechanisms used in NoSQL database applications are more similar to traditional programming using programming languages, such as C or Java, rather than traditional queries using SQL. The distinction between database queries and updates has been considerably lessened. Thus, in order to rigorously manage and use data in these database applications, a theoretical framework for database transformations is required.

Nevertheless, formalising a unifying framework
  for database transformations is challenging. As reported in
  \cite{VandenBusscheNondeterministic}, the non-determinism of
  database transformations is a difficult problem. Particularly, in the presence of unique identifiers, the representation of unique
identifiers is irrelevant and only the interrelationship between
objects represented by them matters \cite{AbiteboulIQL89}. Indeed, the degree of
non-determinism is one of critical factors which determine the upper bound of the
expressiveness of associated languages. Apart from the issue of non-determinism,
it has also been perceived that there is a mismatch between the
declarative semantics of query languages and the operational
semantics of update languages. Unlike many query languages which can
describe what a program should accomplish rather than specify how to
accomplish, update languages usually require explicit specifications
of operations and control flow. An obvious question is ``how can different semantics be integrated within one
database language for database transformations?".

We address these challenges by developing a theoretical framework for database transformations using Abstract State Machines and studying the logical foundations of database transformations in such a theoretical framework.
Abstract State Machine (ASM) is a universal model of computation introduced by Gurevich in his well-known attempts to formalise different notions of algorithms
\cite{gurevich:tocl2000,gurevich2004abstract}.
Gurevich's sequential ASM thesis has shown that sequential ASMs can exactly capture all the
sequential algorithms that are stipulated by three postulates \cite{gurevich:tocl2000}.
As ASMs are essentially state machines operating on states that are first-order structures, dynamics between states can be rigorously captured by the concepts of updates, update sets or update multisets, as widely acknowledged in the ASM research community \cite{boerger:2003,gurevich2004abstract}. The sequential ASM thesis also sheds light into the way of establishing
a unifying theoretical framework for database transformations that contain queries and updates. Intuitively, this is based on the observation that the class of computations
described by database transformations can be formalised as a
class of ASMs respecting database principles. This has led to the development of \emph{Database Abstract State
Machines} (DB-ASMs), a variant of ASMs, as a model of computation
for database transformations, and the DB-ASM thesis which has proven that DB-ASMs
satisfy the five postulates stipulated for database transformations, and all
computations stipulated by the postulates for database
transformations can also be simulated step-by-step by a \emph{behaviourally equivalent}
DB-ASM \cite{schewe:Axiomatization}. This in a sense establishes the database analogue of
Gurevich's sequential ASM thesis \cite{gurevich:tocl2000}.

%As a
%database transformation is behaviourally equivalent to a DB-ASM and
%vice versa, a rigorous logic for DB-ASMs can thus offer vast
%advantages for reasoning about database transformations, such as
%verifying the correctness of specification, deriving static or
%dynamic properties, determining the equivalence of programs,
%comparing the expressive power of computation models, etc.

%\subsection{Contributions}
\medskip\noindent \textbf{Contributions}\hspace{0.2cm} In this paper, we study the logical foundations of the
DB-ASM thesis. Our contributions are as follows.

\medskip
Firstly, we characterise states of DB-ASMs by the
logic of meta-finite structures \cite{graedel:infcomp1998} which is
then incorporated into a logic for DB-ASMs. In database theory,
states of a database transformation are predominantly regarded as finite structures. However, when applying algorithmic
operations to tackle database-related problems, the finiteness condition on
states often turns out to be too restrictive for several reasons: (1)
database transformations may deal with new elements from countably
infinite domains, e.g. counting queries produce natural numbers
even if no natural numbers occur in a finite structure; (2) finite structures may have invariant properties that possibly have
infinite elements implied in satisfying them, such as, numerical
invariants of geometric objects or database constraints; (3) each
database transformation either implicitly or explicitly lives in a
background that supplies all necessary information relating to
computations and usually exists in the form of infinite structures.
Thus, we consider a state of database transformation as a
meta-finite structure consisting of (i) a database part, which is a
finite structure, (ii) an algorithmic part, which may be an infinite
structure, and (iii) a finite number of bridge functions between these two parts.
Characterising states of database transformations using the logic
of meta-finite structures enables us to
reason about aggregate computations commonly existing in database
applications.

Our second contribution is the handling of bounded non-determinism
in the logic of DB-ASMs. This was also the most challenging problem
we faced in this work. It is worth to mention that
non-deterministic transitions manifest themselves as a very
difficult task in the logical formalisation for ASMs. Nanchen and St\"ark analysed potential problems to several
approaches they tried by taking non-determinism into consideration
and concluded \cite{RobertLogicASM}:

\begin{quote}
    ``Unfortunately, the formalisation of consistency cannot be applied directly to non-deterministic ASMs. The formula Con$(R)$ (as defined in Sect. 8.1.2 of \cite{boerger:2003}) expresses the property that the \emph{union of all possible} update sets of $R$ in a given state is consistent. This is clearly not what is meant by consistency. Therefore, in a logic for ASMs with \textbf{choose} one had to add Con$(R)$ as an atomic formula to the logic."
\end{quote}

However, we observe that this conclusion is not necessarily true, as
finite update sets can be made explicit in the formulae of a logic
to capture non-deterministic transitions. In doing so, the
formalisation of consistency defined in \cite{RobertLogicASM} can
still be applied to such an explicitly specified finite update set $U$
yielded by a rule $r$ in the form of the formula con$(r,X)$ where the second-order variable $X$ is interpreted by $U$, as will be discussed in Section \ref{sub:Consistency}. We can thus solve
this problem by adding the modal operator $[X]$ for an
update set generated by a DB-ASM rule. In doing so, DB-ASMs are restricted to have quantifiers only over the
database part of a state which is a finite structure, and consequently update sets
(or multisets) yielded by DB-ASM rules are restricted to be
finite.
%\footnote{We should remark here that for general parallel ASMs this finiteness should also be enforced in the \textbf{forall}-rules. However, in ASMs this only appears implicitly by using the reserve set of values. Nonetheless, it looks as if the result achieved in this paper could be carried back to improve the theory and logic of ASMs.}.
Hence, the logic for DB-ASMs is empowered to capture
non-deterministic database transformations.

Our third contribution is the development of a proof system for the
logic for DB-ASMs, which extends the proof system for the logic for
ASMs \cite{RobertLogicASM} in several aspects:

 \begin{itemize}

   \item DB-ASMs can collect updates yielded in parallel computations under the
multiset semantics, i.e. update multisets, then aggregate
updates in an update multiset to an update set by applying so-called \emph{location
operators}. Our proof system can capture this by incorporating the
axioms for both the predicate of update multisets and the predicate
of update sets. The axioms also specify the interaction between
update multisets and update sets in relating to DB-ASM rules.

   \item A DB-ASM rule may be associated
with different update sets. Applying different update sets
may lead to different successor states to the current
state. As the logic for DB-ASMs includes formulae denoting explicit
update sets and update multisets, and second-order variables that are
bound to update sets or update multisets, our proof system allows us to
reason about the interpretation of a formula over all successor
states or over some successor state after applying a DB-ASM rule
over the current state.

 \item In addition to capturing the consistency of an update set yielded by a DB-ASM rule, our proof system also develops two notions of consistency for a DB-ASM rule (i.e. weak version and strong version). When a DB-ASM rule is deterministic, these two notions coincide.

 \end{itemize}

Our last contribution is a proof of the completeness of the logic for
DB-ASMs. Due to the importance of non-determinism for enhancing
   the expressive power of database transformations and for specifying database transformations at flexible levels of abstraction, DB-ASMs take into account choice rules. Consequently, the logic for DB-ASMs has to handle
   all the issues related to non-determinism which have been identified as the source of problems in the completeness proof of the logic for ASMs \cite{RobertLogicASM, [BS03]}. Nevertheless, we prove that we can use a Henkin semantics for the required (in our approach) second-order quantification, and thus despite of the inclusion of second-order formulae in the logic for DB-ASMs, we can establish a sound and complete proof system for the logic of DB-ASMs. Note that, the logic for DB-ASMs preserves the restriction of the logic of Nanchen and St\"ark for ASMs \cite{RobertLogicASM} dealing only with properties of a \emph{single step} of an ASM, \emph{not} with properties of whole ASM runs. This restriction to a one-step logic allows us to define a Hilbert-style proof theory and to show its completeness, whereas for a logic dealing with properties of whole ASM runs (and even more so, whole DB-ASMs runs) can hardly be expected to be complete.

%For the interest of readers, we also present a second approach
% proving the completeness of the logic for DB-ASMs based on a Henkin
% model construction.

%\subsection{Structure of the Article}
\medskip\noindent \textbf{Outline}\hspace{0.2cm} The remainder of the article is structured as follows. Section~\ref{sec:relatedwork} discusses related work on logical
characterisations of database transformations. Then we provide a motivating example in Section \ref{sec:motivation}. In Section \ref{sec:states} we discuss meta-finite structures and states of DB-ASMs. In Section~\ref{sec:adtm} we present the definitions of DB-ASM.
As states of a database transformation are meta-finite structures, in Section~\ref{sec:dtc}
we define a logic for DB-ASMs that is built upon the logic of
meta-finite structures. Subsequently, a
detailed discussion of basic properties of the logic for DB-ASMs, such as
consistency, update sets and multisets, along with
the formalisation of a proof system is presented in Section~\ref{sec:proof_system}.
In Section~\ref{sec:soundness}, we present some interesting properties of the logic for DB-ASMs which are implied by the axioms and rules of the proof system introduced in Section~\ref{sec:proof_system}.  We prove in Section~\ref{sec:completeness} that the logic for DB-ASMs is complete. We conclude the article with a brief summary in Section \ref{sec:conclusion}.

% Two alternative completeness proofs  and by
% an explicit Henkin construction are elaborated in Subsection
% \ref{sect:definitional_extension} and Appendix
% \ref{sect:henkin_construction}, respectively.

\section{Related Work}\label{sec:relatedwork}
It is widely acknowledged that a logic-based perspective for
database queries can provide a yardstick for measuring the
expressiveness and complexity of query languages. To extend the
application of mathematical logics from database queries to database
updates, a number of logical formalisms have been developed
providing the reasoning for both states and state changes in a
computation model \cite{Survey98,LogicSurvey01}. A popular approach
was to take dynamic logic as a starting point and then to define the
declarative semantics of logical formulae based on Kripke
structures. It led to the development of the \emph{database dynamic
logic} (DDL) and \emph{propositional} \emph{database dynamic logic}
(PDDL) \cite{SpruitDDL92,SpruitPhDThesis94,SpruitPDDL95}. DDL has
atomic updates for inserting, deleting and updating tuples in
predicates and for functions, whereas PDDL has two kinds of atomic
updates: passive and active updates. Passive updates change the
truth value of an atom while active updates compute derived updates
using a logic program. In \cite{SpruitFUL01} Spruit, Wieringa and Meijer proposed \emph{regular
first-order update logic} (FUL), which generalises
dynamic logic towards specification of database updates. A state of
FUL is viewed as a set of non-modal formulae. Unlike standard
dynamic logic, predicate and function symbols rather than variables
are updatable in FUL. There are two instantiations of FUL. One is
called \emph{relational algebra update logic} (RAUL) that is an
extension of relational algebra with assignments as atomic updates.
Another one is DDL that parameterizes FUL by two kinds of atomic
updates: bulk updates to predicates and assignment updates to
functions. It was shown that DDL is also ``update complete'' in
relational databases with respect to the update completeness
criterion proposed by Abiteboul and Vianu in \cite{AbiteboulUpdate87}.

As we explained before ASMs turn out to be a promising approach for
specifying database transformations. The logical foundations for
ASMs have been well studied from several perspectives.
Groenboom and Renardel de Lavalette presented in \cite{GroenboomMLCM94} a logic called \emph{modal logic of
creation and modication} (MLCM) that is a multimodal predicate logic
intended to capture the ideas behind ASMs. On the basis of MLCM they developed a language called
\emph{formal language for evolving algebras} (FLEA) \cite{GroenboomFLEA95}. Instead of
values of variables, states of an MLCM are represented by
mathematical structures expressed in terms of dynamic functions. The
work in \cite{LavalettelogicMCL01} generalises MLCM and other
variations from \cite{FenselMLPM96} to \emph{modification and
creation logic} (MCL) for which there exists a sound and complete
axiomatisation. In \cite{Schoenegge95} Sch\"onegge presented an extension of dynamic
logic with update of functions, extension of universes and
simultaneous execution (called EDL), which allows
statements about ASMs to be directly represented. In addition to
these, a logic complete for \emph{hierarchical ASMs} (i.e., ASMs
that do not contain recursive rule definitions) was developed by Nanchen and St\"ark in
\cite{RobertLogicASM}. This logic for ASMs differs from other logics
in two respects: (1) the consistency of updates has been accounted
for; (2) modal operators are allowed to be eliminated in certain
cases. As already remarked, the ASM logic of Nanchen and St\"ark permits reasoning about ASM rules, but not about ASM runs, which is the price to be paid for obtaining completeness.
In this article we will extend this logic for ASMs towards
database transformations, in which states are regarded as
meta-finite structures and a bounded form of non-determinism is
captured.

It was Chandra and Harel who first observed limitations of finite
structures in database theory \cite{ChandraCompleteness}. They
proposed a notion of an \emph{extended database} that extends finite
structures by adding another countable, enumerable domain containing
interpreted features such as numbers, strings and so forth. The
intention of their study was to provide a more general framework
that can capture queries with interpreted elements.
Another extension of finite structures was driven by the efforts
to solve the problem of expressing cardinality properties
\cite{Yuri02onpolynomial,CaiIFP+C89,GrCounting,Immerman87,Otto95b,Otto96,Torres01,Torres02}.
For example, Gr\"adel and Otto developed a two-sorted structure that
adjoins a one-sorted finite structure with an additional finite
numerical domain and added the terms expressing cardinality
properties~\cite{GrCounting}. They aimed at studying the expressive power
of logical languages that involve induction with counting on such
structures. A promising line of work is \emph{meta-finite model
theory}. Gr\"adel and Gurevich in \cite{graedel:infcomp1998} defined
\emph{meta-finite structures}.
%consisting of a primary part that is a finite structure, a secondary part that may be a finite or infinite structure, and a set of weight functions from the primary part into the secondary part, and further extended a logic suitable for finite structures (e.g., first-order logic, fixed point logic, the infinitary logic, etc.) to a logic of meta-finite structures.
Based
on the work presented in \cite{graedel:infcomp1998}, Hella et al. in
\cite{HellaAggregateLogics} studied the logical grounds of query
languages with aggregation, which is closely related to our work presented in
this article. However, the logic for DB-ASMs covers not only
database queries with aggregation but also database updates. Put it in
another way, it is a logical characterisation for database
transformations including aggregate computing and sequential algorithms.

\section{Motivating Example}\label{sec:motivation}
To motivate our work, we use an example to illustrate how database
transformations can be captured by DB-ASMs and how a logic for
DB-ASMs can be used for verifying database transformations, i.e., one of the potential applications of the logic for DB-ASMs.

\begin{example}\label{exa:DBASM-db}
Consider a relational database schema: \textsc{City}=$\{$Cid, Name$\}$ and 
\textsc{Route}= \\
$\{$FromCid, ToCid, Distance$\}$, which store route information of cities and their distance.
Assume that we have $\forall c_1, c_2,d ((c_1, c_2, d) \in \textsc{Route}\rightarrow (c_2, c_1, d) \in \textsc{Route})$. Then a relation of $\textsc{Route}$ corresponds to an undirected graph in which the nodes represent cities and the edges represent direct routes, for example, the undirected graph in Fig. \ref{fig:graph} corresponds to the relation of $\textsc{Route}$ in Fig. \ref{fig:initialstate}. Assume that such graphs are always connected. Let $Q_1(c)$ be the query ``find a shortest path tree rooted at city $c$''. To answer this query, we would need to find a spanning tree $T$ with the root node $c$ such that the path distance from $c$ to any other node $c'$ in $T$ is the shortest path distance from city $c$ to $c'$ in the graph induced by $\textsc{Route}$.
\begin{figure}[!ht]
\begin{minipage}{2.5cm}
\hspace*{1cm}
\end{minipage}
\begin{minipage}{2.5cm}
{\normalsize\begin{tabular}{|c|c|}\hline
 \multicolumn{2}{|c|}{\textsc{City}}\\
  \hline
  Cid & Name \\\hline
  $c_1$ & A \\
  $c_2$ & B \\
  $c_3$ & C \\
  $c_4$ & D \\
  $c_5$ & E \\
  \hline
\end{tabular}}
\end{minipage}
\begin{minipage}{3.6cm}
\vspace*{1cm}
{\centering
\begin{tabular}{|c|c|c|}
  \hline
   \multicolumn{3}{|c|}{\textsc{Route}}\\\hline
  FromCid & ToCid & Distance \\\hline
  $c_1$ & $c_2$ & $d_1$ \\
  $c_1$ & $c_3$ & $d_3$ \\
  $c_2$ & $c_4$ & $d_2$ \\
  $c_3$ & $c_4$ & $d_1$ \\
  $c_3$ & $c_5$ & $d_2$ \\
  $c_5$ & $c_4$ & $d_4$ \\
  \dots & \dots  & \dots\\
  \hline
\end{tabular}}
\medskip
\end{minipage}

\begin{minipage}{2.5cm}
\hspace*{1cm}
\end{minipage}\begin{minipage}{2cm}
{\normalsize\begin{tabular}{|c|}\hline
 \multicolumn{1}{|c|}{\textsc{Visited}}\\
  \hline
  Cid \\\hline
  \\
  \hline
\end{tabular}}
\end{minipage}
\begin{minipage}{2.5cm}
\hspace*{0.2cm}{\normalsize\begin{tabular}{|c|c|c|}\hline
 \multicolumn{3}{|c|}{\textsc{Result}}\\
  \hline
  Cid & TotalCost & LastStop \\\hline
  & &\\
  \hline
\end{tabular}}
%\end{minipage}
\end{minipage}\caption{An initial state}\label{fig:initialstate}
\end{figure}
%\medskip

\begin{figure}[!htbp]
\begin{minipage}{5cm}
\begin{center}
\includegraphics[scale=0.65]{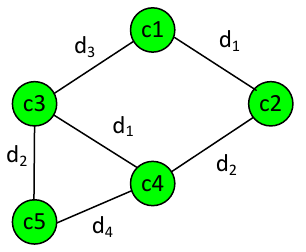}
\end{center}\caption{An undirected graph}\label{fig:graph}
\end{minipage}
%\begin{minipage}{1cm}
%\hspace*{1cm}
%\end{minipage}
\begin{minipage}{7cm}
{\centering\begin{tabular}{|c|}\hline
 \multicolumn{1}{|c|}{$\textsc{Dist}:\mathrm{Cid} \rightarrow \mathbb{N}$}\\
  \hline
    $\textsc{Dist}(c_1)=1$ \\
  $\textsc{Dist}(c_2)=2$ \\
   $\textsc{Dist}(c_3)=3$ \\
   $\textsc{Dist}(c_4)=4$ \\
   $\textsc{Dist}(c_5)=5$ \\
\hline
 \end{tabular}}
{\centering \begin{tabular}{|c|}\hline
 \multicolumn{1}{|c|}{$\textsc{Val}:\mathrm{Distance} \rightarrow \mathbb{N}$}\\
  \hline
    $\textsc{Val}(d_1)=500.50$ \\
  $\textsc{Val}(d_2)=100.00$ \\
   $\textsc{Val}(d_3)=808.20$ \\
   $\textsc{Val}(d_4)=203.20$ \\
\hline
 \end{tabular}}
 \caption{Two bridge functions}
\end{minipage}

\end{figure}

The DB-ASM rule in Fig. \ref{exa:DBASM-code} (of the signature $\Upsilon_G$ described next), which corresponds to the famous Dijkstra's algorithm, expresses the query $Q_1(c)$. Let $\Upsilon_G$ be the signature in Fig.~\ref{exa:DBASM-signature}.
Apart from \textsc{City} and \textsc{Route}, $\Upsilon_G$ includes \textsc{Visited}=$\{$Cid$\}$ to store the cities that have been visited during the computation and \textsc{Result}=$\{$ChildCid, ParentCid$\}$ to store the shortest path tree as a child-parent node relationship. We assume that in every initial state the relations \textsc{Visited} and \textsc{Result} are empty (as shown in Fig. \ref{fig:initialstate}) and that \textsc{Initial}=\textsc{True}. We also assume that the constant symbol \textsc{Infinity} is interpreted by a natural number which is strictly greater than the sum of all the distances in $\textsc{Route}$, that $\textsc{Zero}$ is interpreted by the value $0$,
%that \textsc{Null} is a constant symbol interpreted by an element which is \emph{not} a city id,
and that the values interpreting the constant symbols
%\textsc{Null},
\textsc{True} and \textsc{False} are
%pairwise
different. A state in which every city has been visited, i.e., a state in which \textsc{Visited} contains every city id in the database, is considered as a final state. Notice that entries in the database part of the states which correspond to the distances between adjacent nodes in \textsc{Route} are surrogates for the actual distances which are natural numbers in the algorithmic part of the state. Thus, there are two bridge functions:
 \begin{itemize}
 \item $\textsc{Dist}:\mathrm{Cid} \rightarrow \mathbb{N}$ to keep track of cities visited during a computation and their corresponding shortest distances to $c$, respectively;
 \item $\textsc{Val}:\mathrm{Distance} \rightarrow \mathbb{N}$ to map the surrogates for the actual distances to the natural numbers in the algorithmic part.
 \end{itemize}
%We use the rule \textbf{import} $x$ \textbf{do} $r$ \textbf{enddo} which imports a new ``fresh'' element from a possible infinite reserve to the domain of the database part of the current state. We omit a formal definition of the concept of reserve as well as of the import rule since these technical details are well understood \cite{boerger:2003} and not central for our presentation.

\begin{figure}[!ht]
\hspace*{1cm}\begin{minipage}{6cm}
{\centering\fbox{\parbox{11.5cm}{$\Upsilon_{G} = (\Upsilon_{db}, \Upsilon_a, \mathcal{F}_b)$, where\\[0.2cm]
- \textsc{City}, \textsc{Route}, \textsc{Visited}, \textsc{Result}, c, \textsc{True}, \textsc{False}, \textsc{Initial} $\in\Upsilon_{db}$;\\[0.2cm]
- \textsc{Infinity}, \textsc{Zero}, \textsc{MDist} $\in\Upsilon_a$;\\[0.2cm]
- \textsc{Val}, \textsc{Dist} $\in\mathcal{F}_b$; and \\[0.2cm]
- \textsc{Min} is a location operator.
}} \caption{A signature } \label{exa:DBASM-signature}}
\end{minipage}
\end{figure}

In general, the DB-ASM in Fig. \ref{exa:DBASM-code} proceeds in two stages:

\begin{itemize}

\item The first
stage is described by Lines 2-12. The DB-ASM starts with an initial state in which
\textsc{Initial} = \textsc{True}, then assigns (in parallel) to every city a tentative distance
value (i.e., 0 for the city $c$ and \textsc{Infinity} for all other
cities), and ends with \textsc{Initial}=\textsc{False}.

\item The second stage is described by Lines 13-36. The
shortest paths to reach other cities from the city $c$ are repeatedly
calculated and stored in \textsc{Result} until a final state in which every city has been visited is reached.

\end{itemize}

At Line 15 of the DB-ASM rule in Fig.~\ref{exa:DBASM-code} the location operator
\textsc{Min} is assigned to the location $(\textsc{MDist},())$, and thus
\textsc{MDist} is updated to the shortest distance among the collection of distances
between the city in consideration and all its unvisited neighbor
cities. At Line 20, we can see that the DB-ASM is non-deterministic
because a city is arbitrarily chosen from the non-visited cities
 whose shortest paths are equally minimum at each step
of the computation process. This indeed exemplifies the importance
of non-determinism for specifying database transformations at a
high-level of abstraction.

Now suppose that we want to know whether the properties~P1 and~P2 described next, are satisfied by the DB-ASM corresponding to the DB-ASM rule in Fig.~\ref{exa:DBASM-code} over certain states of signature $\Upsilon_G$. Clearly, the use of a logic to specify such properties of DB-ASMs can contribute significantly to the verification of the correctness of database transformations expressed by means of DB-ASMs. Although the logic proposed for DB-ASMs in this paper can only reason about such properties within one-step of computation, it nevertheless provides a useful tool which is a first step towards developing a logic that can reason about properties of whole DB-ASMs runs.

\begin{enumerate}

\item[(P1)] In every non-initial state of a run, each city in the child/parent node relationship encoded in \textsc{Result} has exactly one parent city, except for $c$ which has none. In other words, \textsc{Result} encodes a tree with root node $c$.\\[0.2cm]
$\neg \textsc{Initial} \rightarrow$\\*
\hspace*{0.8cm}$\forall x y (\textsc{Result}(x,y) \rightarrow  x \neq c \wedge \forall z (z \neq y \rightarrow \neg \textsc{Result}(x,z))) \wedge \exists x (\textsc{Result}(x,c))$\\[0.2cm]

\item[(P2)] In every state of a run, if a city not yet visited (by the algorithm) is a neighbour city of one already visited, then the calculated (shortest so far) distance from $c$ to that city is strictly less than \textsc{Infinity} already.\\[0.2cm]
$\forall x y (\textsc{Visited}(x) \wedge \neg\textsc{Visited}(y) \wedge \exists z (\textsc{Route}(x,y,z)) \rightarrow \textsc{Dist}(y) < \textsc{Infinity})$

\end{enumerate}
\end{example}

\begin{figure}[!ht]
\small
1\hspace{0.2cm}\textbf{par}

2\hspace{0.7cm}\textbf{if} \textsc{Initial} \textbf{then}

3\hspace{1.2cm}\textbf{par}

4\hspace{1.7cm}\textbf{forall} $x$ \textbf{with} $\exists y ( \textsc{City}(x,y) ) \, \textbf{do}$

5\hspace{2.2cm}\textbf{par}

%6\hspace{2.7cm}$\textsc{Result}(x, \textsc{Null})$ := \textsc{True}

6\hspace{2.7cm}\textbf{if} $x = c$ \textbf{then} $\textsc{Dist}(x)$ := \textsc{Zero} \textbf{endif}

7\hspace{2.7cm}\textbf{if} $x \neq c$ \textbf{then} $\textsc{Dist}(x)$ := \textsc{Infinity} \textbf{endif}

8\hspace{2.2cm}\textbf{endpar}

9\hspace{1.70cm}\textbf{enddo}

10\hspace{1.55cm}$\textsc{Initial} := \textsc{false}$

11\hspace{1.05cm}\textbf{endpar}

12\hspace{0.55cm}\textbf{endif}

13\hspace{0.55cm}\textbf{if} $\neg \textsc{Initial}$ \textbf{then}

14\hspace{1.05cm}\textbf{seq}

15\hspace{1.55cm}\textbf{let} $(\textsc{MDist},()) \rightharpoonup \textsc{Min}$ \textbf{in}

16\hspace{2.05cm}\textbf{forall} $x$ \textbf{with} $\exists y ( \textsc{City}(x,y) \wedge \neg \textsc{Visited}(x))$ \textbf{do}

17\hspace{2.55cm}$\textsc{MDist} := \textsc{Dist}(x)$

18\hspace{2.05cm}\textbf{enddo}

19\hspace{1.55cm}\textbf{endlet}

20\hspace{1.55cm}\textbf{choose} $x$ \textbf{with} $\textsc{Dist}(x)=\textsc{MDist} \wedge \neg \textsc{Visited}(x)$ \textbf{do}

21\hspace{2.05cm}\textbf{par}

22\hspace{2.55cm}$\textsc{Visited}(x) := \textsc{True}$

23\hspace{2.55cm}\textbf{forall} $y, z$ \textbf{with} $\textsc{Route}(x,y,z) \wedge \neg \textsc{Visited}(y) \wedge$

24\hspace{5.1cm}$\textsc{MDist}+\textsc{Val}(z) < \textsc{Dist}(y)$ \textbf{do}

25\hspace{3.05cm}\textbf{par}

26\hspace{3.55cm}$\textsc{Dist}(y) := \textsc{MDist}+\textsc{Val}(z)$

27\hspace{3.55cm}$\textsc{Result}(y,x) := \textsc{True}$

28\hspace{3.55cm}\textbf{forall} $x'$ \textbf{with} $x' \neq x \wedge \textsc{Result}(y,x')$ \textbf{do}

29\hspace{4.05cm}$\textsc{Result}(y,x') := \textsc{False}$

30\hspace{3.55cm}\textbf{enddo}

31\hspace{3.05cm}\textbf{endpar}

32\hspace{2.55cm}\textbf{enddo}

33\hspace{2.05cm}\textbf{endpar}

34\hspace{1.55cm}\textbf{enddo}

35\hspace{1.05cm}\textbf{endseq}

36\hspace{0.55cm}\textbf{endif}

37\hspace{0.1cm}\textbf{endpar}
\caption{A DB-ASM}
\label{exa:DBASM-code}
\end{figure}

%\begin{remark}
%In the previous example we make several implicit assumptions. We assume that the signature $\Upsilon$ of the DB-ASMs includes a pairing function symbol $(x,y)$ which evaluates to the ordered pair with first element $x$ and second element $y$. We also assume a nullary static function (constant) symbol $()$ which is interpreted as the empty tuple and a unary function symbol $"x"$ which evaluates to the literal $x$. Finally, for the interpretation of these function symbols we assume that the states include all possible pairs and strings. These elements are part of what is known in the ASM literature as the background of computation \cite{Blass-Background-00} which is usually needed by any meaningful algorithm.
%\end{remark}

%In these examples we abuse the notation by allowing more than one variable to be associated to one forall rule. This allows us to present shorter and more readable DB-ASM rules, but such rules can always be replaced by well formed ones (i.e., that respect our formal definition) by using nested forall rules.

\section{Meta-finite Structures as States}\label{sec:states}

Meta-finite structures were originally studied by Gr{\"a}del and Gurevich in order to extend the methods of finite model theory beyond finite structures \cite{graedel:infcomp1998}. In a nutshell, a meta-finite structure consists of (a) a
primary part, which is a finite structure, (b) a secondary part, which is a (usually
infinite) structure, and
(c) a set of functions mapping from the primary part into the second part. Typical examples of meta-finite structures are finite objects arising in many areas of computer science, which usually consist of both structures and
numbers. For example, graphs with weights on the edges, where a graph may be representable by a finite structure but its weights on the edges may be reals from an infinite domain, and arithmetical operations performed on these weights may not be any a \emph{priori} fixed finite subdomain \cite{graedel:infcomp1998}. Another example is relational databases in which each relation contains only a finite number of tuples. Although theoretically a relational database is viewed as a finite structure, attribute domains of a relation are often assumed to be countably infinite. In particular, such domains may be infinite mathematical structures, e.g., the natural numbers with arithmetic, rather than merely plain sets, and such infinite mathematical structures are widely used by aggregate queries in many real-life database applications.

In \cite{Gurevich-New-Thesis} Gurevich argued that ASMs provide a model of computation that is more powerful and more universal than other standard models of computation such as Turing machines, in the sense that any algorithm, however abstract, can be simulated step-for-step by an ASM. This is because a state of an ASM is abstract, which may include any real world objects and functions at a chosen level of abstraction. Let us consider for example algorithms that work with graphs. The conventional computation models require a string representation of the given graph or similar,  even in those cases when the algorithm is independent of the graph representation. In particular, a same database might have different representations, but the meaning of the data should not change. Database query languages are supposed to reflect only representation-independent properties. 

In ASMs, states are viewed as first-order structures, whereas in DB-ASMs we consider states as meta-finite structures
\cite{schewe:Axiomatization}. Conceptually, each state of a DB-ASM has a finite
\emph{database part} and a possibly infinite \emph{algorithmic part}, which are
linked via \emph{bridge functions} such that actual database entries
are treated merely as surrogates for the real values. This permits
a database to remain finite while allowing database entries to be
interpreted in possibly infinite domains such as the natural numbers
with arithmetic operations. A signature $\Upsilon$ of states
comprises (i) a sub-signature $\Upsilon_{db}$ for the database part, (ii) a
sub-signature $\Upsilon_a$ for the algorithmic part and (iii) a finite set
$\mathcal{F}_b$ of bridge function names. The \emph{base set} of a
state $S$ is a nonempty set of values $B=B_{db}\cup B_{a}$, where
$B_{db}$ is finite, and $B_a$ contains natural numbers, i.e., $\mathbb{N}\subseteq B_a$. Function
symbols $f$ in $\Upsilon_{db}$ and $\Upsilon_{a}$, respectively, are
interpreted as functions $f^S$ over $B_{db}$ and $B_{a}$, and the
 interpretation of a k-ary function symbol $f\in\mathcal{F}_b$
defines a function $f^S$ from $B^{k}_{db}$ to $B_a$. For every state
over $\Upsilon$, the restriction to $\Upsilon_{db}$ results in a
finite structure.

Since the states of DB-ASMs are defined as meta-finite structures, we now need to define a matching logic so that it can be used in the conditional statements of DB-ASMs. That is, we need a  logic of meta-finite structures as introduced in~\cite{graedel:infcomp1998}. Logics of meta-finite structures distinguish among two types of terms. The first type, which we call \emph{database terms}, denote elements of the primary (finite) part of the meta-finite structure. The second type, which we call \emph{algorithmic terms}, denote elements of the secondary (possibly infinite) part of the meta-finite structure. 

\begin{definition}\label{termsOfLFO}
Let $\Upsilon = \Upsilon_{db} \cup \Upsilon_a \cup \mathcal{F}_b$ be a signature of meta-finite states. Fix a countable set $\mathcal{X}_{db}$ 
of first-order variables, denoted with standard lowercase letters $x, y, z , \ldots$, that range over the primary database part of the meta-finite states (i.e., the finite set $B_{db}$). The set of \emph{database terms} $\mathcal{T}_{db}$ is defined as the closure of the set ${\cal X}_{db}$ of variables under the application of function symbols in $\Upsilon_{db}$. We assume that $\Upsilon_{db}$ always include a function symbol for equality. In turn, the set of \emph{algorithmic terms} $\mathcal{T}_{a}$ is defined inductively as follows:
\begin{itemize}
\item If $t_1, \ldots, t_n$ are database terms in $\mathcal{T}_{db}$ and $f$ is an $n$-ary bridge function symbol in $\mathcal{F}_b$, then $f(t_1, \ldots, t_n)$ is an algorithmic term in $\mathcal{T}_{a}$. 
\item If $t_1, \ldots, t_n$ are algorithmic terms in $\mathcal{T}_{a}$ and $f$ is an $n$-ary function symbol in $\Upsilon_a$, then $f(t_1, \ldots, t_n)$ is an algorithmic term in $\mathcal{T}_{a}$.
\item Nothing else is an algorithmic term in $\mathcal{T}_{a}$. 
\end{itemize}
We set ${\cal T}_{\Upsilon, {\cal X}_{db}}=\mathcal{T}_{db}\cup\mathcal{T}_{a}$.
\end{definition}

In this context, a \emph{variable assignment} (or \emph{valuation}) $\zeta$ is a function which assigns to every variable in ${\cal X}_{db}$ a value in the base set of the database part $B_{db}$ of the meta-finite state $S$. The  value of a term $t \in {\cal T}_{\Upsilon, {\cal X}_{db}}$ in a state $S$ under a valuation $\zeta$, denoted $\mathit{val}_{S, \zeta}(t)$, is defined as usual in first-order logic, i.e., using the classical Tarski's semantics.

The \emph{logic of meta-finite states} ${\cal L}^{FO}$ which we use in the formalization DB-ASMs is defined as the first-order logic with equality which is built up from equations between terms in ${\cal T}_{\Upsilon, {\cal X}_{db}}$ by using the standard connectives and first-order quantifiers. Its semantics is defined in the standard way. The truth value of a formula of meta-finite states $\varphi$ in $S$ under the valuation $\zeta$ is denoted as $[\![\varphi]\!]_{S,\zeta}$.

\section{Database Abstract State Machines}\label{sec:adtm}

Our work in this paper concerns the model of \emph{Database Abstract State Machine} (DB-ASM) that captures the class of database transformations defined by the postulates in the DB-ASM thesis \cite{schewe:Axiomatization,DBLP:wangphdbook}. Accordingly, we assume that states of DB-ASMs are meta-finite structures which include a minimum background of computation as required by the background postulate in the axiomatization of database transformations in~\cite{schewe:Axiomatization,DBLP:wangphdbook} (note that this is essentially the same background that is required in the parallel ASM thesis \cite{blass:tocl2003, GurevichParallelCorrection08, FerrarottiSTW16}). That is, every state of a DB-ASM includes:
\begin{itemize} 
\item An infinite reserve of values not used in a current state, but available to be added to the active domain in any state transition.
\item Boolean values ($\mathit{true}$ and $\mathit{false}$), Bolean operations ($\neg$, $\wedge$, $\vee$ and $\rightarrow$), and the undefinedness value ($\textit{undef}$). 
\item A pairing constructor and a multiset constructor together with necessary operators on tuples and multisets. 
\end{itemize}

The bounded exploration postulate in the DB-ASM thesis~\cite{schewe:Axiomatization,DBLP:wangphdbook}, as well as the bounded exploration postulates in Gurevich's sequential ASM thesis~\cite{gurevich:tocl2000} and in the new parallel ASM thesis~\cite{FerrarottiSTW16} (which simplifies the parallel ASM thesis of Blass and Gurevich~\cite{blass:tocl2003, GurevichParallelCorrection08}), are motivated by the \emph{accessibility principle}, which can be defined as the prerequisite that each location of a state must be uniquely identifiable. In fact, unique identifiability also applies to databases as emphasised by Beeri and Thalheim in~\cite{beeri:fomlado1998}, and has to be claimed for the basic updatable units in a database, for example, objects in~\cite{schewe:actacyb1993}. The accessibility principle is also a fundamental assumption used in the characterization proofs of the DB-ASM thesis as well as of the sequential and parallel ASM thesis. We therefore assume that it holds for every state of the DB-ASMs. 

As explained in~\cite{gurevich:tocl2000}, an algorithm $A$ can access an element $a$ of a state by using formulae $\varphi$ and $\psi(x)$ such that $\varphi$ is a sentence and $x$ is the only free variable in $\psi(x)$, and the equation $\psi(x) = \textit{true}$ has a unique solution in every state $S$ satisfying $\varphi$. If this information is available, then $A$ can evaluate $\varphi$ at a given state $S$, and provided that $\varphi$ holds in $S$, point to the unique solution $a$ of the equation $\psi(x) = \textit{true}$. To bridge the gap between the formula $\psi(x)$ and the element $a$, a new nullary function symbol $c$ is introduced, where $c$ is interpreted as the unique solution of the equation $\psi(x) = \textit{true}$ if $\varphi$ holds and as $\mathit{undef}$ otherwise. Using this approach, any given algorithm $A$ (database transformation or DB-ASM for that matter) can be formalized so that it can access any location of its states, and therefore satisfies the accessibility principle. We avoid the formal details here as this is a well known fact in the ASM community. For the remainder of this paper, we simply assume that every element $a$ of a state $S$ can be accessed by producing an appropriate nullary function symbol $c_a$.

\subsection{Syntax of Rules}
For simplicity, we consider function arguments as tuples. That is, if $f$ is an $n$-ary function and $t_1, \ldots, t_n$ are arguments for $f$, we write $f(t)$ where $t$ is a term which evaluates to the tuple $(t_1, \ldots, t_n)$.
Let $t$ and $s$ denote terms in ${\cal T}_{\Upsilon, {\cal X}_{db}}$, $f$ a dynamic function symbol in $\Upsilon$ and let $\varphi$ denote an ${\cal L}^{FO}$-formula of vocabulary $\Upsilon$. The set of {\em DB-ASM rules} over $\Upsilon$ is inductively defined as follows:

\begin{itemize}

\item {\em assignment rule}: update the content of $f$ at the argument $t$ to
$s$;

\hspace{3cm}$f(t) := s$

\item {\em conditional rule}: execute the rule $r$ if $\varphi$ is true; otherwise, do
nothing;

\hspace{3cm}\textbf{if} $\varphi$ \textbf{then} $r$ \textbf{endif}

\item {\em forall rule}: execute the rule $r$ in parallel for each $x$ satisfying
$\varphi$;

\hspace{3cm}\textbf{forall} $x$ \textbf{with} $\varphi$
\textbf{do} $r$ \textbf{enddo}

\item {\em choice rule}: choose a value of $x$ that satisfies $\varphi$ and then execute the rule $r$;

\hspace{3cm}\textbf{choose} $x$ \textbf{with} $\varphi$
\textbf{do} $r$ \textbf{enddo}

\item {\em parallel rule}: execute the rules $r_1$ and $r_2$ in
parallel;

\hspace{3cm}\textbf{par} $r_1$ $r_2$ \textbf{endpar}

\item {\em sequence rule}: first execute the rule $r_1$ and then execute the rule
$r_2$;

\hspace{3cm}\textbf{seq} $r_1$ $r_2$ \textbf{endseq}

\item {\em let rule}: aggregates, using the location operator $\rho$, all updates to the location $(f,t)$ yielded by $r$ (see definition of location and location operator in Section~\ref{usetsandmultisets} next);

\hspace{3cm}\textbf{let} $(f,t) \!\rightharpoonup\! \rho$ \textbf{in}
$r$ \textbf{endlet}

\end{itemize}
Notice that all variables appearing in a DB-ASM rule are \emph{database
variables} that must be interpreted by values in $B_{db}$. A rule
$r$ is {\em closed} if all variables of $r$ are bounded by forall and
choice rules.

%A (first-order) variable assignment $\zeta$ for a (meta-finite) state $S$ is a finite function
%which assigns elements in the (finite) base set $B_{db}$ of the database part of $S$ to a finite number of
%first-order variables. If $\zeta$ is a variable assignment, then $\zeta [x_1
%\mapsto a_1 ,\dots, x_n \mapsto a_n]$ is another variable assignment
%defined by
%\[ \zeta [x_1 \mapsto a_1 ,\dots, x_n \mapsto a_n](x) = \begin{cases} a_i &\text{if } x = x_i (i=1,\dots,n) \\ \zeta(x) &\text{else} \end{cases} \]

%We use $[\bar{x}\mapsto\bar{a}]$ as a shorthand for
%$[x_1\mapsto a_1,\dots,x_n\mapsto a_n]$ and $val_{S,\zeta}(t)$ to
%denote the interpretation of a term $t$ in a state $S$ under a
%variable assignment $\zeta$.  Likewise, we use $val_{S, \zeta}(\varphi)$ to denote the truth value of a first-order logic
%formula $\varphi$ in a state $S$ under a variable assignment $\zeta$.

\subsection{Update Sets and Multisets}\label{usetsandmultisets}
In the ASM literature \cite{boerger:2003}, locations, updates, update sets and update multisets are the key concepts used to formalise the dynamics of computations. Thus, similar to ASMs \cite{boerger:2003}, DB-ASMs can be understood as an extension of finite state machines which proceed by transitions from states to successor states through updates. In a state of a DB-ASM, updatable dynamic functions are called locations, which are distinguished from static functions that cannot be updated. Informally, such locations represent the abstract concept of basic object containers, such as memory units. During each transition step, a DB-ASM produces a set or multiset of updates that are used to change location contents in a state, and location contents in a DB-ASM can only be changed by such updates.

Let $S$ be a state over $\Upsilon$, $f\in\Upsilon$ be a dynamic function symbol of
arity $n$ and $a_1,...,a_n$ be elements in $B_{db}$ or $B_a$ depending on whether $f \in \Upsilon_{db} \cup {\cal F}_b$ or $f \in \Upsilon_{a}$, respectively.
Then $(f,(a_1,...,a_n))$ is called a {\em location} of $S$. An
\emph{update} of $S$ is a pair $(\ell,b)$, where $\ell$ is a
location and $b \in B_{db}$ or $b \in B_a$, depending on whether $f \in \Upsilon_{db}$ or $f \in \Upsilon_{a} \cup {\cal F}_b$, respectively, is the \emph{update value} of $\ell$. To simplify notations, we write $(f
,(a_1,\dots,a_n),b)$ for the update $(\ell,b)$ with the location $\ell = (f,(a_1,\dots,a_n))$. The
interpretation of $\ell$ in $S$ is called the \emph{content} of
$\ell$ in $S$, denoted by $val_{S}(\ell)$. An \emph{update set}
$U$ is a set of updates; an \emph{update multiset}
$\ddot{U}$ is a multiset of updates. A \emph{location operator}
$\rho$ is a multiset function that returns a
single value from a multiset of values, e.g. \textsc{Average},
\textsc{Count}, \textsc{Sum}, \textsc{Max} and \textsc{Min} used in
SQL. An update set $U$ is called \emph{consistent} if it does
not contain conflicting updates, i.e., for all $(\ell,b),
(\ell,b^\prime) \in U$ we have $b = b^\prime$. Likewise, we say that an update multiset $\ddot{U}$ is \emph{consistent} if its corresponding update set $U$, obtained by setting the multiplicity of each element (update) in $\ddot{U}$ to $1$, is a consistent update set. Otherwise, we say that $\ddot{U}$ is \emph{inconsistent}.
If $U$ is a consistent update set, then
there exists a unique state $S + U$ resulting from
updating $S$ with $U$. We have
\[ val_{S + U}(\ell) \quad = \quad \begin{cases}
b &\text{if } (\ell,b) \in U \\
val_S(\ell) &\text{otherwise} \end{cases} \]
If $U$ is not consistent, then $S+U$ is undefined.

To illustrate the concepts of location operator, update sets and update multisets, we provide the following example in which parallel computations are synchronised by using a let rule.

\begin{example}\label{exa:DBASM-let}

Consider the relation \textsc{Route} in Fig.~\ref{fig:DBASM-Route} and the two DB-ASMs presented in Fig.~\ref{fig:tworules}.

\begin{figure}[!htbp]
{\centering
\begin{tabular}{|c|c|c|}
  \hline
  FromCid & ToCid & Distance \\\hline\hline
  $c_1$ & $c_2$ & $d_1$ \\
  $c_1$ & $c_5$ & $d_4$ \\
  $c_2$ & $c_3$ & $d_2$ \\
  $c_1$ & $c_4$ & $d_1$ \\
  \hline
\end{tabular}\caption{A relation
\textsc{Route}}\label{fig:DBASM-Route}}
\end{figure}

\begin{figure}[!htbp]
 \hspace{3cm}\textbf{let} $(\ell_{num},())
\!\rightharpoonup\! \textsc{sum}$ \textbf{in}

\hspace{3.5cm}\textbf{forall} $x_1,x_2$ \textbf{with}
$\exists x_3 (\textsc{Route}(x_1,x_2,x_3))$ \textbf{do}

\hspace{4cm} $\ell_{num}:= 1$

\hspace{3.5cm}\textbf{enddo}

\hspace{3cm}\textbf{endlet}

\smallskip
\hspace{4cm}(a) First DB-ASM

\bigskip
\hspace{3cm}\textbf{forall} $x_1,x_2$ \textbf{with}
$\exists x_3 (\textsc{Route}(x_1,x_2,x_3))$ \textbf{do}

\hspace{3.5cm} $\ell_{num}:= 1$

\hspace{3cm}\textbf{enddo}

\smallskip
\hspace{4cm}(b) Second DB-ASM

\caption{Two DB-ASMs}
\label{fig:tworules}
\end{figure}

The first DB-ASM in Fig. \ref{fig:tworules}.(a) computes the total number of
routes in the relation \textsc{Route}. Here \textsc{sum} is a
location operator assigned to the location $(\ell_{num},())$. In a state containing the
relation \textsc{Route} in Fig.\ref{fig:DBASM-Route}, 
the forall sub-rule yields the update multiset $\{\!\!\{(\ell_{num},(),1), (\ell_{num},(),1),
(\ell_{num},(),1),(\ell_{num},(),1)\}\!\!\}$ and the update set $\{(\ell_{num},(),1)\}$. In turn the let rule (and thus the DB-ASM) yields the corresponding update set $\{(\ell_{num},(),4)\}$, which results from the aggregation produced by the location operator \textsc{sum} of the four updates to the location $(\ell_{num},())$ that appear in the multiset produced by the forall rule. Since the second DB-ASM in Fig. \ref{fig:tworules}.(b) has no location operator associated with the location $(\ell_{num}, ())$ and  the forall rule yields the same update multiset and update set as before, this second DB-ASM produces the update set $\{(\ell_{num},(),1)\}$ instead.
\end{example}

\subsection{Semantics of Rules}
The semantics of DB-ASM rules is defined in terms of update
multisets and update sets. More specifically, each DB-ASM rule is
associated with a set of update multisets, which then ``collapses'' to a set of
update sets. Thus, if $r$ is a DB-ASM rule of
signature $\Upsilon$ and $S$ is a state of $\Upsilon$, then we
associate a set $\Delta(r,S,\zeta)$ of update sets and a set
$\ddot{\Delta}(r,S,\zeta)$ of update multisets with $r$ and $S$,
respectively, where $\zeta$ is a variable assignment.

Let $\zeta[x \mapsto a]$ denote the variable assignment which coincides with $\zeta$ except that it assigns the value $a$ to $x$.
We formally define the sets of update sets and sets of update
multisets yielded by DB-ASM rules in Fig.~\ref{fig:set} and Fig.~\ref{fig:multiset}, respectively. Assignment rules create
updates in update sets and multisets.  Choice rules introduce
non-determinism. Each choice rule generates a \emph{set of update sets} and a corresponding set of update multisets which contain all the different update sets and multisets, respectively, corresponding to all possible choices. Let rules aggregate updates to the same location
into a single update by means of location operators. All other rules
only rearrange updates into different update sets and multisets.

\begin{figure}[!htb]
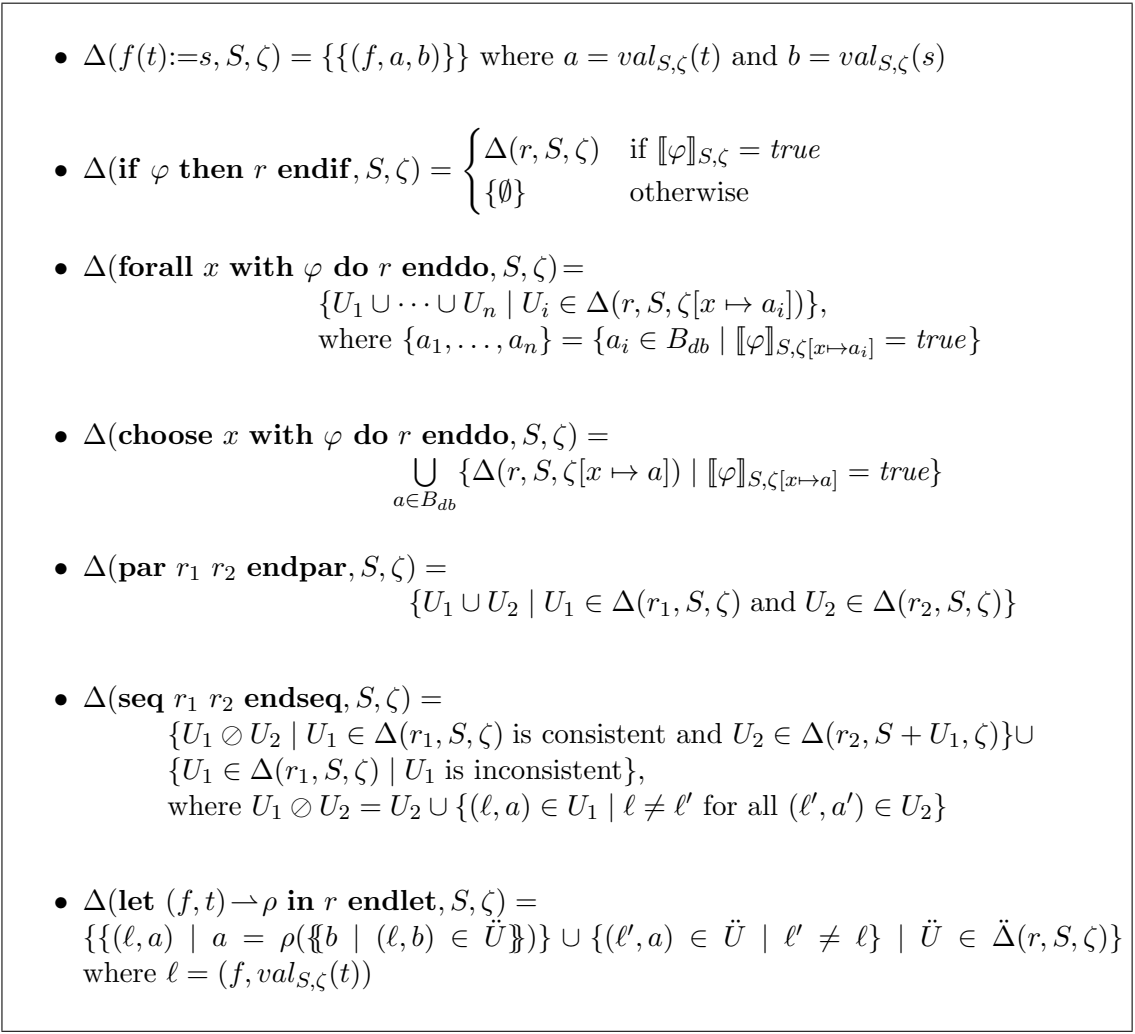

\fbox{\parbox{14.8cm}{
\begin{itemize}

\item  $\Delta(f(t) \text{:=} s,S,\zeta) = \{ \{  (f, a,b) \} \}$ where $a = val_{S,\zeta}(t)$ and $b = val_{S,\zeta}(s)$\\

\item
$\Delta(\text{\textbf{if} }\varphi\text{ \textbf{then} } r\text{
\textbf{endif}},S,\zeta) = \begin{cases} \Delta(r,S,\zeta) &\text{if
}  [\![\varphi]\!]_{S,\zeta} = \mathit{true}
\\ \{\emptyset\} &\text{otherwise}
\end{cases}$\\

\item
$\Delta(\text{\textbf{forall} } x \text{ \textbf{with} }\varphi \text{\textbf{ do} }r\text{ \textbf{enddo}},S,\zeta) \!=$ \\
\hspace*{3cm} $\{ U_1 \cup \dots \cup U_n \mid U_i \in \Delta(r,S,\zeta[x \mapsto a_i]) \}$, \\
\hspace*{3cm} where $\{a_1 ,\dots, a_n \} = \{a_i \in B_{db} \mid [\![\varphi]\!]_{S,\zeta[x \mapsto a_i]} = \mathit{true} \}$\\

\item
$\Delta(\text{\textbf{choose} } x \text{ \textbf{with} }\varphi \text{\textbf{ do} }r\text{ \textbf{enddo}},S,\zeta) =$\\
\hspace*{4cm} $\bigcup\limits_{a \in B_{db}}\{ \Delta(r,S,\zeta[x \mapsto a]) \mid [\![\varphi]\!]_{S,\zeta[x \mapsto a]} = \textit{true} \}$ \\

\item
$\Delta(\text{\textbf{par} }r_1 \text{ } r_2\text{ \textbf{endpar}},S,\zeta) =$\\
\hspace*{4.2cm} $\{ U_1 \cup U_2 \mid U_1 \in \Delta(r_1,S,\zeta) \; \text{and} \; U_2 \in \Delta(r_2,S,\zeta) \}$ \\

\item
$\Delta(\text{\textbf{seq} }r_1 \text{ } r_2 \text{
\textbf{endseq}},S,\zeta) =$ \\
\hspace*{1cm} $\{ U_1 \oslash U_2 \mid
U_1 \in \Delta(r_1,S,\zeta) \;\text{is consistent and }U_2 \in
\Delta(r_2,S+U_1,\zeta)\} \cup$ \\
\hspace*{1cm} $\{ U_1 \in \Delta(r_1,S,\zeta) \mid U_1 \;\text{is inconsistent} \}$,\\
\hspace*{1cm} where $U_1 \oslash U_2 = U_2 \cup \{
(\ell,a) \in U_1 \mid  \ell \neq \ell^\prime \text{ for all }(\ell',a^\prime) \in U_2\}$\\

\item
$\Delta(\text{\textbf{let} } (f,t) \!\rightharpoonup\!\rho \text{
\textbf{in} }r\text{ \textbf{endlet}},S,\zeta) =$\\
$\{\{ (\ell,a)
\mid a = \rho( \{\!\!\{ b \mid (\ell,b) \in
\ddot{U}\}\!\!\} ) \} \cup \{ (\ell^\prime,a) \in \ddot{U} \mid \ell^\prime \neq \ell \} \mid \ddot{U}\in \ddot{\Delta}(r,S,\zeta)\}$
\hspace{6.3cm}where $\ell = (f, val_{S,\zeta}(t))$

\end{itemize}}}\caption{Update sets of DB-ASM rules}\label{fig:set}
\end{figure}

\begin{figure}[h!]
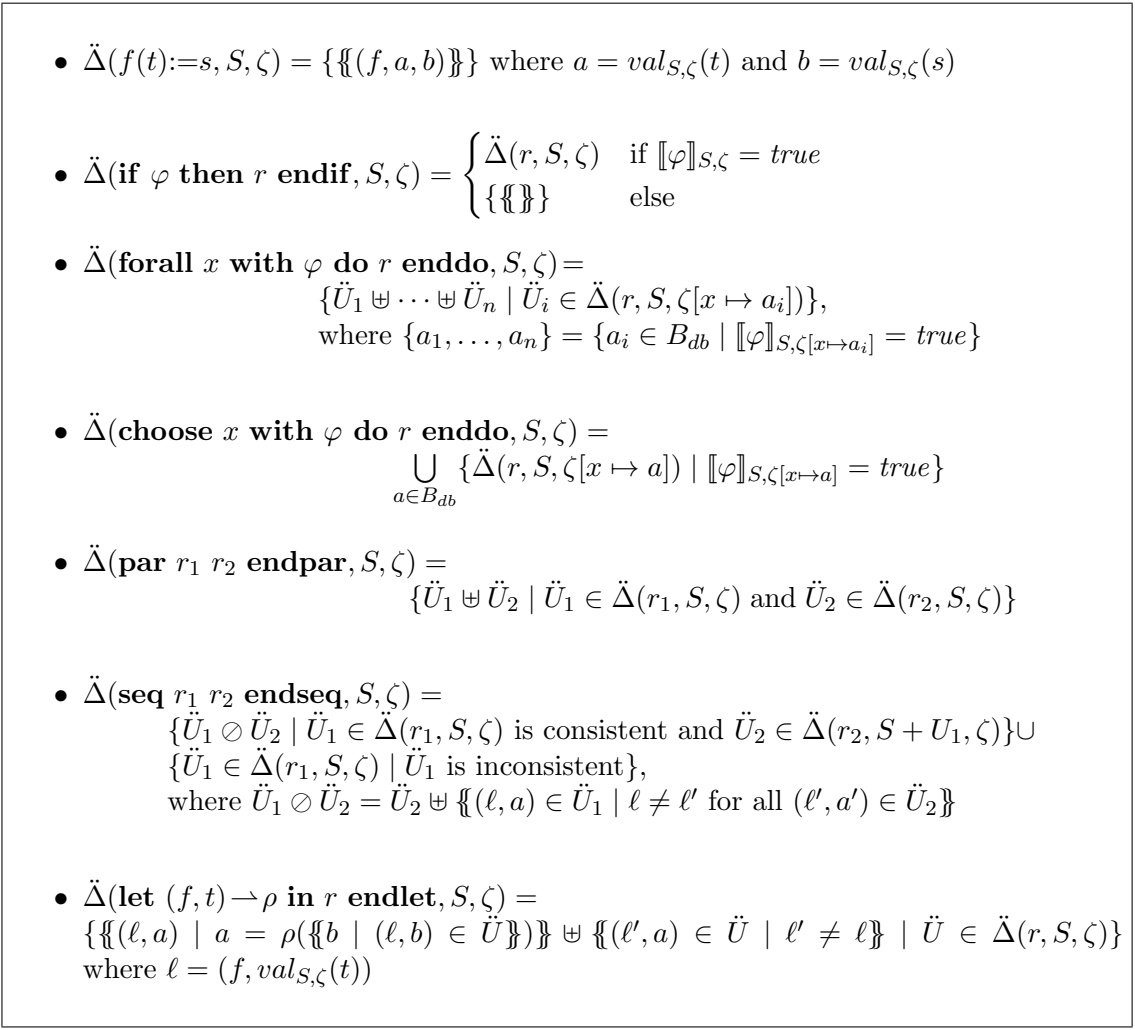

\fbox{\parbox{14.8cm}{
\begin{itemize}

\item  $\ddot{\Delta}(f(t) \text{:=} s,S,\zeta) = \{ \{\!\!\{ (f,a,b)
\}\!\!\} \}$ where $a = val_{S,\zeta}(t)$ and $b = val_{S,\zeta}(s)$\\

\item $\ddot{\Delta}(\text{\textbf{if} }\varphi\text{ \textbf{then} } r\text{ \textbf{endif}},S,\zeta) = \begin{cases}
\ddot{\Delta}(r,S,\zeta) &\text{if } [\![\varphi]\!]_{S,\zeta} = \mathit{true} \\
\{ \{\!\!\{ \}\!\!\} \}&\text{else}
\end{cases}$

\item
$\ddot{\Delta}(\text{\textbf{forall} }x \text{ \textbf{with} }\varphi \text{\textbf{ do} }r \text{ \textbf{enddo}},S,\zeta)\! =$\\
\hspace*{3cm} $\{ \ddot{U}_1 \uplus \dots \uplus \ddot{U}_n \mid \ddot{U}_i \in \ddot{\Delta}(r,S,\zeta[x \mapsto a_i]) \}$, \\
\hspace*{3cm} where $\{a_1, \ldots, a_n\} = \{ a_i \in B_{db} \mid [\![\varphi]\!]_{S,\zeta[x \mapsto a_i]} = \mathit{true} \}$\\

\item
$\ddot{\Delta}(\text{\textbf{choose} }x \text{ \textbf{with} }\varphi \text{\textbf{ do} }r\text{ \textbf{enddo}},S,\zeta) =$\\
\hspace*{4cm} $\bigcup\limits_{a \in B_{db}} \{\ddot{\Delta}(r,S,\zeta[x \mapsto a]) \mid [\![\varphi]\!]_{S,\zeta[x\mapsto a]} = \mathit{true}  \} $\\

\item
$\ddot{\Delta}(\text{\textbf{par} }r_1 \text{ } r_2\text{
\textbf{endpar}},S,\zeta) =$\\
\hspace*{4.2cm} $\{ \ddot{U}_1 \uplus\ddot{U}_2
\mid \ddot{U}_1 \in \ddot{\Delta}(r_1,S,\zeta) \; \text{and} \; \ddot{U}_2 \in
\ddot{\Delta}(r_2,S,\zeta)\}$\\

\item
$\ddot{\Delta}(\text{\textbf{seq} }r_1 \text{ } r_2 \text{
\textbf{endseq}},S,\zeta) =$\\
\hspace*{1cm} $\{ \ddot{U}_1 \oslash
\ddot{U}_2 \mid \ddot{U}_1 \in \ddot{\Delta}(r_1,S,\zeta)
\text{ is consistent and }\ddot{U}_2 \in
\ddot{\Delta}(r_2,S+U_1,\zeta)\} \cup$ \\
\hspace*{1cm} $\{\ddot{U}_1 \in \ddot{\Delta}(r_1,S,\zeta) \mid \ddot{U}_1 \text{ is inconsistent} \}$, \\
\hspace*{1cm} where $\ddot{U}_1 \oslash \ddot{U}_2 =
\ddot{U}_2 \uplus \{\!\!\{ (\ell,a) \in \ddot{U}_1 \mid \ell \neq \ell^\prime \text{ for all } (\ell^\prime,a^\prime) \in \ddot{U}_2 \}\!\!\}$\\
%\hspace*{1cm} and $U_1=\{(\ell,a)|(\ell,a)\in\ddot{U}_1\}$\\

\item
$\ddot{\Delta}(\text{\textbf{let} }(f,t)\!\rightharpoonup\!\rho \text{
\textbf{in} }r\text{ \textbf{endlet}},S,\zeta) =$\\
$\{ \{\!\!\{ (\ell,a)
\mid a = \rho( \{\!\!\{ b \mid (\ell,b) \in
\ddot{U}\}\!\!\} ) \}\!\!\} \uplus
\{\!\!\{ (\ell^\prime,a) \in \ddot{U} \mid \ell^\prime \neq \ell \}\!\!\} \mid \ddot{U}\in \ddot{\Delta}(r,S,\zeta)\}$
\hspace{6.3cm}where $\ell = (f, val_{S,\zeta}(t))$

\end{itemize}}}\caption{Update multisets of DB-ASM rules}\label{fig:multiset}
\end{figure}

%From our definitions, we have the following lemma.

\begin{lemma}\label{lem-finiteness}

For each state $S$, each DB-ASM rule $r$ and each variable assignment $\zeta$ from ${\cal X}_{db}$ to the base set $B_{db}$ of the database part of $S$, the following holds:

\begin{enumerate}

\item $\Delta(r,S,\zeta)$ and $\ddot{\Delta}(r,S,\zeta)$ are finite sets.

\item Each $U \in \Delta(r,S,\zeta)$ is a finite update set.

\item Each $\ddot{U} \in \ddot{\Delta}(r,S,\zeta)$ is a finite update multiset.

\end{enumerate}

\end{lemma}

\begin{proof}(Sketch). We use structural induction on $r$. The case of the assignment rule is obvious, as a single update will be created.

The conditional rule either produces exactly the same update sets and multisets as before or a single empty update set and multiset, respectively. For the forall rule the set $V = \{ a \in B_{db}\mid [\![\varphi]\!]_{S,\zeta[x \mapsto a]} = \mathit{true} \}$ is finite, because $x$ ranges over the finite set $B_{db}$. The stated finiteness then follows by induction, as $\Delta(r,S,\zeta)$ and $\ddot{\Delta}(r,S,\zeta)$ are finite sets, all $U \in \Delta(r,S,\zeta)$ and all $\ddot{U} \in \ddot{\Delta}(r,S,\zeta)$ are finite, and the new update sets and update multisets are built by set and multiset unions, respectively, that range over the finite set $V$.

For all other rules, the individual update sets and multisets are built by $\cup$, $\uplus$, $\oslash$ and aggregation with location operators applied to finite update sets and multisets, which gives the statements 2 and 3. Furthermore, the sets of update sets and update multisets, respectively, are built by comprehensions that range over finite sets. Hence they are finite as well, which gives statement 1 and completes the proof.

\end{proof}

A {\em Database Abstract State Machine} (DB-ASM) $M$ over signature
$\Upsilon$ consists of

\begin{itemize}

\item a set $\mathcal{S}$ of states over $\Upsilon$, non-empty subsets $\mathcal{S}_I \subseteq \mathcal{S}$ of initial states and $\mathcal{S}_F \subseteq \mathcal{S}$ of final states,

\item a closed DB-ASM rule $r$ over $\Upsilon$, and

\item a binary successor \emph{relation} $\delta$ over $\mathcal{S}$ determined by $r$, i.e.
\[ \delta=\{(S,S+U)|U\in \Delta(r,S) \text{ consistent} \} , \]
%\[ \{ S_{i+1} \mid (S_i,S_{i+1}) \in \delta \} = \{ S_i + U \mid U \in \Delta(r,S_i) \} \]

\end{itemize}
where the set $\Delta(r,S)$ ($\zeta$ is omitted from
$\Delta(r,S,\zeta)$ since $r$ is closed) of update sets yielded by
rule $r$ over the state $S$ defines the successor relation $\delta$ of $M$. A {\em
run} of $M$ is a finite sequence $S_0 ,\dots, S_n$ of states with
$S_0 \in \mathcal{S}_I$, $S_n \in \mathcal{S}_F$, $S_i \notin
\mathcal{S}_F$ for $0 < i < n$, and $(S_i,S_{i+1}) \in \delta$ for
all $i=0,\dots,n-1$.

\section{A Logic for DB-ASMs}\label{sec:dtc}

In this section we introduce a logic for DB-ASMs. This logic for DB-ASMs, which we denote as ${\cal L}^{db}$, is built as an extension of the logic ${\cal L}^{FO}$ of meta-finite structures used in the formalization of DB-ASMs (see Section~\ref{sec:states}). 

We start with an informal introduction which highlights the main characteristics of ${\cal L}^{db}$ and provides some illustrative examples. Then we proceed to introduce its formal syntax and semantics. 

Same as in the logic of meta-finite structures ${\cal L}^{FO}$, in ${\cal L}^{db}$ we distinguish between \emph{database terms} which are interpreted in the finite primary part of the states of DB-ASMs, and \emph{algorithmic terms} which are interpreted in the possible infinite secondary part. The set ${\cal T}_{a}$ of algorithmic terms needs however to be extended with first-order variables which range over the secondary part of the state and with a new kind of term (the $\rho$-terms). 

\emph{$\rho$-terms} are terms of the form $\rho_{v}(t|\varphi)$ where $\rho$ is a multiset operator, $t$ is a term in ${\cal T}_a$, $v$ is a variable which ranges either over the primary or secondary part of the state, and $\varphi$ is a ${\cal L}^{db}$ formula.  They are interpreted by the value resulting of applying the multiset operator $\rho$ to the multiset resulting of collecting the values of $t$ under all valuations that satisfy $\varphi$. The need for $\rho$-terms arises from the fact that DB-ASMs are able to collect updates yielded in parallel computations under the multiset semantics, i.e., update multisets, and then aggregate updates in an update multiset to an update set by using location operators.  

\begin{example}\label{exa:logic-syntax-term}
Consider the relation \textsc{Route} in
Fig.~\ref{fig:initialstate}. The following aggregate
queries are expressible by means of $\rho$-terms.

\begin{itemize}
  \item Q$_1$: Calculate the total number of direct routes.

\begin{center}
    $\textsc{Count}_{x}(1 \mid \exists y z (\textsc{Route}(x,y,z)))$
\end{center}

In an SQL database, Q$_1$ can be expressed by the following SQL
statement:\smallskip

    \begin{center}
    SELECT $count$(*) FROM \textsc{Route}
    \end{center}

  \item Q$_2$: Find the maximum number of direct connections of any city in the database.

\begin{center}
    $\textsc{Max}_{x} (\textsc{Count}_{y'}(1 \mid \exists z (\textsc{Route}(x,y',z))) \mid \exists y z (\textsc{Route}(x,y,z)))$
\end{center}

In a similar way, Q$_2$ can be expressed by the following SQL
statement:

\bigskip
          \hspace{1.5cm}SELECT $max$(NumofConnections)

          \hspace{1.5cm}FROM (SELECT Cid, $count$(*) as NumofConnections

                \hspace{2.85cm}FROM \textsc{Route}

                \hspace{2.85cm}GROUP BY Cid)

\end{itemize}

\end{example}

Due to the importance of non-determinism for enhancing the expressive power of database transformations, DB-ASMs include a non-deterministic choice rule. Consequently, it is no longer enough to consider the individual updates associated to the \emph{unique} update set produced by a rule of a deterministic ASM, as it is the case in the logic for ASMs~\cite{RobertLogicASM} of Nanchen and St\"ark. Instead, the logic ${\cal L}^{db}$ needs to be able to describe properties of the different update sets (and multisets) which can be associated to a given DB-ASM rule. That is, the logic ${\cal L}^{db}$ should allow us to handle multiple update sets, since a non-deterministic DB-ASM rule can produce a possible different update set for each of the possible choices.

 A natural and concise way of handling update sets (and multisets) is by means of second-order variables and second-order quantification. We therefore include both in the language of ${\cal L}^{db}$, albeit with a Henkin's semantics instead of the standard Tarski's semantics, so that we can avoid the well known incompleteness result of second-order logic.  
For the same reason we additionally include the multi-modal operator $[X]$, where $X$ is a second-order variable of arity $3$. The intended meaning of a formula $[X]\varphi$ is that $\varphi$ is true in the state obtained by applying the updates in $X$ to the current state. In turn, to express that $X$ is the update set produced by a rule $r$, we also include atomic formulae of the form $\textrm{upd}(r, X)$ to the language of ${\cal L}^{db}$.  

In order to encode update sets into second-order variables, we need to make some assumptions more precise.

\begin{definition}\label{extendedState}
Given a DB-ASM of some schema $\Upsilon$, we extend the sub-schema $\Upsilon_a$ of the algorithmic part with a new nullary and static function symbol (constant) $c_{f_i}$ for each dynamic function $f_i \in \Upsilon$. We assume that in every state $S$, these new constant symbols are interpreted by arbitrary, but pairwise different values. That is, if $c_{f_i}$ and $c_{f_j}$ are among the new constant symbols, then $c^S_{f_i} \neq c^S_{f_j}$.    

\smallskip

Let $S$ be a state of this extended signature $\Upsilon$, let $\zeta$ be a variable assignment into $S$, let $X$ be a second-order variable of arity $3$ and let $U = \{(f_i, a_1, a_2) \mid (a_0, a_1, a_2) \in \zeta(X) \, \text{and} \, a_0 = c_{f_i}^S\}$. We say $\zeta(X)$ \emph{represents} $U$ if $U$ constitutes an update set for $S$ and $(a_0, a_1, a_2) \in \zeta(X)$ iff $a_0 = c^S_f$ for some dynamic function $f \in \Upsilon$ and $(f, a_1, a_2) \in U$. 
\end{definition}

As noted earlier, the multiset semantics allows DB-ASMs to collect updates yielded in parallel computations, i.e., update multisets. This multiset semantics is handled via the inclusion of atomic formulae of the form $\mathrm{upm}(r,X)$. In this case, the intended meaning is that $\mathrm{upm}(r,X)$ is true if $X$ is a second-order variable of arity $4$ which represents an update multiset yielded by the rule $r$. We say that $X$ represents an update multiset $\ddot{U}$ iff for every update $(f, a_0, a_1) \in \ddot{U}$ with multiplicity $n > 0$ there are \emph{exactly} $n$ distinct $b_1, \ldots, b_n$ such that $(f,a_0,a_1,b_i) \in X$ and vice versa. 

\begin{example}\label{exa:syntaxformulae}
Consider Example~\ref{exa:DBASM-db} and the corresponding DB-ASM depicted in Fig.~\ref{exa:DBASM-code}. Let $r$ denote the main rule of the DB-ASM and $S$ denote one of its states.
The following formulae illustrate how the logic ${\cal L}^{db}$ can be used to express desirable properties of this DB-ASM. 

\begin{description}

\item [1.] If the rule $r$ over $S$ yields an update set $U$ containing an update $(\textsc{Visited}, x, \textsc{True})$, then for every neighbour city $y$ of $x$, the (current) shortest distance in state $S+U$ (calculated by the algorithm) between $y$ and $c$ is no longer \textsc{Infinity}. Representing $U$ by the second-order variable $X$, we obtain:

$\exists X x \Big(\mathrm{upd}(r,X)\wedge X(c_{\textsc{Visited}}, x, \textsc{True}) \rightarrow$ \\
\hspace*{5cm}$[X] \forall y z (\textsc{Route}(x,y,z) \rightarrow \textsc{Dist}(y) \neq \textsc{Infinity})\Big)$\\[0.1cm]

\item [2.] If in the current state $S$, the distance between a non-visited (by the algorithm) city $x$ and $c$ is minimal among the non-visited cities, then there is an update set $U$ yielded by the rule $r$ in state $S$ which updates the status of $x$ to visited. Representing $U$ by the second-order variable $X$, we obtain:

$\exists x \Big(\neg \textsc{Visited}(x) \wedge \forall y \big(\neg \textsc{Visited}(y) \rightarrow \textsc{Dist}(x) \leq \textsc{Dist}(y)\big) \rightarrow$ \\
\hspace*{7.5cm}$\exists X (\mathrm{upd}(r, X) \wedge [X] \textsc{Visited}(x))\Big)$\\[0.1cm]

\item [3.] If the current state $S$ is not an initial nor a final state and $U$ is an update set yielded by the rule $r$ in $S$, then the value of $\textsc{MDist}$ in the successor state $S+U$ equals the distance between $c$ and the closest unvisited (by the algorithm in state $S$) city. Representing $U$ by the second-order variable $X$ and using a $\rho$-term with location operator $\textsc{Min}$, we obtain:

$\forall X \Big(\neg \textsc{Initial} \wedge \exists  x y (\textsc{City}(x, y) \wedge \neg \textsc{Visited}(x)) \wedge \mathrm{upd}(r,X) \rightarrow$\\
\hspace*{2.8cm} $[X] \textsc{MDist} = \textsc{Min}_x\big(\textsc{Dist}(x) \mid \exists y (\textsc{City}(x, y) \wedge \neg \textsc{Visited}(x))\big)\Big)$\\[0.1cm]

\item [4.] If the current state $S$ is not an initial nor a final state and $U$ is an update set yielded by the rule $r$ in $S$, then every update multiset $\ddot{U}$ yielded by $r$ in $S$ contains at least one update $(\textsc{MDist}, (), a_i)$ such that $a_i$ coincides with the value stored in the location $(\textsc{MDist},())$ in the successor state $S+U$ and $a_i \leq a_j$ for every update $(\textsc{MDist},(), a_j)$ in $\ddot{U}$. Representing $U$ and $\ddot{U}$ by the second-order variables $X$ and $Y$, respectively, we obtain:

$\forall X \Big(\neg \textsc{Initial} \wedge \exists  x y (\textsc{City}(x, y) \wedge \neg \textsc{Visited}(x)) \wedge \mathrm{upd}(r,X) \rightarrow$\\
\hspace*{2.8cm} $\forall Y \big(\mathrm{upm}(r,Y) \rightarrow \exists \mathtt{x} \mathtt{y} (Y(c_{\textsc{MDist}},(),\mathtt{x},\mathtt{y}) \wedge [X]\textsc{MDist} = \mathtt{x} \, \wedge$\\
\hspace*{6.3cm} $\forall \mathtt{x}'\mathtt{y}' (Y(c_{\textsc{MDist}},(),\mathtt{x}',\mathtt{y}') \rightarrow \mathtt{x} \leq \mathtt{x}' )) \big)\Big)$

\end{description}

\end{example}

\subsection{Syntax}\label{subsec-adtmlogic-syntax}

The set of database and algorithmic terms of the logic ${\cal L}^{db}$ is defined as follows. 

\begin{definition}\label{def-meta-finitelogic-term-formula}
Let $\Upsilon=\Upsilon_{db}\cup\Upsilon_a\cup\mathcal{F}_b$ be a
signature of meta-finite states and let $\Lambda$ denote a set of location operators. Fix a countable set $\mathcal{X} = \mathcal{X}_{db} \cup \mathcal{X}_a$ of first-order variables. Variables in $\mathcal{X}_{db}$, denoted with standard lowercase letters $x, y, z , \ldots$, range over the database part of meta-finite states (i.e., the finite base set $B_{db}$), whereas variables in $\mathcal{X}_a$, denoted with typewriter-style lowercase letters $\texttt{x}, \texttt{y}, \texttt{z}, \ldots$, range over the algorithmic part of meta-finite states (i.e. the possible infinite base set $B_a$). 
The set of terms of ${\cal L}^{db}$ is formed by the set $\mathcal{T}_{db}$ of \emph{database terms} and the set
$\mathcal{T}_{a}$ of \emph{algorithmic terms} as defined by the following rules:
\begin{itemize}
    \item If $x\in \mathcal{X}_{db}$, then $x$ is a database term in $\mathcal{T}_{db}$.
    \item If $\texttt{x}\in \mathcal{X}_{a}$, then $\texttt{x}$ is an algorithmic term in $\mathcal{T}_{a}$. 
    \item If $f \in \Upsilon_{db}$ and $t \in \mathcal{T}_{db}$, then $f(t)$ is a database term in $\mathcal{T}_{db}$.
    \item If $f\in \mathcal{F}_b$ and $t \in \mathcal{T}_{db}$, then $f(t)$ is an algorithmic term in $\mathcal{T}_{a}$.
    \item If $f\in \Upsilon_{a}$ and $t \in \mathcal{T}_{a}$, then $f(t)$ is an algorithmic term in $\mathcal{T}_{a}$.
    \item If $\rho$ is a location operator in $\Lambda$, $\varphi$ is ${\cal L}^{db}$-formula (as in Definition~\ref{def:logicsyntaxformulae} below), $t\in \mathcal{T}_{a}$ and $v \in \mathcal{X}_{db} \cup \mathcal{X}_{a}$, then $\rho_{v}(t \mid \varphi)$ is an algorithmic term in $\mathcal{T}_{a}$.
    \item Nothing else is a term in $\mathcal{T}_{db}$ or $\mathcal{T}_{a}$.  
    \end{itemize}
\end{definition}

A \emph{pure term} is defined as a term that is \emph{not} a $\rho$-term and does \emph{not} contain any sub-term which is a $\rho$-term. 

Next, we formally introduce the set of well formed formulae of ${\cal L}^{db}$. 

\begin{definition}\label{def:logicsyntaxformulae}
Let $\Upsilon=\Upsilon_{db}\cup\Upsilon_a\cup\mathcal{F}_b$ be a
signature of meta-finite states and let $\Lambda$ denote a set of location operators. Let $\mathcal{T}_{db}$ and $\mathcal{T}_{a}$ be the corresponding set of database and algorithmic terms over $\Upsilon$ and $\Lambda$ (as per Definition~\ref{def-meta-finitelogic-term-formula}). Extend the set $\cal X$ of first-order variables with a countable set of second-order (relation) variables of arity $r$ for each $r \geq 1$. 
The following rules define the set of well formed formulae (wff) of ${\cal L}^{db}$.
\begin{enumerate}
\item If $s$ and $t$ are terms in $\mathcal{T}_{db}$, then $s=t$ is a wff.
\item If $s$ and $t$ are terms in $\mathcal{T}_{a}$, then $s=t$ is a wff.
\item If $t_1, \ldots, t_r$ are terms in $\mathcal{T}_{db} \cup \mathcal{T}_{a}$ and $X$ is a second-order variable of arity $r$, then $X(t_1, \ldots, t_r)$ is a wff.
\item If $r$ is a DB-ASM rule and $X$ is a second-order variable of arity $3$, then $\textrm{upd}(r, X)$ is a wff.
\item If $r$ is a DB-ASM rule and $X$ is a second-order variable of arity $4$, then $\textrm{upm}(r,X)$ is a wff.
\item If $\varphi$ is a wff, then $(\neg \varphi)$ is a wff.
\item If $\varphi$ and $\psi$ are wff's, then $(\varphi \vee \psi)$ is a wff.
\item If $\varphi$ is a wff and $x \in {\cal X}_{db}$, then $\forall x (\varphi)$ is a wff.
\item If $\varphi$ is a wff and $\mathtt{x} \in {\cal X}_{a}$, then $\forall \mathtt{x} (\varphi)$ is a wff.
\item If $\varphi$ is a wff and $X$ is a second-order variable, then $\forall X (\varphi)$ is a wff.
\item If $\varphi$ is a wff and $X$ is a second-order variable of arity $3$, then $([X]\varphi)$ is a wff.
\item Nothing else is a wff.  
\end{enumerate}
\end{definition}

Formulae of the form $\varphi \wedge \psi$, $\varphi \rightarrow \psi$, $\exists x (\varphi)$, $\exists \mathtt{x} (\varphi)$ and $\exists X (\varphi)$ are considered as abbreviations of $\neg (\neg \varphi \vee \neg \psi)$, $\neg \varphi \vee \psi$, $\neg \forall x (\neg \varphi)$, $\neg \forall \mathtt{x} (\neg \varphi)$ and $\neg \forall X (\neg \varphi)$, respectively. 

We say that a formula of ${\cal L}^{db}$ is a \emph{pure formula} if it can be defined using only the rules 1--2 and 6--9 in Definition~\ref{def:logicsyntaxformulae} and does \emph{not} contains any $\rho$-term. Notice that the formula occurring in the \textbf{if}, \textbf{forall} and \textbf{choose} rules of DB-ASMs satisfy this definition, i.e., they are pure formulae of ${\cal L}^{db}$. We also say that a term or formula of ${\cal L}^{db}$ is \emph{static} if it does \emph{not} contain any dynamic function symbol. Since static functions cannot be updated,  it is clear that the value of a static term as well as the truth value of a static formula cannot change during the run of a DB-ASM. 

The atomic formulae~$\textrm{upd}(r, X)$ and~$\textrm{upm}(r,X)$ are not strictly necessary. As shown latter (see Lemmas~\ref{lem-upd} and~\ref{lem-upm}), they can be eliminated from the language of ${\cal L}^{db}$ without affecting its expressive power.

\subsection{Semantics}\label{subsec-adtmlogic-semantics}

We use a semantics due to Henkin~\cite{henkin1950}, in which the interpretation of second-order quantifiers is part of the specification of a structure (state) rather than an invariant through all models as in the case of the standard Tarski's semantics. 

\begin{definition}\label{HenkinPreStructure}
Let $\Upsilon=\Upsilon_{db}\cup\Upsilon_a\cup\mathcal{F}_b$ be a signature of meta-finite states. A \emph{Henkin meta-finite $\Upsilon$-prestructure} $S$ is a meta-finite state of signature $\Upsilon$ and nonempty base set $B = B_{db} \cup B_a$, which is extended with a new universe $D_n$ of $n$-ary relations for each $n \geq 1$, where $D_n \subseteq {\cal P}(\underbrace{B \times \cdots \times B}_{n})$.  
\end{definition} 

Variable assignments into a Henkin meta-finite prestructure $S$ are defined as usual, except that we require that every assignment $\zeta$ satisfies the following conditions:
\begin{itemize}
\item If $x$ is a first-order variable in ${\cal X}_{db}$, then $\zeta(x) \in B_{db}$.
\item If $\mathtt{x}$ is a first-order variable in ${\cal X}_a$, then $\zeta(\mathtt{x}) \in B_a$.
\item If $X$ is a second-order variable of arity $n$, then $\zeta(X) \in D_n$.
\end{itemize}  

Given a variable assignment, terms of the logic ${\cal L}^{db}$ can be interpreted into a Henkin meta-finite prestructure. 

\begin{definition}\label{termsInterpretation}
Let $S$ be a Henkin meta-finite prestructure of signature $\Upsilon=\Upsilon_{db}\cup\Upsilon_a\cup\mathcal{F}_b$ and let $\zeta$ be a variable assignment into $S$. If $t$ is a term (either a database term or an algorithmic term), then the value (interpretation) of $t$ in $S$ under $\zeta$ (denoted $val_{S,\zeta}(t)$) is defined by the following rules:
\begin{itemize}
    \item If $t$ is a variable $x \in \mathcal{X}_{db}$ or $\texttt{x} \in \mathcal{X}_{a}$, then $val_{s,\zeta}(t) = \zeta(t)$.
    \item If $t$ is of the form $f(t')$ where $f \in \Upsilon$ and $t'$ is a term, then $val_{s,\zeta}(t) = f^S(val_{S,\zeta}(t'))$.
    \item If $t$ is of the form  $\rho_{v}(t' \mid \varphi)$, then \[val_{S,\zeta}(t)= \rho(\{\!\!\{val_{S,\zeta[v\mapsto a_i]}(t') \mid a_i \in D \text{ and } [\![\varphi]\!]_{S,\zeta[v \mapsto a_i]}=true \}\!\!\}),\] where $D = B_{db}$ or $D = B_{a}$ depending on whether $v \in {\cal X}_{db}$ or $v \in {\cal X}_{a}$, respectively.
\end{itemize}
\end{definition}

The interpretation of ${\cal L}^{db}$-formulae into Henkin meta-finite prestructures is defined as follows.  

\begin{definition}
Let $S$ be a Henkin meta-finite prestructure of signature $\Upsilon=\Upsilon_{db}\cup\Upsilon_a\cup\mathcal{F}_b$, extended as per Definition~\ref{extendedState} with a new and different constant symbol $c_{f_i}$ for each dynamic function symbol $f_i \in \Upsilon$. Let $\zeta$ be a variable assignment into $S$. For $X$ a second-order variable of arity $3$, we abuse the notation by writing $val_{S,\zeta}(X) \in \Delta(r,S,\zeta)$, meaning that there is a set $U \in \Delta(r,S,\zeta)$ such that $(f, a_0, a_1) \in U$ iff $(c_f^S, a_0, a_1) \in val_{S,\zeta}(X)$. Likewise, for $X$ a second-order variable of arity $4$, we write $val_{S,\zeta}(X) \in \ddot{\Delta}(r,S,\zeta)$, meaning that there is a multiset $\ddot{U} \in \Delta(r,S,\zeta)$ such that $(f, a_0, a_1) \in \ddot{U}$ with multiplicity $n$ iff there are exactly $b_1, \ldots, b_n$ pairwise different values such that $(c_f^S, a_0, a_1, b_i) \in val_{S,\zeta}(X)$ for every $1 \leq i \leq n$.

If $\varphi$ is an ${\cal L}^{db}$-formula, then the truth value of $\varphi$ on $S$ under $\zeta$ (denoted as $[\![\varphi]\!]_{S,\zeta}$) is either $\mathit{true}$ or $\mathit{false}$ and it is determined by the following rules:
\begin{itemize}
\item If $\varphi$ is of the form $s=t$, then $[\![\varphi]\!]_{S,\zeta} = \begin{cases} \mathit{true} &\text{if } val_{S,\zeta}(s)=val_{S,\zeta}(t);\\ \textit{false} & \text{otherwise.}\end{cases}$
\item If $\varphi$ is of the form $X(t_1, \ldots, t_r)$, then \\\hspace*{3.5cm}$[\![\varphi]\!]_{S,\zeta} = \begin{cases} \mathit{true} &\text{if } (val_{S,\zeta}(t_1), \ldots, val_{S,\zeta}(t_n)) \in val_{S, \zeta}(X);\\ \textit{false} & \text{otherwise.} \end{cases}$
\item If $\varphi$ is of the form $\textrm{upd}(r,X)$, then $[\![\varphi]\!]_{S,\zeta} = \begin{cases} \mathit{true} &\text{if } val_{S,\zeta}(X) \in \Delta(r,S,\zeta);\\ \textit{false} & \text{otherwise.} \end{cases}$
\item If $\varphi$ is of the form $\textrm{upm}(r, X)$, then $[\![\varphi]\!]_{S,\zeta} = \begin{cases} \mathit{true} &\text{if } val_{S,\zeta}(X) \in \ddot{\Delta}(r,S,\zeta);\\ \textit{false} & \text{otherwise.} \end{cases}$
\item If $\varphi$ is of the form $(\neg \psi)$, then $[\![\varphi]\!]_{S,\zeta} = \begin{cases} \mathit{true} &\text{if } [\![\psi]\!]_{S,\zeta}=\mathit{false};\\ \textit{false} & \text{otherwise.}\end{cases}$
\item If $\varphi$ is of the form $(\alpha \vee \psi)$, then $[\![\varphi]\!]_{S,\zeta} = \begin{cases} \mathit{true} &\text{if } [\![\alpha]\!]_{S,\zeta}=\mathit{true} \, \text{or} \, [\![\psi]\!]_{S,\zeta}=\mathit{true};\\ \textit{false} & \text{otherwise.}\end{cases}$
%\item If $\varphi$ is of the form $(\alpha \wedge \psi)$, then $[\![\varphi]\!]_{S,\zeta} = \begin{cases} \mathit{true} &\text{if } [\![\alpha]\!]_{S,\zeta}=\mathit{true} \, \text{and} \, [\![\psi]\!]_{S,\zeta}=\mathit{true};\\ \textit{false} & \text{otherwise.}\end{cases}$
%\item If $\varphi$ is of the form $(\alpha\rightarrow \psi)$, then $[\![\varphi]\!]_{S,\zeta} = \begin{cases} \mathit{true} &\text{if } [\![\alpha]\!]_{S,\zeta}=\mathit{false} \, \text{or} \, [\![\psi]\!]_{S,\zeta}=\mathit{true};\\ \textit{false} & \text{otherwise.}\end{cases}$ 
\item If $\varphi$ is of the form $\forall x (\psi)$, then $[\![\varphi]\!]_{S,\zeta} = \begin{cases} \mathit{true} &\text{if } [\![\psi]\!]_{S,\zeta[x\mapsto a]}=\mathit{true} \, \text{for all} \, a\in B_{db};\\ \textit{false} & \text{otherwise.}\end{cases}$  
%\item If $\varphi$ is of the form $\exists x (\psi)$, then $[\![\varphi]\!]_{S,\zeta} = \begin{cases} \mathit{true} &\text{if } [\![\psi]\!]_{S,\zeta[x\mapsto a]}=\mathit{true} \, \text{for some} \, a\in B_{db};\\ \textit{false} & \text{otherwise.}\end{cases}$  
\item If $\varphi$ is of the form $\forall \mathtt{x} (\psi)$, then $[\![\varphi]\!]_{S,\zeta} = \begin{cases} \mathit{true} &\text{if } [\![\psi]\!]_{S,\zeta[\mathtt{x}\mapsto a]}=\mathit{true} \, \text{for all} \, a\in B_{a};\\ \textit{false} & \text{otherwise.}\end{cases}$  
\item If $\varphi$ is of the form $\forall X (\psi)$, where $X$ is a second-order variable of arity $n$, then  $[\![\varphi]\!]_{S,\zeta} = \begin{cases} \mathit{true} &\text{if } [\![\psi]\!]_{S,\zeta[X \mapsto R]}=\mathit{true} \, \text{for all} \, R \in D_n;\\ \textit{false} & \text{otherwise.}\end{cases}$  
\item If $\varphi$ is of the form $([X]\psi)$, then \\ \hspace*{0.7cm}$[\![\varphi]\!]_{S,\zeta} = \begin{cases} \mathit{false} & \text{if} \, \zeta(X) \, \text{represents (as per Definition~\ref{extendedState}) an update set} \, U
\\ & \text{such that} \, U \, \text{is consistent} \, \text{and} \, [\![\psi]\!]_{S+U,\zeta}=\mathit{false};\\ \textit{true} & \text{otherwise.}\end{cases}$  
\end{itemize}
\end{definition}

\begin{remark}
Note that if $\varphi$ is of the from $([X]\psi)$, then $\varphi$ is interpreted as $\textit{true}$ in any of the following cases: 
\begin{itemize}
\item $\zeta(X)$ represents an update set $U$ which is inconsistent. 
\item $\zeta(X)$ does \emph{not} represents an update set.
\item $\zeta(X)$ represents a consistent update set $U$ and $\psi$ is interpreted as $\mathit{true}$ in $S+U$. 
\end{itemize}
\end{remark}

For a sentence $\varphi$ of ${\cal L}^{db}$ to be valid in the given Henkin's semantics, it must be true in all Henkin meta-finite prestructures. This is a stronger requirement than saying that $\varphi$ is valid in the standard Tarski's semantics. A sentence that is valid in the standard Tarski's semantics is true in those Henkin meta-finite prestructures for which each universe $D_n$ is interpreted as the set of all relations of arity $n$. But such a sentence $\varphi$ might turn out to be false in some Henkin meta-finite prestructure (i.e., $\neg \varphi$ might evaluate to true in some Henkin meta-finite prestructure).

Clearly, we do \emph{not} want the universes $D_n$ of the Henkin meta-finite prestructures to be any arbitrary collections of $n$-ary relations. It is then reasonable to restrict our attention to some collections of n-ary relations that we know about, because we can define them.  

\begin{definition}\label{HenkinStructure}
A \emph{Henkin meta-finite structure} for a second-order language is a Henkin meta-finite prestructure $S$ that is closed under definability, i.e., for every formula $\varphi$, variable assignement $\zeta$ and arity $n \geq 1$, we have that \[\{\bar{a} \in A^n \mid [\![\varphi]\!]_{S,\zeta[a_1 \mapsto x_1, \ldots, a_n \mapsto x_n]} = \textit{true}\} \in D_n.\]  
\end{definition}

In the following, we restrict our attention to Henkin meta-finite structures. Notice that, if $M$ is a DB-ASM of some vocabulary $\Upsilon$ of meta-finite structures, we can use ${\cal L}^{db}$ formulae of the vocabulary $\Upsilon$ (extended with constant symbols $c_{f_i}$ for each dynamic function symbol $f_i \in \Upsilon$ as per Definition~\ref{extendedState}) to express properties of $M$. We can then verify these properties by evaluating the formulae over appropriate Henkin meta-finite structures of the (extended) vocabulary $\Upsilon$, and use the complete proof system which we introduce next to derive logical consequences. 

\section{A Proof System} \label{sec:proof_system}

In this section we develop a proof system for the logic ${\cal L}^{db}$ for DB-ASMs.

\begin{definition}\label{def-implied-formula}
We say that a Henkin meta-finite structure $S$ is a \emph{model} of a formula $\varphi$ (denoted as $S \models \varphi$) iff $[\![\varphi]\!]_{S,\zeta}= \textit{true}$ holds for every variable assignment $\zeta$.  If $\Psi$ is a set of formulae, we say that $S$ \emph{models} $\Psi$ (denoted as $S \models \Psi$) iff $S \models \varphi$ for each $\varphi \in \Psi$.
A formula $\varphi$ is said to be a \emph{logical consequence} of a set $\Psi$ of formulae (denoted as $\Psi\models\varphi$) if for every Henkin meta-finite structure $S$, if $S \models \Psi$, then $S \models \varphi$.
A formula $\varphi$ is said to be \emph{valid} (denoted as $\models \varphi$) if $[\![\varphi]\!]_{S,\zeta}=true$ in every Henkin meta-finite structure $S$ with every variable assignment $\zeta$.
A formula $\varphi$ is said to be \emph{derivable} from a set $\Psi$ of formulae (denoted as $\Psi\vdash_{\mathfrak{R}}\varphi$) if there is a deduction from formulae in $\Psi$
to $\varphi$ by using a set $\mathfrak{R}$ of axioms and inference rules.

\end{definition}

We will define such a set $\mathfrak{R}$ of axioms and rules in Subsection \ref{sub:AxiomsRules}. Then we simply write $\vdash$ instead of $\vdash_{\mathfrak{R}}$. We also define equivalence between two DB-ASM rules. 

\begin{definition} \label{def-equivalent-rules}Let $r_1$ and $r_2$ be two DB-ASM rules. Then
$r_1$ and $r_2$ are \emph{equivalent} (denoted as $r_1\equiv r_2$)
if for every Henkin meta-finite structure $S$ it holds that \[S \models \forall X (\mathrm{upd}(r_1,X)\leftrightarrow \mathrm{upd}(r_2,X)).\]
\end{definition}

The substitution of a term $t$ for a variable $x$ in a formula
$\varphi$ (denoted as $\varphi[t/x]$) is defined by the rule of
substitution. That is, $\varphi[t/x]$ is the result of replacing all
free instances of $x$ by $t$ in $\varphi$ provided that no free
variable of $t$ becomes bound after substitution.

\subsection{Consistency}\label{sub:Consistency}

In \cite{RobertLogicASM} Nanchen and St\"ark use a predicate $\mathrm{Con}(r)$ as an
abbreviation for the statement that the rule $r$ is consistent. As a rule $r$ in their work is considered to be
deterministic, there is no ambiguity with the reference to the
update set associated with $r$, i.e., each deterministic rule $r$ generates exactly one (possibly empty) update set. Thus a deterministic rule $r$ is consistent iff the update set generated by $r$ is consistent.
However, in the logic for DB-ASMs, the presence of non-determinism
makes the situation less straightforward.

Instead, given a rule $r$ of a signature $\Upsilon = \Upsilon_{db} \cup \Upsilon_a \cup {\cal F}_b$ of a DB-ASM, we can use $\mathrm{con}(r,X)$ to expresses that $X$ represents one of the possible update sets generated by the rule $r$ (which in our setting can be non-deterministic) and that $X$ is consistent. This can be expressed in the logic ${\cal L}^{db}$ with the following formula:
\begin{equation}\label{conr}
\text{con}(r,X)\equiv\mathrm{upd}(r,X)\wedge\mathrm{conUSet}(X)
\end{equation}
where 
\begin{eqnarray}\label{con}
\mathrm{conUSet}(X)\equiv\bigwedge\limits_{c_f\in \mathcal{F}_{\mathit{dyn}} \wedge f \in \Upsilon_{db}} \forall x y z ((X(c_f,x,y) \wedge X(c_f, x, z)) \rightarrow y=z) \wedge \\
\bigwedge\limits_{c_f\in \mathcal{F}_{\mathit{dyn}} \wedge f \in \Upsilon_{a}} \forall \mathtt{x} \mathtt{y} \mathtt{z} ((X(c_f,\mathtt{x},\mathtt{y}) \wedge X(c_f, \mathtt{x}, \mathtt{z})) \rightarrow \mathtt{y}=\mathtt{z}) \wedge \nonumber\\
\bigwedge\limits_{c_f\in \mathcal{F}_{\mathit{dyn}} \wedge f \in {\cal F}_b} \forall x \mathtt{y} \mathtt{z} ((X(c_f,x,\mathtt{y}) \wedge X(c_f, x, \mathtt{z})) \rightarrow \mathtt{y}=\mathtt{z}) \nonumber
\end{eqnarray}
for $\mathcal{F}_{\mathit{dyn}}$ the set of constants representing the dynamic function symbols in $\Upsilon$ (see Definition~\ref{extendedState}).   

As the rule $r$ may be non-deterministic, it is possible that $r$
yields several update sets. Thus, we develop the consistency of DB-ASM rules in
two versions:

\begin{itemize}

\item A rule $r$ is \emph{weakly consistent} (denoted as
$\mathrm{wcon}(r)$) if at least one update set generated
by $r$ is consistent. This can be expressed as follows:
\begin{equation}\label{wcon}
  \text{wcon}(r)\equiv \exists X (\mathrm{con}(r,X))
\end{equation}

\item A rule $r$ is \emph{strongly consistent} (denoted as
scon$(r)$) if every update set generated by $r$
is consistent. This can be expressed as follows:
\begin{equation}\label{scon}
  \text{scon}(r)\equiv \forall X (\text{upd}(r,X)\rightarrow\mathrm{conUSet}(X))
\end{equation}

\end{itemize}

In the case that a rule $r$ is deterministic, the weak notion of
consistency coincides with the strong notion of consistency, i.e.,
$\mathrm{wcon}(r) \equiv \mathrm{scon}(r)$.
%Clearly, if a rule $r$ is not
%defined, then it is neither weakly nor strongly consistent.

\subsection{Update Sets}\label{sub:UpdateSets}

We present the axioms for the predicate $\mathrm{upd}(r,X)$ of the logic ${\cal L}^{db}$. Since a DB-ASM rules may be non-deterministic, a straightforward extension of the formalisation of $\mathrm{upd}$ for the forall and parallel rules used in the logic for ASMs~\cite{RobertLogicASM} is not sufficient in our case (cf. Axioms~\textbf{U3} and~\textbf{U4} below with the corresponding axions in~\cite{RobertLogicASM}). 

As before, we assume that if $f$ is a dynamic function symbol in the given signature of meta-finite states $\Upsilon = \Upsilon_{db} \cup \Upsilon_a \cup {\cal F}_b$, then there is a corresponding constant (static nullary function) symbol $c_f \in \Upsilon_a$ as per Definition~\ref{extendedState}. We use ${\cal F}_{\mathit{dyn}}$ to denote the set of all $c_f$ such that $f$ is a dynamic function symbol in $\Upsilon$. In the following, $S$ denotes an arbitrary Henkin structure of signature $\Upsilon$ and $B=B_{db}\cup B_a$ denotes the base set (domain) of the database and algorithmic parts of $S$. W.l.o.g. we further assume that $B_{db} \cap B_a = \emptyset$. 

Let $\zeta$ be a valuation into $S$. In the formulation of the axioms we use the predicate $\mathrm{isUSet}(X)$ to denote that $\zeta(X)$ represents an update set for $S$ (see Definition~\ref{extendedState}).  That is, for every triple $(a_1, a_2, a_3) \in \zeta(X)$, we have that $a_1 = c^S_f$ for some dynamic function $f \in \Upsilon$ and $a_2, a_3$ are values of the appropriate database or algorithmic base sets depending on whether $f$ is a database, algorithmic or bridge function symbol. 
\begin{align*}
\mathrm{isUSet}(X) \equiv& \forall \mathtt{x}_1 x_2 x_3 \Big(X(\mathtt{x}_1, x_2, x_3) \rightarrow \bigvee_{c_f \in {\cal F}_\mathit{dyn} \wedge f \in \Upsilon_{db}} \mathtt{x}_1 = c_f\Big) \wedge&\\
&\forall \mathtt{x}_1 \mathtt{x}_2 \mathtt{x}_3 \Big(X(\mathtt{x}_1, \mathtt{x}_2, \mathtt{x}_3) \rightarrow \bigvee_{c_f \in {\cal F}_\mathit{dyn} \wedge f \in \Upsilon_{a}} \mathtt{x}_1 = c_f\Big) \wedge&\\
&\forall \mathtt{x}_1 x_2 \mathtt{x}_3 \Big(X(\mathtt{x}_1, x_2, \mathtt{x}_3) \rightarrow \bigvee_{c_f \in {\cal F}_\mathit{dyn} \wedge f \in {\cal F}_{b}} \mathtt{x}_1 = c_f\Big) \wedge&\\
&\forall \mathtt{x}_1 \mathtt{x}_2 x_3( \neg X(\mathtt{x}_1, \mathtt{x}_2, x_3)) \wedge&\\
&\forall x_1 x_2 x_3( \neg X(x_1, x_2, x_3)) \wedge \forall x_1 \mathtt{x}_2 x_3( \neg X(x_1, \mathtt{x}_2, x_3)) \wedge&\\
&\forall x_1 x_2 \mathtt{x}_3( \neg X(x_1, x_2, \mathtt{x}_3)) \wedge \forall x_1 \mathtt{x}_2 \mathtt{x}_3( \neg X(x_1, \mathtt{x}_2, \mathtt{x}_3))&
\end{align*}

The axioms for $\mathrm{upd}(r,X)$ are as follows (cf. the definition of update sets in Fig.~\ref{fig:set}):
\begin{itemize}
\item Our first three axioms express that $X$ represents an update set yielded by the assignment rule $f(t) := s$ iff it contains exactly one update which is $(f,t,s)$. 
\begin{align*}
\textbf{U1.1: } & \text{If} \; f \; \text{is a database function symbol in} \; \Upsilon_{db} \; \text{then}&\\  
&\mathrm{upd}(f(t) := s, X) \leftrightarrow \mathrm{isUSet}(X) \wedge X(c_f, t, s) \wedge&\\
&\qquad \qquad \qquad \qquad \quad \forall \mathtt{x}_1 x_2 x_3 (X(\mathtt{x}_1,x_2,x_3) \rightarrow \mathtt{x}_1 = c_f \wedge x_2 = t \wedge x_3 = s) \wedge&\\
&\qquad \qquad \qquad \qquad \quad \forall \mathtt{x}_1 \mathtt{x}_2 \mathtt{x}_3 (\neg X(\mathtt{x}_1,\mathtt{x}_2,\mathtt{x}_3)) \wedge \forall \mathtt{x}_1 x_2 \mathtt{x}_3 (\neg X(\mathtt{x}_1,x_2,\mathtt{x}_3)) 
\end{align*}
\begin{align*}
\textbf{U1.2: } & \text{If} \; f \; \text{is an algorithmic function symbol in} \; \Upsilon_{a} \; \text{then}&\\  
&\mathrm{upd}(f(t) := s, X) \leftrightarrow \mathrm{isUSet}(X) \wedge X(c_f, t, s) \wedge &\\
&\qquad \qquad \qquad \qquad \quad \forall \mathtt{x}_1 \mathtt{x}_2 \mathtt{x}_3 (X(\mathtt{x}_1,\mathtt{x}_2,\mathtt{x}_3) \rightarrow \mathtt{x}_1 = c_f \wedge \mathtt{x}_2 = t \wedge \mathtt{x}_3 = s) \wedge&\\
&\qquad \qquad \qquad \qquad \quad \forall \mathtt{x}_1 x_2 x_3 (\neg X(\mathtt{x}_1,x_2,x_3)) \wedge \forall \mathtt{x}_1 x_2 \mathtt{x}_3 (\neg X(\mathtt{x}_1,x_2,\mathtt{x}_3)) 
\end{align*}
\begin{align*}
\textbf{U1.3: } & \text{If} \; f \; \text{is a bridge function symbol in} \; {\cal F}_b \; \text{then}&\\  
&\mathrm{upd}(f(t) := s, X) \leftrightarrow \mathrm{isUSet}(X) \wedge X(c_f, t, s) \wedge &\\
&\qquad \qquad \qquad \qquad \quad \forall \mathtt{x}_1 x_2 \mathtt{x}_3 (X(\mathtt{x}_1,x_2,\mathtt{x}_3) \rightarrow \mathtt{x}_1 = c_f \wedge x_2 = t \wedge \mathtt{x}_3 = s) \wedge&\\
&\qquad \qquad \qquad \qquad \quad \forall \mathtt{x}_1 x_2 x_3 (\neg X(\mathtt{x}_1,x_2,x_3)) \wedge \forall \mathtt{x}_1 \mathtt{x}_2 \mathtt{x}_3 (\neg X(\mathtt{x}_1,\mathtt{x}_2,\mathtt{x}_3))& 
\end{align*}
\item Axiom \textbf{U2} asserts that, if the formula $\varphi$ evaluates to $\mathit{true}$, then
$X$ is an update set yielded by the conditional rule \textbf{if} $\varphi$
\textbf{then} $r$ \textbf{endif} iff $X$ is an update set yielded by the rule
$r$. Otherwise, the conditional rule yields only an empty update
set.
\begin{flalign*}
\textbf{U2: } \mathrm{upd}(\textbf{if} \, \varphi \, \textbf{then}\, r \, \textbf{endif}, &X) \leftrightarrow (\varphi \wedge \mathrm{upd}(r,X)) \vee &\\
& \big(\neg \varphi \wedge \mathrm{isUSet}(X) \wedge \forall \mathtt{x}_1 x_2 x_3 (\neg X(\mathtt{x}_1,x_2,x_3)) \wedge &\\
&\quad \forall \mathtt{x}_1 \mathtt{x}_2 \mathtt{x}_3 (\neg X(\mathtt{x}_1,\mathtt{x}_2,\mathtt{x}_3)) \wedge \forall \mathtt{x}_1 x_2 \mathtt{x}_3 (\neg X(\mathtt{x}_1,x_2,\mathtt{x}_3))\big) &
\end{flalign*}
\item Axiom \textbf{U3} states that $X$ is an update set yielded by the rule \textbf{forall} $x$ \textbf{with} $\varphi$ \textbf{do} $r$ \textbf{enddo} iff
$X$ coincides with $U_{a_1} \cup \cdots \cup U_{a_n}$, where $\{a_1, \ldots, a_n\} = \{ a_i \in B_{db} \mid val_{S,\zeta[x \mapsto a_i]}(\varphi) = \mathit{true}\}$ and $U_{a_i}$ (for $1 \leq i \leq n$) is an update set yielded by the rule $r$ under the variable assignment $\zeta[x \mapsto a_i]$. Note that the update sets $U_{a_1}, \ldots, U_{a_n}$ are encoded into the second-order variable $Y$ of arity four.
\begin{flalign*}
\textbf{U3: } &\mathrm{upd}(\textbf{forall} \, x \, \textbf{with} \, \varphi \, \textbf{do} \, r \, \textbf{enddo},X) \leftrightarrow \mathrm{isUSet}(X)\wedge&\\
&\exists Y \big(\forall \mathtt{z} y_1 y_2 (X(\mathtt{z},y_1,y_2) \leftrightarrow \exists x (Y(\mathtt{z},y_1,y_2,x))) \wedge&\\
& \hspace*{0.65cm} \forall \mathtt{z} \mathtt{y}_1 \mathtt{y}_2 (X(\mathtt{z},\mathtt{y}_1,\mathtt{y}_2) \leftrightarrow \exists x(Y(\mathtt{z},\mathtt{y}_1,\mathtt{y}_2,x))) \wedge&\\
& \hspace*{0.65cm} \forall \mathtt{z} y_1 \mathtt{y}_2 (X(\mathtt{z},y_1,\mathtt{y}_2) \leftrightarrow \exists x (Y(\mathtt{z},y_1,\mathtt{y}_2,x))) \wedge&\\
& \hspace*{0.65cm} \forall x \big( (\varphi \rightarrow \exists Z (\mathrm{upd}(r,Z) \wedge&\\
&\qquad\qquad\qquad\qquad \forall \mathtt{z} y_1 y_2 (Z(\mathtt{z},y_1,y_2) \leftrightarrow Y(\mathtt{z},y_1,y_2,x)) \wedge&\\
&\qquad\qquad\qquad\qquad \forall \mathtt{z} \mathtt{y}_1 \mathtt{y}_2 (Z(\mathtt{z},\mathtt{y}_1,\mathtt{y}_2) \leftrightarrow Y(\mathtt{z},\mathtt{y}_1,\mathtt{y}_2,x)) \wedge&\\
&\qquad\qquad\qquad\qquad \forall \mathtt{z} y_1 \mathtt{y}_2 (Z(\mathtt{z},y_1,\mathtt{y}_2) \leftrightarrow Y(\mathtt{z},y_1,\mathtt{y}_2,x)))) \wedge&\\
& \hspace*{1.3cm} (\neg \varphi \rightarrow \forall \mathtt{z} y_1 y_2 (\neg Y(\mathtt{z},y_1,y_2,x)) \wedge&\\
&\quad\qquad\qquad\qquad \forall \mathtt{z} \mathtt{y}_1 \mathtt{y}_2 (\neg Y(\mathtt{z},\mathtt{y}_1,\mathtt{y}_2,x)) \wedge&\\
&\quad\qquad\qquad\qquad \forall \mathtt{z} y_1 \mathtt{y}_2 (\neg Y(\mathtt{z},y_1,\mathtt{y}_2,x)))\big)\big)
\end{flalign*}
\item Axiom \textbf{U4} states that $X$ is an update set yielded by the parallel rule \textbf{par} $r_1\hspace{0.2cm} r_2$ \textbf{endpar} iff it corresponds to the
union of an update set yielded by $r_1$ and an update set yielded by $r_2$.
\begin{flalign*}
\textbf{U4: } \mathrm{upd}(\textbf{par} \, r_1 \; r_2 \,& \textbf{endpar}, X) \leftrightarrow \mathrm{isUSet}(X) \wedge&\\
&\exists Y_1 Y_2 \big(\mathrm{upd}(r_1,Y_1) \wedge \mathrm{upd}(r_2,Y_2) \wedge&\\
&\qquad \quad \forall \mathtt{z} y_1 y_2 (X(\mathtt{z},y_1,y_2) \leftrightarrow (Y_1(\mathtt{z},y_1,y_2) \vee Y_2(\mathtt{z},y_1,y_2))) \wedge&\\
&\qquad \quad \forall \mathtt{z} \mathtt{y}_1 \mathtt{y}_2 (X(\mathtt{z},\mathtt{y}_1,\mathtt{y}_2) \leftrightarrow (Y_1(\mathtt{z},\mathtt{y}_1,\mathtt{y}_2) \vee Y_2(\mathtt{z},\mathtt{y}_1,\mathtt{y}_2))) \wedge&\\
&\qquad \quad \forall \mathtt{z} y_1 \mathtt{y}_2 (X(\mathtt{z},y_1,\mathtt{y}_2) \leftrightarrow (Y_1(\mathtt{z},y_1,\mathtt{y}_2) \vee Y_2(\mathtt{z},y_1,\mathtt{y}_2))) \big)&
\end{flalign*}
\item Axiom \textbf{U5} asserts that $X$ is an update set yielded by the rule \textbf{choose} $x$ \textbf{with} $\varphi$
\textbf{do} $r$  \textbf{enddo} iff it is an update set yielded by the rule $r$ under a variable assignment $\zeta[x \mapsto a]$ which satisfies $\varphi$.
\begin{flalign*}
\textbf{U5: }&\mathrm{upd}(\textbf{choose}\, x \, \textbf{with} \, \varphi \, \textbf{do} \, r \, \textbf{enddo}, X) \leftrightarrow \exists x (\varphi \wedge \mathrm{upd}(r,X))&
\end{flalign*}
\item Axiom \textbf{U6} asserts that $X$ is an update set yielded by a sequence rule \textbf{seq} $r_1\hspace{0.2cm} r_2$ \textbf{endseq}
iff it corresponds either to an inconsistent update set yielded by rule $r_1$, or to an update set formed by the updates in an update set $Y_2$ yielded by rule $r_2$ in a successor state $S+Y_1$, where $Y_1$ encodes a consistent set of updates produced by rule $r_1$, plus the updates in $Y_1$ that correspond to locations other than the locations updated by $Y_2$.
\begin{flalign*}
\textbf{U6: }&\mathrm{upd}(\textbf{seq} \, r_1 \; r_2 \, \textbf{endseq}, X) \leftrightarrow \big(\text{upd}(r_1,X) \wedge \neg\mathrm{conUSet}(X)\big) \vee&\\
& \big( \mathit{isUSet}(X) \wedge &\\
& \exists Y_1 Y_2 (\mathrm{upd}(r_1,Y_1) \wedge \mathrm{conUSet}(Y_1) \wedge [Y_1]\mathrm{upd}(r_2,Y_2) \wedge&\\
& \hspace*{0.3cm} \forall \mathtt{z} y_1 y_2 (X(\mathtt{z},y_1,y_2) \leftrightarrow ( (Y_1(\mathtt{z},y_1,y_2) \wedge \forall x (\neg Y_2(\mathtt{z}, y_1, x))) \vee Y_2(\mathtt{z},y_1,y_2))) \wedge&\\
& \hspace*{0.3cm} \forall \mathtt{z} \mathtt{y}_1 \mathtt{y}_2 (X(\mathtt{z},\mathtt{y}_1,\mathtt{y}_2) \leftrightarrow ( (Y_1(\mathtt{z},\mathtt{y}_1,\mathtt{y}_2) \wedge \forall \mathtt{x} (\neg Y_2(\mathtt{z}, \mathtt{y}_1, \mathtt{x}))) \vee Y_2(\mathtt{z},\mathtt{y}_1,\mathtt{y}_2))) \wedge&\\
& \hspace*{0.3cm} \forall \mathtt{z} y_1 \mathtt{y}_2 (X(\mathtt{z},y_1,\mathtt{y}_2) \leftrightarrow ( (Y_1(\mathtt{z},y_1,\mathtt{y}_2) \wedge \forall \mathtt{x} (\neg Y_2(\mathtt{z}, y_1, \mathtt{x}))) \vee Y_2(\mathtt{z},y_1,\mathtt{y}_2))) )\big)
\end{flalign*}
\item Our next axioms assert that $X$ is an update set yielded by the rule \textbf{let} $(f,t) \rightharpoonup\!\rho$ \textbf{in} $r$ \textbf{endlet} iff there is an update \emph{multiset} $Y$ yielded by the rule $r$ that collapses into $X$, when the update values to the
location $(f,t)$ which appear in $Y$ are aggregated using the location operator $\rho$, and
the multiplicity of identical updates to a same location other than $(f,t)$ is ignored. Since $\rho$-terms are algorithmic terms, $f$ can either be a bridge or an algorithmic function symbol. Thus we have two possible cases. 
\begin{align*}
\textbf{U7.1: } & \text{If} \; f \; \text{is an algorithmic function symbol in} \; \Upsilon_a \; \text{then}&\\  
& \mathrm{upd}(\textbf{let} \, (f,t)\!\rightharpoonup\!\rho \, \textbf{in} \, r \,\textbf{endlet},X) \leftrightarrow \mathrm{isUSet}(X) \wedge \exists Y \big(\mathrm{upm}(r,Y) \wedge&\\
&\forall \mathtt{x}_1 \mathtt{x}_2 \mathtt{x}_3 \big(X(\mathtt{x}_1,\mathtt{x}_2,\mathtt{x}_3) \leftrightarrow \big(((\mathtt{x}_1 \neq c_f \vee t \neq \mathtt{x}_2) \wedge \exists \mathtt{z} (Y(\mathtt{x}_1,\mathtt{x}_2,\mathtt{x}_3,\mathtt{z}))) \vee &\\
&\hspace*{3.1cm}(\mathtt{x}_1 = c_f \wedge \mathtt{x}_2 = t \wedge \mathtt{x}_3=\rho_{\mathtt{y}}(\mathtt{y}|\exists \mathtt{z} (Y(\mathtt{x}_1,\mathtt{x}_2,\mathtt{y},\mathtt{z}))))\big)\big) \wedge&\\
&\forall \mathtt{x}_1 x_2 \mathtt{x}_3 \big(X(\mathtt{x}_1,x_2,\mathtt{x}_3) \leftrightarrow \exists \mathtt{z} (Y(\mathtt{x}_1,x_2,\mathtt{x}_3,\mathtt{z}))\big) \wedge &\\
&\forall \mathtt{x}_1 x_2 x_3 \big(X(\mathtt{x}_1,x_2,x_3) \leftrightarrow \exists \mathtt{z} (Y(\mathtt{x}_1,x_2,x_3,\mathtt{z}))\big)\big) &
\end{align*}
\begin{align*}
\textbf{U7.2: } & \text{If} \; f \; \text{is a bridge function symbol in} \; {\cal F}_b \; \text{then}&\\  
& \mathrm{upd}(\textbf{let} \, (f,t)\!\rightharpoonup\!\rho \, \textbf{in} \, r \,\textbf{endlet},X) \leftrightarrow \mathrm{isUSet}(X) \wedge \exists Y \big(\mathrm{upm}(r,Y) \wedge&\\
&\forall \mathtt{x}_1 x_2 \mathtt{x}_3 \big(X(\mathtt{x}_1,x_2,\mathtt{x}_3) \leftrightarrow \big(((\mathtt{x}_1 \neq c_f \vee t \neq x_2) \wedge \exists \mathtt{z} (Y(\mathtt{x}_1,x_2,\mathtt{x}_3,\mathtt{z}))) \vee &\\
&\hspace*{3.1cm}(\mathtt{x}_1 = c_f \wedge x_2 = t \wedge \mathtt{x}_3=\rho_{\mathtt{y}}(\mathtt{y}|\exists \mathtt{z} (Y(\mathtt{x}_1,x_2,\mathtt{y},\mathtt{z}))))\big)\big) \wedge&\\
&\forall \mathtt{x}_1 \mathtt{x}_2 \mathtt{x}_3 \big(X(\mathtt{x}_1,\mathtt{x}_2,\mathtt{x}_3) \leftrightarrow \exists \mathtt{z} (Y(\mathtt{x}_1,\mathtt{x}_2,\mathtt{x}_3,\mathtt{z}))\big) \wedge &\\
&\forall \mathtt{x}_1 x_2 x_3 \big(X(\mathtt{x}_1,x_2,x_3) \leftrightarrow \exists \mathtt{z} (Y(\mathtt{x}_1,x_2,x_3,\mathtt{z}))\big)\big) &
\end{align*}
\end{itemize}

The following lemma is a direct consequence of Axioms~\textbf{U1}--\textbf{U7}. 

\begin{lemma}\label{lem-upd}

Each formula in the DB-ASM logic ${\cal L}^{db}$ can be replaced by an equivalent formula not containing any subformulae of the form $\mathrm{upd}(r,X)$.

\end{lemma}

\subsection{Update Multisets}\label{sub:UpdateMultisets}

Each DB-ASM rule is associated with a set of update multisets as defined in Fig.~\ref{fig:multiset}. The axioms presented in this section assert how an update multiset is yielded by a DB-ASM rule, i.e., they define the predicate $\mathrm{upm}(r,X)$.  

Same as in the axioms for update sets, we assume that if $f$ is a dynamic function symbol in the given signature of meta-finite states $\Upsilon = \Upsilon_{db} \cup \Upsilon_a \cup {\cal F}_b$, then there is a corresponding constant (static nullary function) symbol $c_f \in \Upsilon_a$ as per Definition~\ref{extendedState}. We use ${\cal F}_{\mathit{dyn}}$ to denote the set of all $c_f$ such that $f$ is a dynamic function symbol in $\Upsilon$. Again, $S$ denotes an arbitrary Henkin structure of signature $\Upsilon$, $B=B_{db}\cup B_a$ denotes the base set (domain) of the database and algorithmic parts of $S$, and w.l.o.g. we assume $B_{db} \cap B_a = \emptyset$. 

In the formulation of the axioms we use the predicate $\mathrm{is\ddot{U}Set}(X)$ which is analogous to the predicate $\mathrm{isUSet}(X)$ defined in the case of update sets. Let $\zeta$ be a valuation into $S$, $\mathrm{is\ddot{U}Set}(X)$ expresses that $\zeta(X)$ represents an update multiset set for $S$.  That is, for every tuple $(a_1, a_2, a_3, a_4) \in \zeta(X)$, we have that $a_1 = c^S_f$ for some dynamic function $f \in \Upsilon$, $a_4$ is an arbitrary value of $B_a$, and $a_2, a_3$ are values of the appropriate database or algorithmic base sets depending on whether $f$ is a database, algorithmic or bridge function symbol.
\begin{align*}
\mathrm{is\ddot{U}Set}(X) \equiv& \forall \mathtt{x}_1 x_2 x_3 \mathtt{x}_4 \Big(X(\mathtt{x}_1, x_2, x_3, \mathtt{x}_4) \rightarrow \bigvee_{c_f \in {\cal F}_\mathit{dyn} \wedge f \in \Upsilon_{db}} \mathtt{x}_1 = c_f\Big) \wedge&\\
&\forall \mathtt{x}_1 \mathtt{x}_2 \mathtt{x}_3 \mathtt{x}_4 \Big(X(\mathtt{x}_1, \mathtt{x}_2, \mathtt{x}_3, \mathtt{x}_4) \rightarrow \bigvee_{c_f \in {\cal F}_\mathit{dyn} \wedge f \in \Upsilon_{a}} \mathtt{x}_1 = c_f\Big) \wedge&\\
&\forall \mathtt{x}_1 x_2 \mathtt{x}_3 \mathtt{x}_4 \Big(X(\mathtt{x}_1, x_2, \mathtt{x}_3, \mathtt{x}_4) \rightarrow \bigvee_{c_f \in {\cal F}_\mathit{dyn} \wedge f \in {\cal F}_{b}} \mathtt{x}_1 = c_f\Big) \wedge&\\
&\forall \mathtt{x}_1 \mathtt{x}_2 x_3 \mathtt{x}_4(\neg X(\mathtt{x}_1, \mathtt{x}_2, x_3, \mathtt{x}_4)) \wedge \forall \mathtt{x}_1 \mathtt{x}_2 x_3 x_4(\neg X(\mathtt{x}_1, \mathtt{x}_2, x_3, x_4)) \wedge&\\
&\forall x_1 x_2 x_3 \mathtt{x}_4 ( \neg X(x_1, x_2, x_3, \mathtt{x}_4)) \wedge \forall x_1 \mathtt{x}_2 x_3 \mathtt{x}_4( \neg X(x_1, \mathtt{x}_2, x_3, \mathtt{x}_4)) \wedge&\\
&\forall x_1 x_2 x_3 x_4 ( \neg X(x_1, x_2, x_3, x_4)) \wedge \forall x_1 \mathtt{x}_2 x_3 x_4( \neg X(x_1, \mathtt{x}_2, x_3, x_4)) \wedge&\\
&\forall x_1 x_2 \mathtt{x}_3 \mathtt{x}_4 ( \neg X(x_1, x_2, \mathtt{x}_3, \mathtt{x}_4)) \wedge \forall x_1 \mathtt{x}_2 \mathtt{x}_3 \mathtt{x}_4 ( \neg X(x_1, \mathtt{x}_2, \mathtt{x}_3, \mathtt{x}_4))\wedge&\\
&\forall x_1 x_2 \mathtt{x}_3 x_4 ( \neg X(x_1, x_2, \mathtt{x}_3, x_4)) \wedge \forall x_1 \mathtt{x}_2 \mathtt{x}_3 x_4 ( \neg X(x_1, \mathtt{x}_2, \mathtt{x}_3, x_4))&
\end{align*}

The axioms for the predicate upm$(r,X)$ are analogous to the axioms for the predicate upd$(r,X)$, except for the fact that we need to deal with multisets represented as relations.
\begin{itemize}
\item Axioms~$\mathbf{\ddot{U}1.1}$--$\mathbf{\ddot{U}1.3}$ express that $X$ represents an update multiset yielded by the assignment rule $f(t) := s$ iff it contains exactly one update with multiplicity $1$, and that update is $(f,t,s)$. 
\begin{align*}
\mathbf{\ddot{U}1.1}\textbf{: } & \text{If} \; f \; \text{is a database function symbol in} \; \Upsilon_{db} \; \text{then}&\\  
&\mathrm{upm}(f(t) := s, X) \leftrightarrow \mathrm{is\ddot{U}Set}(X) \wedge \exists \mathtt{z} \big(X(c_f, t, s, \mathtt{z}) \wedge &\\
&\qquad \forall \mathtt{x}_1 x_2 x_3 \mathtt{x}_4 (X(\mathtt{x}_1,x_2,x_3, \mathtt{x}_4) \rightarrow \mathtt{x}_1 = c_f \wedge x_2 = t \wedge x_3 = s \wedge \mathtt{x}_4 = \mathtt{z})\big) \wedge&\\
&\qquad \forall \mathtt{x}_1 \mathtt{x}_2 \mathtt{x}_3 \mathtt{x}_4 (\neg X(\mathtt{x}_1,\mathtt{x}_2,\mathtt{x}_3,\mathtt{x}_4)) \wedge \forall \mathtt{x}_1 x_2 \mathtt{x}_3 \mathtt{x}_4(\neg X(\mathtt{x}_1,x_2,\mathtt{x}_3,\mathtt{x}_4)) 
\end{align*}
\begin{align*}
\mathbf{\ddot{U}1.2}\textbf{: } & \text{If} \; f \; \text{is an algorithmic function symbol in} \; \Upsilon_{a} \; \text{then}&\\  
&\mathrm{upm}(f(t) := s, X) \leftrightarrow \mathrm{is\ddot{U}Set}(X) \wedge \exists \mathtt{z} \big(X(c_f, t, s, \mathtt{z}) \wedge &\\
&\qquad \forall \mathtt{x}_1 \mathtt{x}_2 \mathtt{x}_3 \mathtt{x}_4 (X(\mathtt{x}_1,\mathtt{x}_2,\mathtt{x}_3,\mathtt{x}_4) \rightarrow \mathtt{x}_1 = c_f \wedge \mathtt{x}_2 = t \wedge \mathtt{x}_3 = s \wedge \mathtt{x}_4 = \mathtt{z})\big) \wedge&\\
&\qquad \forall \mathtt{x}_1 x_2 x_3 \mathtt{x}_4 (\neg X(\mathtt{x}_1,x_2,x_3,\mathtt{x}_4)) \wedge \forall \mathtt{x}_1 x_2 \mathtt{x}_3 \mathtt{x}_4 (\neg X(\mathtt{x}_1,x_2,\mathtt{x}_3,\mathtt{x}_4)) 
\end{align*}
\begin{align*}
\mathbf{\ddot{U}1.3}\textbf{: }  & \text{If} \; f \; \text{is a bridge function symbol in} \; {\cal F}_b \; \text{then}&\\  
&\mathrm{upm}(f(t) := s, X) \leftrightarrow \mathrm{is\ddot{U}Set}(X) \wedge \exists \mathtt{z} \big(X(c_f, t, s, \mathtt{z}) \wedge&\\
&\qquad \forall \mathtt{x}_1 x_2 \mathtt{x}_3 \mathtt{x}_4 (X(\mathtt{x}_1,x_2,\mathtt{x}_3,\mathtt{x}_4) \rightarrow \mathtt{x}_1 = c_f \wedge x_2 = t \wedge \mathtt{x}_3 = s \wedge \mathtt{x}_4 = \mathtt{z})\big) \wedge&\\
&\qquad \forall \mathtt{x}_1 x_2 x_3 \mathtt{x}_4 (\neg X(\mathtt{x}_1,x_2,x_3,\mathtt{x}_4)) \wedge \forall \mathtt{x}_1 \mathtt{x}_2 \mathtt{x}_3 \mathtt{x}_4 (\neg X(\mathtt{x}_1,\mathtt{x}_2,\mathtt{x}_3,\mathtt{x}_4))& 
\end{align*}
\item Axiom $\mathbf{\ddot{U}2}$ asserts that, if the formula $\varphi$ evaluates to $\mathit{true}$, then
$X$ is an update multiset yielded by the conditional rule \textbf{if} $\varphi$
\textbf{then} $r$ \textbf{endif} iff $X$ is an update multiset yielded by the rule
$r$. Otherwise, the conditional rule yields only an empty update multiset.
\begin{flalign*}
\mathbf{\ddot{U}2}\textbf{: } \mathrm{upm}(\textbf{if} \, \varphi \,& \textbf{then} \, r \, \textbf{endif}, X) \leftrightarrow (\varphi \wedge \mathrm{upm}(r,X)) \vee &\\
& \big(\neg \varphi \wedge \mathrm{is\ddot{U}Set}(X) \wedge \forall \mathtt{x}_1 x_2 x_3 \mathtt{x}_4 (\neg X(\mathtt{x}_1,x_2,x_3, \mathtt{x}_4)) \wedge &\\
&\quad \forall \mathtt{x}_1 \mathtt{x}_2 \mathtt{x}_3 \mathtt{x}_4 (\neg X(\mathtt{x}_1,\mathtt{x}_2,\mathtt{x}_3, \mathtt{x}_4)) \wedge \forall \mathtt{x}_1 x_2 \mathtt{x}_3 \mathtt{x}_4(\neg X(\mathtt{x}_1,x_2,\mathtt{x}_3, \mathtt{x}_4))\big) &
\end{flalign*}
\item Axiom $\mathbf{\ddot{U}3}$ states that $X$ is an update multiset yielded by the rule \textbf{forall} $x$ \textbf{with} $\varphi$ \textbf{do} $r$ \textbf{enddo} iff
$X$ coincides with $\ddot{U}_{a_1} \uplus \cdots \uplus \ddot{U}_{a_n}$, where $\{a_1, \ldots, a_n\} = \{ a_i \in B_{db} \mid val_{S,\zeta[x \mapsto a_i]}(\varphi) = \mathit{true}\}$ and $\ddot{U}_{a_i}$ (for $1 \leq i \leq n$) is an update multiset yielded by the rule $r$ under the variable assignment $\zeta[x \mapsto a_i]$. Note that the update multisets $\ddot{U}_{a_1}, \ldots, \ddot{U}_{a_n}$ are encoded into the second-order variable $Y$ of arity five. 
We use the informal expression ``$F$ is a bijection from $X$ to $Y$'' to denote that there is a bijection $f$ from $\zeta(X)$ to $\zeta(Y)$ such that $F(a_1, a_2, a_3, a_4, b_1, b_2, b_3, b_4, b_5)$ iff $f((a_1, a_2, a_3, a_4)) = (b_1, b_2, b_3, b_4, b_5)$. It is a well known fact that $F$ can be easily defined in first-order logic (see for instance~\cite{FerrarottiRT14}). \\

\begin{flalign*}
\mathbf{\ddot{U}3}\textbf{: } &\mathrm{upm}(\textbf{forall} \, x \, \textbf{with} \, \varphi \, \textbf{do} \, r \, \textbf{enddo},X) \leftrightarrow \mathrm{is\ddot{U}Set}(X)\wedge&\\
&\exists Y F \big(\text{``$F$ is a bijection from $X$ to $Y$''} \wedge &\\
&\qquad \quad \forall \mathtt{z} y_1 y_2 \mathtt{y}_3 \mathtt{z}' y_1' y_2' \mathtt{y}_3' x (F(\mathtt{z},y_1,y_2,\mathtt{y}_3, \mathtt{z}',y_1',y_2',\mathtt{y}_3',x) \rightarrow&\\
& \hspace*{7cm} (\mathtt{z} = \mathtt{z}' \wedge y_1 = y_1' \wedge  y_2 = y_2')) \wedge&\\
& \qquad \quad \forall \mathtt{z} \mathtt{y}_1 \mathtt{y}_2 \mathtt{y}_3 \mathtt{z}' \mathtt{y}_1' \mathtt{y}_2' \mathtt{y}_3' x (F(\mathtt{z}, \mathtt{y}_1, \mathtt{y}_2, \mathtt{y}_3, \mathtt{z}', \mathtt{y}_1', \mathtt{y}_2', \mathtt{y}_3', x) \rightarrow&\\
& \hspace*{7cm} (\mathtt{z} = \mathtt{z}' \wedge \mathtt{y}_1 = \mathtt{y}_1' \wedge  \mathtt{y}_2 = \mathtt{y}_2')) \wedge&\\
& \qquad \quad \forall \mathtt{z} y_1 \mathtt{y}_2 \mathtt{y}_3 \mathtt{z}' y_1' \mathtt{y}_2' \mathtt{y}_3' x (F(\mathtt{z}, y_1, \mathtt{y}_2, \mathtt{y}_3, \mathtt{z}', y_1', \mathtt{y}_2', \mathtt{y}_3', x) \rightarrow&\\
& \hspace*{7cm} (\mathtt{z} = \mathtt{z}' \wedge y_1 = y_1' \wedge  \mathtt{y}_2 = \mathtt{y}_2')) \wedge&\\
& \qquad \quad \forall x \big( (\varphi \rightarrow \exists Z (\mathrm{upm}(r,Z) \wedge&\\
&\quad\qquad\qquad\qquad\qquad \forall \mathtt{z} y_1 y_2 \mathtt{y}_3 (Z(\mathtt{z},y_1,y_2,\mathtt{y}_3) \leftrightarrow Y(\mathtt{z},y_1,y_2,\mathtt{y}_3,x)) \wedge&\\
&\quad\qquad\qquad\qquad\qquad \forall \mathtt{z} \mathtt{y}_1 \mathtt{y}_2 \mathtt{y}_3 (Z(\mathtt{z},\mathtt{y}_1,\mathtt{y}_2, \mathtt{y}_3) \leftrightarrow Y(\mathtt{z},\mathtt{y}_1,\mathtt{y}_2, \mathtt{y}_3, x)) \wedge&\\
&\quad\qquad\qquad\qquad\qquad \forall \mathtt{z} y_1 \mathtt{y}_2 \mathtt{y}_3 (Z(\mathtt{z},y_1,\mathtt{y}_2,\mathtt{y}_3) \leftrightarrow Y(\mathtt{z},y_1,\mathtt{y}_2,\mathtt{y}_3,x)))) \wedge&\\
& \qquad \qquad \quad (\neg \varphi \rightarrow \forall \mathtt{z} y_1 y_2 \mathtt{y}_3 (\neg Y(\mathtt{z},y_1,y_2,\mathtt{y}_3,x)) \wedge&\\
&\qquad\qquad\qquad\qquad\; \forall \mathtt{z} \mathtt{y}_1 \mathtt{y}_2 \mathtt{y}_3 (\neg Y(\mathtt{z},\mathtt{y}_1,\mathtt{y}_2,\mathtt{y}_3,x)) \wedge&\\
&\qquad\qquad\qquad\qquad\; \forall \mathtt{z} y_1 \mathtt{y}_2 \mathtt{y}_3 (\neg Y(\mathtt{z},y_1,\mathtt{y}_2,\mathtt{y}_3, x)))\big)\big)
\end{flalign*}
\item Axiom $\mathbf{\ddot{U}4}$ states that $X$ represents an update multiset yielded by the rule \textbf{par} $r_1\hspace{0.2cm} r_2$ \textbf{endpar} iff $X$ represents an update multiset $\ddot{U}_1 \uplus \ddot{U}_1$ where $\ddot{U}_1$ is an update multiset yielded by $r_1$ and $\ddot{U}_2$ is an update multiset yielded by $r_2$. 
\begin{flalign*}
\mathbf{\ddot{U}4}\textbf{: } &\mathrm{upm}(\textbf{par} \, r_1 \; r_2 \, \textbf{endpar}, X) \leftrightarrow \mathrm{is\ddot{U}Set}(X) \wedge&\\
&\exists Y_1 Y_2 \big(\mathrm{upm}(r_1,Y_1) \wedge \mathrm{upm}(r_2,Y_2) \wedge&\\
&\quad \qquad\forall \mathtt{z} y_1 y_2 \mathtt{y}_3 \mathtt{z}' y_1' y_2' \mathtt{y}_3' (Y_1(\mathtt{z},y_1,y_2, \mathtt{y}_3) \wedge Y_2(\mathtt{z}',y_1',y_2',\mathtt{y}_3') \rightarrow \mathtt{y}_3 \neq \mathtt{y}_3') \wedge&\\
&\quad \qquad\forall \mathtt{z} \mathtt{y}_1 \mathtt{y}_2 \mathtt{y}_3 \mathtt{z}' \mathtt{y}_1' \mathtt{y}_2' \mathtt{y}_3' (Y_1(\mathtt{z},\mathtt{y}_1,\mathtt{y}_2, \mathtt{y}_3) \wedge Y_2(\mathtt{z}',\mathtt{y}_1',\mathtt{y}_2',\mathtt{y}_3') \rightarrow \mathtt{y}_3 \neq \mathtt{y}_3') \wedge&\\
&\quad \qquad\forall \mathtt{z} y_1 \mathtt{y}_2 \mathtt{y}_3 \mathtt{z}' y_1' \mathtt{y}_2' \mathtt{y}_3' (Y_1(\mathtt{z},y_1,\mathtt{y}_2, \mathtt{y}_3) \wedge Y_2(\mathtt{z}',y_1',\mathtt{y}_2',\mathtt{y}_3') \rightarrow \mathtt{y}_3 \neq \mathtt{y}_3') \wedge&\\
&\qquad \quad \forall \mathtt{z} y_1 y_2 \mathtt{y}_3 (X(\mathtt{z},y_1,y_2, \mathtt{y}_3) \leftrightarrow (Y_1(\mathtt{z},y_1,y_2, \mathtt{y}_3) \vee Y_2(\mathtt{z},y_1,y_2,\mathtt{y}_3))) \wedge&\\
&\qquad \quad \forall \mathtt{z} \mathtt{y}_1 \mathtt{y}_2 \mathtt{y}_3 (X(\mathtt{z},\mathtt{y}_1,\mathtt{y}_2, \mathtt{y}_3) \leftrightarrow (Y_1(\mathtt{z},\mathtt{y}_1,\mathtt{y}_2, \mathtt{y}_3) \vee Y_2(\mathtt{z},\mathtt{y}_1,\mathtt{y}_2, \mathtt{y}_3))) \wedge&\\
&\qquad \quad \forall \mathtt{z} y_1 \mathtt{y}_2 \mathtt{y}_3 (X(\mathtt{z},y_1,\mathtt{y}_2, \mathtt{y}_3) \leftrightarrow (Y_1(\mathtt{z},y_1,\mathtt{y}_2, \mathtt{y}_3) \vee Y_2(\mathtt{z},y_1,\mathtt{y}_2, \mathtt{y}_3))) \big)&
\end{flalign*}
\item Axiom  $\mathbf{\ddot{U}5}$ asserts that $X$ is an update multiset yielded by the rule \textbf{choose} $x$ \textbf{with} $\varphi$
\textbf{do} $r$  \textbf{enddo} iff it is an update multiset yielded by the rule $r$ under a variable assignment $\zeta[x \mapsto a]$ which satisfies $\varphi$.
\begin{flalign*}
\mathbf{\ddot{U}5}\textbf{: }&\mathrm{upm}(\textbf{choose}\, x \, \textbf{with} \, \varphi \, \textbf{do} \, r \, \textbf{enddo}, X) \leftrightarrow \exists x (\varphi \wedge \mathrm{upm}(r,X))&
\end{flalign*}
\item Axiom $\mathbf{\ddot{U}6}$ asserts that $X$ is an update multiset yielded by a sequence rule \textbf{seq} $r_1\hspace{0.2cm} r_2$ \textbf{endseq}
iff it corresponds either to an inconsistent update multiset $Y_1$ yielded by rule $r_1$, or to an update multiset formed by the updates in an update multiset $Y_2$ yielded by rule $r_2$ in a successor state $S+U$, where $U$ is the set of updates which appear in the multiset $Y_1$, plus the updates in $Y_1$ that correspond to locations other than the locations that appear in the updates in $Y_2$.
\begin{flalign*}
\mathbf{\ddot{U}6}\textbf{: }&\mathrm{upm}(\textbf{seq} \, r_1 \; r_2 \, \textbf{endseq}, X) \leftrightarrow \big(\text{upm}(r_1,X) \wedge \neg\mathrm{con\ddot{U}Set}(X)\big) \vee&\\
& \big( \mathrm{is\ddot{U}Set}(X) \wedge &\\
& \exists Y_1 Y_1' Y_2 \big(\mathrm{upm}(r_1,Y_1) \wedge \mathrm{con\ddot{U}Set}(Y_1) \wedge \textrm{isUSet}(Y_1') \wedge [Y_1']\mathrm{upm}(r_2,Y_2) \wedge&\\
& \forall \mathtt{z} y_1 y_2 (Y_1'(\mathtt{z},y_1,y_2) \leftrightarrow \exists \mathtt{y}_3 (Y_1(\mathtt{z},y_1,y_2,\mathtt{y}_3))) \wedge&\\
& \forall \mathtt{z} \mathtt{y}_1 \mathtt{y}_2 (Y_1'(\mathtt{z},\mathtt{y}_1,\mathtt{y}_2) \leftrightarrow \exists \mathtt{y}_3 (Y_1(\mathtt{z},\mathtt{y}_1,\mathtt{y}_2, \mathtt{y}_3))) \wedge&\\
& \forall \mathtt{z} y_1 \mathtt{y}_2 (Y_1'(\mathtt{z},y_1,\mathtt{y}_2) \leftrightarrow \exists \mathtt{y}_3 (Y_1(\mathtt{z},y_1,\mathtt{y}_2, \mathtt{y}_3))) \wedge &\\
& \forall \mathtt{z} y_1 y_2 \mathtt{y}_3(X(\mathtt{z},y_1,y_2,\mathtt{y_3}) \leftrightarrow ( (Y_1(\mathtt{z},y_1,y_2,\mathtt{y}_3) \wedge \forall x_1 \mathtt{x}_2 (\neg Y_2(\mathtt{z}, y_1, x_1, \mathtt{x}_2))) \vee&\\
&\hspace*{4.6cm} Y_2(\mathtt{z},y_1,y_2, \mathtt{y}_3))) \wedge&\\
& \forall \mathtt{z} \mathtt{y}_1 \mathtt{y}_2 \mathtt{y}_3 (X(\mathtt{z},\mathtt{y}_1,\mathtt{y}_2, \mathtt{y}_3) \leftrightarrow ( (Y_1(\mathtt{z},\mathtt{y}_1,\mathtt{y}_2, \mathtt{y}_3) \wedge \forall \mathtt{x}_1 \mathtt{x}_2 (\neg Y_2(\mathtt{z}, \mathtt{y}_1, \mathtt{x}_1, \mathtt{x}_2))) \vee &\\
&\hspace*{4.6cm} Y_2(\mathtt{z},\mathtt{y}_1,\mathtt{y}_2, \mathtt{y}_3))) \wedge&\\
& \forall \mathtt{z} y_1 \mathtt{y}_2 \mathtt{y}_3 (X(\mathtt{z},y_1,\mathtt{y}_2,\mathtt{y}_3) \leftrightarrow ( (Y_1(\mathtt{z},y_1,\mathtt{y}_2,\mathtt{y}_3) \wedge \forall \mathtt{x}_1 \mathtt{x}_2 (\neg Y_2(\mathtt{z}, y_1, \mathtt{x}_1, \mathtt{x}_2))) \vee&\\
& \hspace*{4.6cm} Y_2(\mathtt{z},y_1,\mathtt{y}_2, \mathtt{y}_3))) \big)\big)
\end{flalign*}
\item Our next axioms assert that $X$ is  an update multiset yielded by a let rule $\textbf{let} \, (f, t) \!\rightharpoonup\! \rho \, \textbf{in} \, r \, \textbf{endlet}$ iff it corresponds to an update multiset $Y$ yielded by the rule $r$ except for the updates to the location $(f,t)$ which are collapsed into a unique update in $X$ by aggregating their values using the operator $\rho$. Again notice that a $\rho$-term is an algorithmic term and thus we have to consider only two cases.
\begin{align*}
\mathbf{\ddot{U}7.1}\textbf{: } & \text{If} \; f \; \text{is an algorithmic function symbol in} \; \Upsilon_a \; \text{then}&\\  
& \mathrm{upm}(\textbf{let} \, (f,t)\!\rightharpoonup\!\rho \, \textbf{in} \, r \,\textbf{endlet},X) \leftrightarrow \mathrm{is\ddot{U}Set}(X) \wedge \exists Y \mathtt{z} \big(\mathrm{upm}(r,Y) \wedge&\\
&\forall \mathtt{x}_1 \mathtt{x}_2 \mathtt{x}_3 \mathtt{x}_4 \big(X(\mathtt{x}_1,\mathtt{x}_2,\mathtt{x}_3,\mathtt{x}_4) \leftrightarrow \big(((\mathtt{x}_1 \neq c_f \vee t \neq \mathtt{x}_2) \wedge Y(\mathtt{x}_1,\mathtt{x}_2,\mathtt{x}_3,\mathtt{x}_4)) \vee &\\
&\hspace*{1.6cm}(\mathtt{x}_1 = c_f \wedge \mathtt{x}_2 = t \wedge \mathtt{x}_3=\rho_{\mathtt{y}}(\mathtt{y}|\exists \mathtt{x}_0 (Y(\mathtt{x}_1,\mathtt{x}_2,\mathtt{y},\mathtt{x}_0))) \wedge \mathtt{x}_4 = \mathtt{z})\big)\big) \wedge&\\
&\forall \mathtt{x}_1 x_2 \mathtt{x}_3 \mathtt{x}_4 \big(X(\mathtt{x}_1,x_2,\mathtt{x}_3, \mathtt{x}_4) \leftrightarrow Y(\mathtt{x}_1,x_2,\mathtt{x}_3,\mathtt{x}_4)\big) \wedge &\\
&\forall \mathtt{x}_1 x_2 x_3 \mathtt{x}_4 \big(X(\mathtt{x}_1,x_2,x_3, \mathtt{x}_4) \leftrightarrow Y(\mathtt{x}_1,x_2,x_3,\mathtt{x}_4)\big)\big) &
\end{align*}
\begin{align*}
\mathbf{\ddot{U}7.2} \textbf{: } & \text{If} \; f \; \text{is a bridge function symbol in} \; {\cal F}_b \; \text{then}&\\
& \mathrm{upm}(\textbf{let} \, (f,t)\!\rightharpoonup\!\rho \, \textbf{in} \, r \,\textbf{endlet},X) \leftrightarrow \mathrm{is\ddot{U}Set}(X) \wedge \exists Y \mathtt{z} \big(\mathrm{upm}(r,Y) \wedge&\\
&\forall \mathtt{x}_1 x_2 \mathtt{x}_3 \mathtt{x}_4 \big(X(\mathtt{x}_1,x_2,\mathtt{x}_3,\mathtt{x}_4) \leftrightarrow \big(((\mathtt{x}_1 \neq c_f \vee t \neq x_2) \wedge Y(\mathtt{x}_1,x_2,\mathtt{x}_3,\mathtt{x}_4)) \vee &\\
&\hspace*{1.6cm}(\mathtt{x}_1 = c_f \wedge x_2 = t \wedge \mathtt{x}_3=\rho_{\mathtt{y}}(\mathtt{y}|\exists \mathtt{x}_0 (Y(\mathtt{x}_1,x_2,\mathtt{y},\mathtt{x}_0))) \wedge \mathtt{x}_4 = \mathtt{z})\big)\big) \wedge&\\
&\forall \mathtt{x}_1 \mathtt{x}_2 \mathtt{x}_3 \mathtt{x}_4 \big(X(\mathtt{x}_1,\mathtt{x}_2,\mathtt{x}_3, \mathtt{x}_4) \leftrightarrow Y(\mathtt{x}_1,\mathtt{x}_2,\mathtt{x}_3,\mathtt{x}_4)\big) \wedge &\\
&\forall \mathtt{x}_1 x_2 x_3 \mathtt{x}_4 \big(X(\mathtt{x}_1,x_2,x_3, \mathtt{x}_4) \leftrightarrow Y(\mathtt{x}_1,x_2,x_3,\mathtt{x}_4)\big)\big) &
\end{align*}
\end{itemize}

Analogous to Lemma \ref{lem-upd}, the following result is a straightforward consequence of Axioms~$\mathbf{\ddot{U}1}$--$\mathbf{\ddot{U}7}$.

\begin{lemma}\label{lem-upm}

Each formula in the DB-ASM logic ${\cal L}^{db}$ can be replaced by an equivalent formula not containing any subformulae of the form $\mathrm{upm}(r,X)$.

\end{lemma}

\subsection{Axioms and Inference Rules}\label{sub:AxiomsRules}

We present a set of axioms and inference rules which constitute a proof system for the logic ${\cal L}^{db}$ for DB-ASMs. A good starting point is the natural formalism $L_2$ as defined in~\cite{Leivant94} for the relational variant of second-order logic on which ${\cal L}^{db}$ is based. $L_2$ uses the usual axioms and rules for first-order logic, with quantifier rules applying to second-order variables as well as first-order variables, and with the stipulation that the range of the second-order variables includes \emph{at least} all the relations definable by the formulae of the language. 

A deductive calculus for $L_2$ is obtained by augmenting the inference rules and axioms of first-order logic with the comprehension axiom schema~\textbf{SO-C} (which is a form of the Comprehension Principle of Set Theory), and with the axiom schema of universal instantiation \textbf{SO-UI} and the inference rule of universal generalization \textbf{SO-UG} for second-order variables. 

\begin{description}

  \item[\textbf{SO-C}] $\exists X \forall v_1, \ldots, v_k ( X(v_1, \ldots, v_k) \leftrightarrow \varphi)$, where $k \geq 1$, $v_1, \ldots, v_k$ are first-order variables from ${\cal X}_{db} \cup {\cal X}_a$, and $X$ is a $k$-ary second-order variable which does not occur free in the formula $\varphi$. \smallskip

  \item[\textbf{SO-UI}] $\forall X(\varphi) \rightarrow \varphi[Y/X]$, provided the arity of $X$ and $Y$ coincides. \smallskip

  \item[\textbf{SO-UG}] $\frac{\psi \rightarrow \varphi[Y/X]}{\psi \rightarrow \forall X (\varphi)}$, provided $Y$ is not free in $\psi$. \smallskip

\end{description}

The axioms and rules of $L_2$ together with the axioms for update sets and multisets form the basis of the proposed proof system for the logic ${\cal L}^{db}$ for DB-ASMs. The complete list of axioms and rules is composed by: 

\begin{itemize}

\item The Axioms~\textbf{SO-C}, \textbf{SO-UI} and \textbf{SO-UG} of the deductive calculus $L_2$.

\smallskip

\item The Axioms~\textbf{U1}--\textbf{U7} in Section~\ref{sub:UpdateSets} which assert the properties of upd$(r, X)$.

\smallskip

\item The axioms $\mathbf{\ddot{U}1}$--$\mathbf{\ddot{U}7}$ in Section~\ref{sub:UpdateMultisets} which assert the properties of upm$(r, X)$.

\smallskip

\item Axiom \textbf{M1} and Rules \textbf{M2-M3} from the axiom system K of modal logic, which is the weakest
normal modal logic system \cite{hughes:modallogic1996}. Axiom~\textbf{M1} is called \emph{Distribution Axiom} of K,
Rule~\textbf{M2} is called \emph{Necessitation Rule} of K and Rule~\textbf{M3} is the inference rule called \emph{Modus Ponens} in the
classical logic. By using these axioms and rules together, we are able to derive all modal properties that are valid in Kripke frames.

\begin{description}

  \item[\textbf{M1}] $[X](\varphi\rightarrow\psi)\rightarrow ([X]\varphi\rightarrow[X]\psi)$\smallskip

  \item[\textbf{M2}] $\frac{\varphi}{[X]\varphi}$ \smallskip

  \item[\textbf{M3}] $\frac{\varphi, \varphi\rightarrow\psi}{\psi}$\smallskip
\end{description}

\item Axiom \textbf{M4} asserts that, if an update set $X$ is not consistent, then
there is no successor state obtained after applying $X$ over
the current state and thus $[X]\varphi$ is interpreted as true
for any formula $\varphi$. As applying a consistent update set $X$
over the current state is deterministic, Axiom \textbf{M5} describes
the deterministic accessibility relation in terms of $[X]$.

\begin{description}

  \item[\textbf{M4}] $\neg \mathrm{conUSet}(X)\rightarrow[X]\varphi $ \smallskip

  \item[\textbf{M5}] $\neg[X]\varphi\rightarrow [X]\neg\varphi$ \smallskip

\end{description}

\item Axiom \textbf{M6} is called \emph{Barcan Axiom}. It originates from the fact that all
states in a run of a DB-ASM have the same base set, and thus the
quantifiers in all states always range over the same set of
elements. 

\begin{description}

\item[\textbf{M6}] $\forall v ([X]\varphi)\rightarrow [X]\forall
v(\varphi)$, where $v$ stands for any first-order variable in ${\cal X}_{db} \cup {\cal X}_{a}$ or any second-order variable. \smallskip

\end{description}

\item Axioms \textbf{M7} and \textbf{M8} assert that the interpretation of static and pure formulae is the same in all states of a DB-ASM, which is not affected by the execution of any DB-ASM rule $r$. 

\begin{description}
  \item[\textbf{M7}] $\varphi  \wedge \mathrm{upd}(r, X) \rightarrow  [X]\varphi$, for
  static and pure $\varphi$\smallskip

  \item[\textbf{M8}]  con$(r,X)\wedge[X]\varphi\rightarrow
  \varphi$, for static and pure $\varphi$\smallskip

\end{description}

\item Axioms \textbf{A1.1}--\textbf{A1.3} assert that, if a consistent update set $X$ does
not contain any update to a given location, then the
content of that location in a successor state obtained after
applying $X$ remains unchanged. 
Axioms~\textbf{A2.1}--\textbf{A2.3} assert that, if a consistent update set $X$
contains an update $(f,a,b)$, then the content of the location $(f,a)$ in
the successor state obtained after applying $X$ is equal to
$b$. Axiom \textbf{A3} says that, if a DB-ASM
rule $r$ yields an update multiset, then the rule $r$ also yields an
update set.

\begin{description}

  \item[\textbf{A1.1}] If $f$ is a database function symbol, \\
$\mathrm{conUSet}(X) \wedge \forall z (\neg X(c_f,x,z)) \wedge f(x)=y \rightarrow [X]f(x) = y$  \smallskip

  \item[\textbf{A1.2}] If $f$ is an algorithmic function symbol, \\
$\mathrm{conUSet}(X) \wedge \forall \mathtt{z} (\neg X(c_f,\mathtt{x},\mathtt{z})) \wedge f(\mathtt{x})=\mathtt{y} \rightarrow [X]f(\mathtt{x}) = \mathtt{y}$  \smallskip

  \item[\textbf{A1.3}] If $f$ is a bridge function symbol, \\
$\mathrm{conUSet}(X) \wedge \forall \mathtt{z} (\neg X(c_f,x,\mathtt{z})) \wedge f(x)=\mathtt{y} \rightarrow [X]f(x) = \mathtt{y}$  \smallskip

  \item[\textbf{A2.1}] If $f$ is a database function symbol, \\
$\mathrm{conUSet}(X) \wedge X(c_f,x,y) \rightarrow [X]f(x)=y$\smallskip

  \item[\textbf{A2.2}] If $f$ is an algorithmic function symbol, \\
$\mathrm{conUSet}(X) \wedge X(c_f,\mathtt{x},\mathtt{y}) \rightarrow [X]f(\mathtt{x})=\mathtt{y}$\smallskip

  \item[\textbf{A2.3}] If $f$ is a bridge function symbol, \\
$\mathrm{conUSet}(X) \wedge X(c_f,x,\mathtt{y}) \rightarrow [X]f(x)=\mathtt{y}$\smallskip

%  \item[\textbf{\emph{A3}}] upd$(r,X)\rightarrow $ def$(r)$  \smallskip

  \item[\textbf{A3}] $\mathrm{upm}(r,X) \rightarrow \exists  Y (\mathrm{upd}(r,Y))$  \smallskip

\end{description}

\item The following are axiom schemes from first-order logic.

\begin{description}
\item[\textbf{\emph{P1}}] $\varphi\rightarrow(\psi\rightarrow\varphi)$\smallskip

\item[\textbf{\emph{P2}}] $(\varphi\rightarrow(\psi\rightarrow\chi))\rightarrow((\varphi\rightarrow\psi)\rightarrow (\varphi\rightarrow\chi))$\smallskip

\item[\textbf{\emph{P3}}]
$(\neg\varphi\rightarrow\neg\psi)\rightarrow(\psi\rightarrow\varphi)$\smallskip

\end{description}

\item The standard axiom of universal instantiation~\textbf{UI} of the classical first-order calculus needs to be restricted to static terms (which do not contain dynamic function names). Otherwise, if we substitute a term $t$ for a variable $x$, then $t$ can be evaluated in different states due to sequential composition of transition rules. 
The rule of universal generalization $\textbf{UG}$ is the same as in the classical first-order calculus and applies to both types of first-order variables. An analogous axiom and inference rule can be added for $\exists$. This however is not necessary since in this paper $\exists$ is viewed as an abbreviation of $\neg \forall \neg$.

\begin{description}

\item[\textbf{UI}] $\forall v (\varphi(v)) \rightarrow \varphi[t/v]$ if $\varphi$ is pure or $t$ is static, $t$ is a database term or an algorithmic term depending on whether $v$ is a first-order variable in ${\cal X}_{db}$ or ${\cal X}_{a}$, respectively, and $t$ is free for $v$ in $\varphi(v)$. \smallskip

\item[\textbf{UG}] $\frac{\psi \rightarrow \varphi[v'/v]}{\psi \rightarrow \forall v (\varphi)}$ if $v$ and $v'$ are first-order variables of a same type, i.e., both belong to ${\cal X}_{db}$ or both belong to ${\cal X}_{a}$, and $v'$ is not free in $\psi$. \smallskip

\end{description}

\item  The following are the equality axioms adapted from first-order logic with equality.
Axiom~\textbf{EQ1} asserts the reflexivity property,
Axiom~\textbf{EQ2} asserts the substitutions for functions,
Axiom \textbf{EQ3} asserts the substitutions for second-order variables, 
and Axiom~\textbf{EQ4} asserts the substitutions for $\rho$-terms. Again,
terms occurring in the axioms are restricted to be static, which do
not contain any dynamic function symbols.

\begin{description}

\item[\textbf{EQ1}] $t=t$ for static term $t$. \smallskip

\item[\textbf{EQ2}] $t_1=t_{n+1}\wedge...\wedge t_n=t_{2n}\rightarrow f(t_1,...,t_n) = f(t_{n+1},...,t_{2n})$ for any
function $f$ and static terms $t_i$ $(i=1,...,2n)$.\smallskip

\item[\textbf{EQ3}] $t_1=t_{n+1}\wedge...\wedge t_n=t_{2n}\rightarrow (X(t_1,...,t_n) \leftrightarrow X(t_{n+1},...,t_{2n}))$ for any second-order variable $X$ and static terms $t_i$ $(i=1,...,2n)$.\smallskip

\item[\textbf{EQ4}] $t_1=t_2\wedge(\varphi_1 \leftrightarrow \varphi_2)\rightarrow \rho_{v}(t_1|\varphi_1)=\rho_{v}(t_2|\varphi_2)$ for pure formulae $\varphi_1$ and $\varphi_2$, static terms $t_1$ and $t_2$, and $v$ a first-order variable. \smallskip

\end{description}

\smallskip

\item The following axiom is taken from dynamic logic. It asserts that that executing a sequence rule is equivalent to executing its sub-rules sequentially.

\begin{description}

%\item[\textbf{\emph{DY1}}] \ $\exists X. \text{upd}(\textbf{seq}\; r_1\text{ }r_2 \;\textbf{endseq}, X) \wedge [X]\varphi \leftrightarrow \exists X_1 . \text{upd}(r_1,X_1) \wedge \exists X_2 . $ $\text{upd}(r_2,X_2) \wedge [X_1] [X_2] \varphi$
\item[\textbf{\emph{DY1}}] \ $\exists X (\text{upd}(\textbf{seq}\; r_1\text{ }r_2 \;\textbf{endseq}, X) \wedge [X]\varphi) \leftrightarrow$\\
\hspace*{4.4cm} $\exists X_1 (\text{upd}(r_1,X_1) \wedge [X_1]\exists X_2(\text{upd}(r_2,X_2) \wedge [X_2] \varphi))$
\end{description}

\item Axiom \textbf{E} is the extensionality axiom. Recall that $r_1\equiv r_2$ if for every Henkin meta-finite structure $S$ it holds that $S \models \forall X (\mathrm{upd}(r_1,X)\leftrightarrow \mathrm{upd}(r_2,X))$ (see Definition~\ref{def-equivalent-rules}).

\begin{description}
 \item[\textbf{\emph{E}}] $r_1\equiv r_2\rightarrow (\exists X_1 ($upd$(r_1,X_1)\wedge [X_1]\varphi)\leftrightarrow \exists X_2 ($upd$(r_2,X_2)\wedge [X_2]\varphi))$

\end{description}
\smallskip
\end{itemize}

The following soundness theorem for the proof system is relatively straightforward, since the non-standard axioms and rules are just a formalisation of the definitions of the semantics of rules, update sets and update multisets.

\begin{theorem}\label{c5-theoremsoundness}

Let $\varphi$ be a formula and let $\Phi$ be a set of formulae in the logic
${\cal L}^{db}$  for DB-ASMs. If $\Phi\vdash_{L_2}\varphi$, then $\Phi\models\varphi$.

\end{theorem}

\section{Derivation}\label{sec:soundness}

In this section we present some properties of the logic for DB-ASMs which are implied by the axioms and rules from the previous section. This includes some properties known for the logic for ASMs~\cite{RobertLogicASM}. In particular, the logic for ASMs uses the modal expressions $[r]\varphi$ and
$\langle r\rangle\varphi$ with the following semantics:

\begin{itemize}

\item \hspace{0.2cm} [\![$[r]\varphi]\!]_{S,\zeta}=\begin{cases}
true &\text{if } [\![\varphi]\!]_{S+U,\zeta}=true \text{ for
all consistent }U\in \Delta(r,S,\zeta), \\
false &\text{otherwise} \end{cases}$

\item \hspace{0.2cm} [\![$\langle r\rangle\varphi]\!]_{S,\zeta}=\begin{cases}
true &\text{if } [\![\varphi]\!]_{S+U,\zeta}=true \text{ for
at least one consistent } \\
&U\in \Delta(r,S,\zeta),\\ false &\text{otherwise} \end{cases}$

\end{itemize}

Instead of introducing modal operators $[\hspace{0.1cm}]$ and
$\langle\hspace{0.1cm}\rangle$ for a DB-ASM rule $r$, we use the
modal expression $[X]\varphi$ for an update set yielded by a
possibly non-deterministic rule. The modal expressions $[r]\varphi$
and $\langle r\rangle\varphi$ in the logic for ASMs can be treated
as the shortcuts for the following formulae in our logic.
\begin{equation}\label{ASM1}
[r]\varphi \equiv \forall X (\text{upd}(r,X)\rightarrow[X]\varphi).
\end{equation}
\begin{equation}\label{ASM2}
\langle r\rangle\varphi \equiv\exists
X (\text{upd}(r,X)\wedge[X]\varphi).
\end{equation}

\begin{lemma}\label{lem-modalAxiomsASMs}

The following axioms and rules used in the logic for ASMs are derivable in the logic for DB-ASMs, where the rule $r$ in Axioms (c) and (d) is assumed to be defined and deterministic.

\begin{description}

  \item[(a)] $([r](\varphi\rightarrow\psi)\wedge [r]\varphi)\rightarrow[r]\psi$\smallskip

  \item[(b)] $\varphi\rightarrow[r]\varphi$, for static and pure $\varphi$. \smallskip

  \item[(c)] $\neg$wcon$(r)\rightarrow [r]\varphi$\smallskip

  \item[(d)] $[r]\varphi\leftrightarrow \neg[r]\neg\varphi$\smallskip

\end{description}
\end{lemma}

\begin{proof} We can prove them as follows:
\begin{itemize}
  \item (a): By Equation~\ref{ASM1}, we have that $[r](\varphi\rightarrow\psi)\wedge [r]\varphi \equiv \forall X (\mathrm{upd}(r,X)\rightarrow [X](\varphi\rightarrow\psi)) \wedge \forall X (\mathrm{upd}(r,X)\rightarrow [X]\varphi)$. By the axioms from classical logic, this is in turn equivalent to $\forall X (\mathrm{upd}(r,X)\rightarrow([X](\varphi\rightarrow\psi)\wedge[X]\varphi))$.
Then by Axiom~\textbf{M1}, we get $\forall X (\mathrm{upd}(r,X)\rightarrow(([X]\varphi \rightarrow [X]\psi) \wedge [X]\varphi))$. Finally, by rule~\textbf{M3} (Modus Ponens), we derive $\forall X (\mathrm{upd}(r,X)\rightarrow [X]\psi)$, which by Equation~\ref{ASM1} is equivalent to $[r]\psi$ in the logic for ASMs.
  \item (b): By Rule~\textbf{M7}, we have that $\varphi\rightarrow \mathrm{upd}(r, X) \wedge [X]\varphi$, for
  static and pure $\varphi$. By the universal generalization rule for second-order variables (\textbf{SO-UG}), we obtain $\varphi\rightarrow \forall X(\mathrm{upd}(r, X) \wedge [X]\varphi)$. Finally, Equation~\ref{ASM1} gives us $\varphi\rightarrow [r]\varphi$ for static and pure $\varphi$.
  \item (c): By Equation~\ref{wcon}, we have $\neg \mathrm{wcon}(r)\leftrightarrow \neg\exists
  X (\mathrm{con}(r,X))$. In turn, by Equation~\ref{conr}, we get $\neg \mathrm{wcon}(r) \leftrightarrow \neg \exists X (\mathrm{upd}(r,X)\wedge \mathrm{conUSet}(X))$.
  Since a rule $r$ in the logic for ASMs is deterministic, we get $\neg \mathrm{wcon}(r)\leftrightarrow \neg \mathrm{conUSet}(X)$. By Axiom \textbf{M4}, we get $\neg \mathrm{wcon}(r)\rightarrow
  [r]\varphi$.
  \item (d): By Equation~\ref{ASM1}, we have $\neg[r]\neg\varphi \equiv \exists X(\mathrm{upd}(r,X)\wedge \neg[X]\neg\varphi)$. By applying Axiom~\textbf{M5} to $\neg[X]\neg\varphi$, we get
$\neg[r]\neg\varphi \equiv\exists X(\mathrm{upd}(r,X)\wedge[X]\varphi)$. When the rule $r$ is deterministic, the interpretation of $\forall X (\mathrm{upd}(r,X)\rightarrow [X]\varphi)$ coincides the interpretation of $\exists X (\mathrm{upd}(r,X)\wedge [X]\varphi)$ and therefore $[r]\varphi\leftrightarrow \neg[r]\neg\varphi$.
\end{itemize}
\end{proof}

The logic for ASMs introduced in \cite{RobertLogicASM} is deterministic, i.e., it excludes nondeterministic choice rules.
In contrast, our logic for DB-ASMs includes a nondeterministic choice rule. Note that
the formula Con$(R)$ in \textbf{Axiom 5} in \cite{RobertLogicASM} (i.e.,
in $\neg$Con$(R)\rightarrow [R]\varphi$) corresponds to the weak version of consistency (i.e., wcon$(r)$)
in the context of our logic for DB-ASMs.

\begin{lemma}\label{lem-soundness-modaloperator}
The following properties are derivable in the logic for DB-ASMs.
\begin{description}
  \item[(e)] $\mathrm{con}(r,X) \wedge [X]f(v_1)=v_2\rightarrow X(c_f,v_1,v_2)\vee(\forall v_3 (\neg X(c_f,v_1,v_3)) \wedge f(v_1)=v_2)$, where $v_1, v_2, v_3 \in {\cal X}_{db}$ if $f$ is a database function symbol, $v_1, v_2, v_3 \in {\cal X}_{a}$ if $f$ is an algorithmic function symbol, and $v_1 \in {\cal X}_{db}$ and $v_2, v_3 \in {\cal X}_{a}$ if $f$ is a bridge function symbol\smallskip
  \item[(f)] $\mathrm{con}(r,X) \wedge [X]\varphi \rightarrow \neg[X]\neg\varphi$\smallskip
  \item[(g)] $[X]\exists v (\varphi) \rightarrow \exists v ([X]\varphi)$, where $v \in {\cal X}_{db} \cup {\cal X}_{a}$ is a first-order variable.\smallskip
  \item[(h)] $[X]\varphi_1 \wedge [X]\varphi_2 \rightarrow [X] (\varphi_1 \wedge \varphi_2)$\smallskip
\end{description}

\end{lemma}

\begin{proof}
(e) is derivable by applying Axioms~\textbf{A1} and~\textbf{A2}. (f)
is a straightforward result of Axiom~\textbf{M5}. (g) can be derived
by applying Axioms~\textbf{M5} and~\textbf{M6}. Regarding (h), it is
derivable by using Axioms \textbf{M1}-\textbf{M3}.
\end{proof}

\begin{lemma}For arbitrary terms $t, s$ and first-order variables $v_1, v_2$ of the appropriate type (depending on whether $f$ is a database, algorithmic or bridge function symbol), the following properties in \cite{GroenboomFLEA95} are derivable in the logic for DB-ASMs.
\begin{itemize}
  \item $v_1=t \rightarrow (v_2=s \leftrightarrow [f(t):=s] f(v_1)=v_2)$
  \item $v_1 \neq t \rightarrow (v_2 = f(v_1) \leftrightarrow[f(t):=s]f(v_1)=v_2)$
\end{itemize}
\end{lemma}

In DB-ASMs, two parallel computations may produce an update multiset,
in which there are identical updates to a location assigned with a
location operator. Without an outer \textbf{let} rule, the rule \textbf{par} $ r \; r $ \textbf{endpar}
could be simplified to $r$. This however is no longer the case if we consider update multisets.

\begin{example}

In the DB-ASM rule below, \textsc{sum} is a location operator assigned to the location (\textsc{tnum},()). Two identical updates (i.e., (\textsc{tnum},(),1) and (\textsc{tnum},(),1)) are first
generated in an update multiset, and then aggregated into one update (\textsc{tnum},(),2) in an update set.

\vspace{3.5cm}
 \hspace{2.5cm}\textbf{let} $(\textsc{tnum},())
\!\rightharpoonup\! \textsc{sum}$ \textbf{in}

\hspace{3cm}\textbf{par}

\hspace{3.5cm} $\textsc{tnum}:=1$

\hspace{3.5cm} $\textsc{tnum}:= 1$

\hspace{3cm}\textbf{endpar}

\hspace{2.5cm}\textbf{endlet}\\[0.2cm]
The update multiset is collapsed into the update set $\{ (\textsc{tnum},(),2) \}$, whereas without the \textbf{let} rule we would obtain $\{ (\textsc{tnum},(),1) \}$.

\end{example}

Following the approach of defining the predicate joinable in
\cite{RobertLogicASM}, we define the predicate joinable over two
DB-ASM rules. As DB-ASM rules are allowed to be nondeterministic,
the predicate joinable$(r_1,r_2)$ means that there exists a pair of
update sets without conflicting updates, which are yielded by rules
$r_1$ and $r_2$, respectively. Then, based on the use of predicate
joinable, the properties in Lemma \ref{lem-soundness-consistency}
are all derivable.

\begin{equation}\label{joinable}
\begin{split}
\text{joinable}(r_1,r_2) \equiv &  \exists X_1 X_2 \big(\mathrm{upd}(r_1,X_1)\wedge\mathrm{upd}(r_2,X_2)\wedge\hspace{2cm}  \\
  &  \bigwedge\limits_{c_f\in\mathcal{F}_{dyn} \wedge f \in \Upsilon_{db}}\forall x y z (X_1(c_f,x,y)\wedge X_2(c_f,x,z)\rightarrow y=z ) \wedge\\
  &  \bigwedge\limits_{c_f\in\mathcal{F}_{dyn} \wedge f \in \Upsilon_{a}}\forall \mathtt{x} \mathtt{y} \mathtt{z} (X_1(c_f,\mathtt{x},\mathtt{y})\wedge X_2(c_f,\mathtt{x},\mathtt{z})\rightarrow \mathtt{y}=\mathtt{z} ) \wedge\\
  &  \bigwedge\limits_{c_f\in\mathcal{F}_{dyn} \wedge f \in {\cal F}_b}\forall x \mathtt{y} \mathtt{z} (X_1(c_f,x,\mathtt{y})\wedge X_2(c_f,x,\mathtt{z})\rightarrow \mathtt{y}=\mathtt{z} ) \big)
\end{split}
\end{equation}

\begin{lemma}\label{lem-soundness-consistency}
The following properties for weak consistency are derivable in the
logic of DB-ASMs.
\begin{description}
  \item[(i)] $\mathrm{wcon}(f(t):=s)$

  \item[(j)] $\mathrm{wcon}(\textbf{if}\, \varphi\, \textbf{then}\, r \,\textbf{endif}) \leftrightarrow \neg\varphi\vee(\varphi\wedge \mathrm{wcon}(r))$

  \item[(k)] $\mathrm{wcon}(\textbf{forall} \, x \, \textbf{with} \, \varphi \, \textbf{do} \, r \, \textbf{enddo}) \leftrightarrow$\\
    \hspace*{4.5cm}$\forall x (\varphi\rightarrow \mathrm{wcon}(r) \wedge \forall y (\varphi[y/x]\rightarrow \text{joinable}(r,r[y/x])))$
\item[(l)] $\mathrm{wcon}(\textbf{par} \, r_1 \, r_2 \, \textbf{endpar}) \leftrightarrow \mathrm{wcon}(r_1) \wedge \mathrm{wcon}(r_2) \wedge joinable(r_1,r_2)$
  \item[(m)] $\mathrm{wcon}(\textbf{choose} \, x \, \textbf{with} \, \varphi \, \textbf{do} \, r \, \textbf{enddo}) \leftrightarrow \exists x (\varphi\wedge \mathrm{wcon}(r))$
  \item[(n)] $\mathrm{wcon}(\textbf{seq} \, r_1 \, r_2 \, \textbf{endseq}) \leftrightarrow \exists X (\mathrm{con}(r_1,X) \wedge [X]\mathrm{wcon}(r_2))$
  \item[(o)] If $f$ is a bridge function symbol: \\
$\mathrm{wcon}(\textbf{let} \, (f,t) \!\rightharpoonup\! \rho \, \textbf{in} \, r \,\textbf{endlet}) \leftrightarrow$\\
\hspace*{1.3cm}$\exists X Y (\mathrm{upd}(r,X) \wedge \mathrm{conUSet}(Y) \wedge \forall x \mathtt{y} (Y(c_f,x,\mathtt{y}) \leftrightarrow (t = x \vee X(c_f,x,\mathtt{y}))) \wedge $\\
\hspace*{2cm}   $\forall \mathtt{z} x y (X(\mathtt{z},x,y) \leftrightarrow Y(\mathtt{z},x,y)) \wedge \forall \mathtt{z} \mathtt{x} \mathtt{y} (X(\mathtt{z},\mathtt{x},\mathtt{y}) \leftrightarrow Y(\mathtt{z},\mathtt{x},\mathtt{y})))$
\end{description}
\end{lemma}
We omit the proof of the previous lemma as well as the proof of the remaining lemmas in this section, since they are lengthy but relatively easy exercises. Furthermore, most of them are similar to the proofs of the analogous results in Nanchen's thesis \cite{Nanchen07}. 

\begin{lemma}\label{lem-soundness-rule}The following properties for the formula $[r]\varphi$ are derivable in the logic for DB-ASMs.
\begin{description}
  \item[(p)] $[\textbf{if} \, \varphi \, \textbf{then} \, r \, \textbf{endif}]\psi \leftrightarrow (\varphi\wedge[r]\psi) \vee (\neg \varphi \wedge \psi)$\smallskip
  \item[(q)] $[\textbf{choose} \, x \, \textbf{with} \, \varphi \, \textbf{do} \, r \, \textbf{enddo}]\psi \leftrightarrow \forall x (\varphi\rightarrow [r]\psi)$\smallskip
\end{description}
\end{lemma}
Lemma \ref{lem-soundness-composition} states that a parallel
composition is commutative and associative while a sequential
composition is associative.

\begin{lemma}\label{lem-soundness-composition}The following properties for parallel and sequential compositions are derivable in the logic for DB-ASMs.
\begin{description}
  \item[(r)] \textbf{par} $r_1\hspace{0.2cm} r_2$ \textbf{endpar} $\equiv$ \textbf{par} $r_2\hspace{0.2cm} r_1$ \textbf{endpar}\smallskip
  \item[(s)] \textbf{par} (\textbf{par} $r_1\hspace{0.2cm} r_2$ \textbf{endpar}) $r_3$ \textbf{endpar} $\equiv$ \textbf{par} $r_1$ (\textbf{par} $r_2\hspace{0.2cm} r_3$
  \textbf{endpar}) \textbf{endpar}\smallskip
  \item[(t)] \textbf{seq} (\textbf{seq} $r_1\hspace{0.2cm} r_2$ \textbf{endseq}) $r_3$ \textbf{endseq} $\equiv$ \textbf{seq} $r_1$ (\textbf{seq} $r_2\hspace{0.2cm} r_3$
  \textbf{endseq}) \textbf{endseq}\smallskip
\end{description}
\end{lemma}

\begin{lemma}The extensionality axiom for transition rules in the logic for ASMs~\cite{RobertLogicASM} is derivable in the logic for DB-ASMs.
\begin{description}

\item[(u)] $r_1\equiv r_2\rightarrow([r_1]\varphi\leftrightarrow [r_2]\varphi)$
\end{description}
\end{lemma}

\section{Completeness}\label{sec:completeness}

In this section we prove the completeness of the proof system of the logic ${\cal L}^{db}$ for DB-ASMs which we introduced in the previous section.

In the following, $S$ denotes an arbitrary Henkin meta-finite structure of signature (of meta-finite states) $\Upsilon = \Upsilon_{db} \cup \Upsilon_a \cup {\cal F}_b$ (recall Definition~\ref{HenkinStructure}. As before, $B = B_{db} \cup B_a$ denotes the base set (domain) of individual of $S$, where $B_{db}$ and $B_a$ are the base sets of the database and algorithmic parts, respectively, and $D_n$ the universe of $n$-ary relations. 

Clearly, we cannot axiomatize an arbitrary set $\Lambda = \{\rho^1, \ldots, \rho^m\}$ of location operators. Note that even if we just take a simple location operator such as $\mathrm{PRODUCT}$ and axiomatize it, that leads us outside linear arithmetic and thus to an incomplete theory. 
As a compromise solution for this problem, we treat location (multiset) operators as standard non-axiomatized functions as follows. 

\begin{definition}\label{multisetOpAsFunc}
We assume that $\Upsilon$ further includes a subset $\Upsilon_\Lambda = \{f_{\rho^1}, \ldots, f_{\rho^m}\}$ of static function symbols, where each $f_{\rho^i}$ is interpreted in $S$ by a corresponding function $f^{S}_{\rho^i}: D_2 \rightarrow B_a$ defined as follows:  \[f^{S}_{\rho^i}(A) = \begin{cases} \rho^i(\{\!\!\{a \mid (a,b) \in A\}\!\!\}) & \text{if } \{\!\!\{a \mid (a,b) \in A\}\!\!\} \in \textit{dom}(\rho_i)  \\ \textit{undef} & \text{otherwise.} \end{cases}\] 
\end{definition}

We then assume that the formulae of ${\cal L}^{db}$ do not include any $\rho$-term of the form $\rho_{v}(t|\varphi)$. This does not affect the expressive power of ${\cal L}^{db}$ since every formula $\varphi$ can be translated (under the assumption made in Definition~\ref{multisetOpAsFunc}) into an equivalent formula $\varphi'$ which does not use any $\rho$-term. We can proceed as follows. Let $\rho^1_{v_1}(t_1|\psi_1), \ldots, \rho^n_{v_n}(t_n|\psi_n)$ be the $\rho$-terms which appear in an atomic sub-formula $\alpha$ of $\varphi$. Let $\alpha'$ be the following formula:
\[\forall X_1 \ldots X_n v_1 \ldots v_n \mathtt{z}_1 \ldots \mathtt{z}_n \Big( \bigwedge_{1 \leq i \leq n} \big(X_i(\mathtt{z}_i, v_i) \leftrightarrow (\psi_i' \wedge (t_i = \texttt{z}_i)')\big) \rightarrow \alpha'' \Big),\]
where $X_1, \ldots, X_n, v_1, \ldots, v_n, \mathtt{z}_1, \ldots, \mathtt{z}_n$ are variables which do not appear free in $\alpha$, $\psi_i'$ and $(t_i = \texttt{z}_i)'$ are obtained by recursively applying this procedure to every atomic sub-formula of $\psi_i$ and to $t_i = \mathtt{z}_i$, respectively, and $\alpha''$ is obtained by replacing $\rho^1_{v_1}(t_1|\psi_1), \ldots, \rho^n_{v_n}(t_n|\psi_n)$ in $\alpha$ by $f_{\rho^1}(X_1), \ldots,f_{\rho^n}(X_n)$, respectively.  Then the formula $\varphi'$ can be defined as the formula obtained by replacing every atomic sub-formula $\alpha$ of $\varphi$ by $\alpha'$.

We can now proceed with proving the completeness of ${\cal L}^{db}$. The strategy is to show that ${\cal L}^{db}$ (with $\rho$-terms conveniently replaced by functions as explained above) is a syntactic variant of a complete first-order theory of types. 

Let $\Upsilon = \Upsilon_{db} \cup \Upsilon_a \cup {\cal F}_b \cup \Upsilon_{\Lambda}$ be a signature of Henkin meta-finite structures. Assume w.l.o.g. that $\Upsilon_{db}$, $\Upsilon_a$, ${\cal F}_b$ and $\Upsilon_\Lambda$ are pairwise disjoint. Let $\Upsilon^T$ be the signature formed by:
\begin{itemize}
\item The function symbols of $\Upsilon$. 
\item Unary relations $T_{db}$ and $T_a$.
\item For each $n \geq 1$ a $1$-ary relation symbol $T_n$. 
\item For each $n \geq 1$ a $(n+1)$-ary relation symbol $E_n$. 
\end{itemize} 
$T_{db}(x)$ and $T_a(x)$ are intended to state that $x$ is an individual belonging to the database part and to the algorithmic part, respectively. Likewise, $T_n(x)$ is intended to state that $x$ is a relation of arity $n$. Finally, $T_n(y_1, \ldots, y_n, x)$ is intended to state that the tuple $(y_1, \ldots, y_n)$ belongs to the relation $x$.  

A Henkin meta-finite structure $S$ of signature $\Upsilon$ \emph{determines} a unique first-order structure $S'$ of vocabulary $\Upsilon^T$ as follows:

\begin{itemize}
\item The domain of $S'$ is $\mathit{dom}(S') = B_{db} \cup B_{a} \cup \bigcup_{n \geq 1} D_n$, where $B_{db}$ and $B_a$ denote the base sets of the database and algorithmic parts of $S$, respectively, and $D_n$ denotes the universe of $n$-ary relations of $S$. 
\item The interpretation in $S'$ of the function symbols in $\Upsilon^T \cap (\Upsilon_{db} \cup {\cal F}_b)$ is the same as their interpretation in $S$ for arguments in $\mathit{dom}(S') \cap B_{db}$ and it is extended arbitrarily to arguments in $\mathit{dom}(S') \setminus B_{db}$. Likewise, the interpretation in $S'$ of the function symbols in $\Upsilon^T \cap \Upsilon_{a}$ is the same as their interpretation in $S$ for arguments in $\mathit{dom}(S') \cap B_{a}$ and it is extended arbitrarily to arguments in $\mathit{dom}(S') \setminus B_{a}$.
\item The interpretation in $S'$ of function symbols in $\Upsilon^T \cap \Upsilon_\Lambda$ is as per Definition~\ref{multisetOpAsFunc} for arguments in $D_2$ and it is extended arbitrarily to arguments in $\mathit{dom}(S') \setminus D_2$. 
\item $T_{db}$ is interpreted as $B_{db}$ and $T_a$ as $B_{a}$.
\item For every $n \geq 1$, $T_n$ is interpreted as $D_n$ and $E_n$ as set membership restricted to $n$-tuples.
\end{itemize}

Each ${\cal L}^{db}$-formula $\varphi$ of signature $\Upsilon$ can be rewritten as a first-order formula $\varphi^*$ of signature $\Upsilon^T$, where $\varphi^*$ is obtained from $\varphi$ by applying the following steps:
\begin{enumerate}
\item Replace each atomic formula of the form $\mathrm{upd}(r, X)$ and $\mathrm{upm}(r, X)$ by their corresponding definitions using the Axioms~$\mathbf{U1}$--$\mathbf{U7}$ and~$\mathbf{\ddot{U}1}$--$\mathbf{\ddot{U}7}$, respectively. 
\item Bring all remaining atomic formulae into the form $v_1 = v_2$, $f(v_2) = v_1$ or $X(v_1, \ldots, v_n)$ (where each $v_i$ denotes an appropriate first-order variable $x_i$ or $\mathtt{x}_i$ depending on the context) by applying the following equivalences:
\begin{align*}
s=t & \quad \leftrightarrow \quad \exists v_1 (s = v_1 \wedge t = v_1)\\
X(t_1, \ldots, t_n) & \quad \leftrightarrow \quad \exists v_1 \ldots v_n (t_1 = v_1 \wedge \cdots \wedge t_n = v_n \wedge X(v_1, \ldots, v_n))\\
f(s) = v_1 & \quad \leftrightarrow \quad \exists v_2 (s = v_2 \wedge f(v_2) = v_1) 
\end{align*}
\item Eliminate all modal operators by applying the following equivalences (again where each $v_i$ denotes an appropriate first-order variable $x_i$ or $\mathtt{x}_i$ depending on the context):
\begin{align*}
[X]v_1=v_2 & \quad \leftrightarrow \quad (\mathrm{IsUSet}(X) \wedge \mathrm{conUSet}(X)\rightarrow v_1=v_2) \\ 
[X]Y(v_1, \ldots, v_n) & \quad \leftrightarrow \quad (\mathrm{IsUSet}(X) \wedge \mathrm{conUSet}(X)\rightarrow Y(v_1, \ldots v_n)) \\
[X]f(v_2) = v_1 & \quad \leftrightarrow \quad (\mathrm{IsUSet}(X) \wedge \mathrm{conUSet}(X)\rightarrow \\
& \qquad \qquad X(c_f, v_2, v_1) \vee \forall v_3 (\neg X(c_f, v_2, v_3) \wedge f(v_2) = v_1))\\
[X]\neg \varphi & \quad \leftrightarrow \quad (\mathrm{IsUSet}(X) \wedge \mathrm{conUSet}(X)\rightarrow \neg [X]\varphi) \\
[X](\varphi \vee \psi) & \quad \leftrightarrow \quad ([X]\varphi \vee [X]\psi) \\
[X]\forall v (\varphi)& \quad \leftrightarrow \quad \forall v ([X] \varphi)\\
[X]\forall Y (\varphi)& \quad \leftrightarrow \quad \forall Y ([X] \varphi)
\end{align*}
\item Replace each atomic formula of the form $X(t_1, \ldots, t_n)$ by $E(t_1, \ldots, t_n, X)$, and relativise quantifiers over individuals in $B_{db}$ to $T_{db}$, quantifiers over individuals in $B_a$ to $T_a$, and quantifiers over $n$-ary relations in $D_n$ for some $n \geq 1$ to $T_n$. More precisely, $\varphi^*$ is obtained by the recurrent application of the following rules to the formula $\varphi$ obtained after applying steps $1$--$3$. 
\begin{align*}
(v_1=v_2)^* & \quad = \quad v_1=v_2 \\ 
(Y(v_1, \ldots, v_n))^* & \quad = \quad E(v_1, \ldots, v_n, Y) \\
(f(v_2) = v_1)^* & \quad = \quad f(v_2) = v_1 \\
(\neg \varphi)^* & \quad = \quad \neg (\varphi^*) \\
(\varphi \vee \psi)^* & \quad = \quad (\varphi^* \vee \psi^*) \\
(\forall x (\varphi))^*& \quad = \quad \forall x (T_{db}(x) \rightarrow  \varphi^*)\\
(\forall \mathtt{x} (\varphi))^*& \quad = \quad \forall \mathtt{x} (T_{a}(\mathtt{x}) \rightarrow  \varphi^*)\\
(\forall X (\varphi))^*& \quad = \quad \forall X (T_{n}(X) \rightarrow  \varphi^*) \qquad \text{(if $X$ has arity $n$)}\\
\end{align*}
\end{enumerate}  

It is then relatively easy to prove:

\begin{lemma}\label{translationLemma}
A ${\cal L}^{db}$-formula $\varphi$ is true in $S$ iff $\varphi^*$ is true in $S'$.
\end{lemma}

Thus, if $\varphi^*$ is valid, then $\varphi$ is true in all Henkin meta-finite structures. Note that the converse does not always holds. For example, $\exists \mathtt{x} (\mathtt{x} = \mathtt{x})$ is true in all Henkin meta-finite structure (note that by the Background Postulate $B_a$ is not empty), but $\exists x (T_a(x) \wedge (x = x))$ is not valid. In general, \emph{not every} $\Upsilon^T$-structure is an $S'$ structure for some Henkin meta-finite $\Upsilon$-structure $S$.  Indeed, each $\Upsilon^T$-structure $S'$ which \emph{does} correspond to some Henkin meta-finite structure $S$ satisfies the following properties (cf.~\cite{Leivant94}):
\begin{enumerate}
\item $\Upsilon$-correctness: 
\begin{itemize}
\item $f(x_1) = x_2 \rightarrow T_{db}(x_1) \wedge T_{db}(x_2)$ for every $f \in \Upsilon_{db}$, 
\item $f(x_1) = x_2 \rightarrow T_{db}(x_1) \wedge T_{a}(x_2)$ for every $f \in {\cal F}_b$,
\item $f(x_1) = x_2 \rightarrow T_{a}(x_1) \wedge T_{a}(x_2)$ for every $f \in \Upsilon_{a}$,
\item $f(x_1) = x_2 \rightarrow T_2(x_1) \wedge T_{a}(x_2)$ for every $f \in \Upsilon_{\Lambda}$. 
\end{itemize}
\item Non-emptiness: $\exists x (T_{a}(x))$.
\item Disjointness: $T_i(x) \rightarrow \neg T_j(x)$ for every $i, j \in \mathbb{N} \cup \{a, db\}$ such that  $i \neq j$.
\item Elementhood: $E_n(x_1, \ldots, x_n, y) \rightarrow T_n(y) \wedge (T_{db}(x_1) \vee T_a(x_1)) \wedge \cdots \wedge (T_{db}(x_n) \vee T_a(x_n))$ for every $n \geq 1$.
\item Extensionality: $T_n(x) \wedge T_n(y) \wedge \forall \bar{z}(E_n(\bar{z}, x) \leftrightarrow  E_n(\bar{z}, y)) \rightarrow x = y$ for every $n \geq 1$.
\item Comprehension: $\exists y \forall \bar{x} (E_n(\bar{x}, y) \leftrightarrow \psi)$ for every $n \geq 1$ and $y$ non-free in $\psi$.  
\end{enumerate}
  
\begin{lemma}\label{hasHenkinS}
If $A$ is a first-order structure of signature $\Upsilon^T$ which satisfies properties~1--6 above and $\mathit{sub}(A)$ is the sub-structure of $A$ generated by the elements of $(T_{db})^A \cup (T_a)^A \cup \bigcup_{n \geq 1} (T_n)^A$, then $\mathit{sub}(A) = S'$ for some Henkin meta-finite structure $S$ of signature $\Upsilon$ and corresponding first-order structure $S'$ of signature $\Upsilon^T$ determined by $S$. 
\end{lemma}

\begin{proof}
Given $A$ with domain $\mathit{dom}(A)$, we define $S$ as follows:
\begin{itemize}
\item $B_{db} = (T_{db})^A$ is the base set of the database part of $S$.
\item $B_a = (T_a)^A$ is the base set of the algorithmic part of $S$.
\item For each $n \geq 1$, the universe $D_n$ of $n$-ary relations consists of the sets $\{\bar{a} \in (B_{db} \cup B_a)^n \mid (E_n)^A(\bar{a}, s)\}$ for all $s \in (T_n)^A$. 
\item The interpretation of each function symbol $f \in \Upsilon$ is the same as in $A$ but restricted to arguments from $B_{db}$, $B_a$ or $D_2$ depending on whether $f$ belongs to $\Upsilon_{db}$, $\Upsilon_a \cup {\cal F}_b$ or $\Upsilon_\Lambda$, respectively.
\end{itemize}    
By the $\Upsilon$-correctness, non-emptiness and comprehension properties of $A$, we get that $S$ is a Henkin meta-finite structure. 

We claim that $\mathit{sub}(A)$ is isomorphic to $S'$ via $g: \mathit{dom}(S') \rightarrow \mathit{dom}(\mathit{sub}(A))$ where 
\[g(x) = \begin{cases} x & \text{if } x \in (T_{db})^{S'} \cup (T_{a})^{S'}\\
\{\bar{a} \in ((T_{db})^{S'} \cup (T_a)^{S'})^n \mid (E_n)^{S'}(\bar{a}, x)\} & \text{if } x \in (T_n)^{S'}\end{cases}\]
First, we get that $g$ is well defined by the disjointness property and by the fact that, by definition of $S$ and $S'$, every element $x$ in $\mathit{dom}(S')$ is in $(T_{db})^{S'} \cup (T_{a})^{S'} \cup \bigcup_{n \geq 1} (T_n)^{S'}$. That $g$ is surjective follows from the definition of $S'$ from $A$ and the fact that $\mathit{dom}(\mathit{sub}(A))$ is the restriction of $\mathit{dom}(A)$ to $\mathit{dom}(S')$. By the extensionality property, we get that $g$ is injective. By definition we get that $g$ preserves the function symbols in $\Upsilon$ as well as the relation symbols $T_{db}$, $T_a$ and $T_n$ for every $n \geq 1$. Finally, for every $n \geq 1$, we get that $g$ preserves $E_n$ by the elementhood property.  
\end{proof}

Let $\Psi$ be the set of formulae listed under properties 1--6 above. We then get the following Henkin style completeness theorem.

\begin{theorem}\label{HenkinCompleteness}
A ${\cal L}^{db}$-formula $\varphi$ is true in all Henkin meta-finite structures iff $\varphi^*$ is derivable in first-order logic from $\Psi$ (i.e., iff $\Psi \vdash \varphi^*$).
\end{theorem}

\begin{proof}
Assume that $\Psi \vdash \varphi^*$, and let $S$ be a Henkin meta-finite structure. Then $S' \models \Psi$ and therefore $S' \models \varphi^*$. By Lemma~\ref{translationLemma}, we get that $S \models \varphi$. 

Conversely, assume that $\varphi$ is true in all Henkin meta-finite structures. Towards showing $\Psi \models \varphi^*$, let us assume that $A \models \Psi$, and let $\mathit{sub}(A)$ be its substructure generated by the elements of $(T_{db})^A \cup (T_a)^A \cup \bigcup_{n \geq 1} (T_n)^A$. Then by Lemma~\ref{hasHenkinS}, $\mathit{sub}(A) = S'$ for some first-order structure $S'$ determined by a Henkin meta-finite structure $S$. Since by assumption we have that $S \models \varphi$, it follows from Lemma~\ref{translationLemma} that $S' \models \varphi^*$ and therefore $\mathit{sub}(A) \models \varphi^*$. But each quantifier in $\varphi^*$ is relativised to $(T_{db})^A$, $(T_a)^A$ or $(T_n)^A$ for some $n\geq 1$, and then we also have that $A \models \varphi^*$.  We have shown that $\Psi \models \varphi^*$, and then, by the completeness theorem of first-order logic, we get that $\Psi \vdash \varphi^*$.
\end{proof}

We know from Theorem~\ref{c5-theoremsoundness} that the deductive calculus $L_2$ introduced in Section~\ref{sub:AxiomsRules} is sound. Thus, if $\varphi$ is a ${\cal L}^{db}$-formula derivable in $L_2$, then $\varphi$ is true in all Henkin meta-finite structures. It is then immediate from Theorem~\ref{HenkinCompleteness} that $\varphi^*$ is derivable in first-order logic from $\Psi$. On the other hand, it can be proven by an easy but lengthy induction on the length of the derivations that if $\varphi^*$ is derivable in first-order from $\Psi$, then $\varphi$ is derivable in $L_2$. 

\begin{lemma}\label{derivabilityLemma}
$\varphi^*$ is derivable in first-order from $\Psi$ iff $\varphi$ is derivable in $L_2$.
\end{lemma}

Finally, Theorem~\ref{HenkinCompleteness} and Lemma~\ref{derivabilityLemma} immediately imply that the logic ${\cal L}^{db}$ is complete to reason about DB-ASMs.  

\begin{theorem} \label{thm-adtm-completeness-FOL}
Let $\varphi$ be a ${\cal L}^{db}$-formula and $\Phi$ be a set of ${\cal L}^{db}$-formulae. If $\Phi\models\varphi$, then $\Phi\vdash_{L_2}\varphi$.
\end{theorem}

\section{Conclusions}\label{sec:conclusion}

This article presents a logic for DB-ASMs. In accordance with the
result that DB-ASMs and database transformations are behaviourally
equivalent, it thus represents a logical characterisation for
database transformations in general.

The logic for DB-ASMs is built upon the logic of
  meta-finite structures. The formalisation of
  multiset operations is captured by the notion of $\rho$-term.
% that permits an inductive construction between formulae and terms.
The use of $\rho$-terms greatly enhances the
  expressive power of the logic for DB-ASMs since aggregate computing
  in database applications can be easily expressed by using $\rho$-terms.
On the other hand, $\rho$-terms can easily lead to incompleteness if we try to axiomatize them in the proof system. We avoid this problem by considering them as non-interpreted functions. In this way, we cannot reason about properties of $\rho$-terms themselves, but we can still use them in the formulae of our complete proof system to express meaningful properties of DB-ASMs.

As discussed in~\cite{RobertLogicASM} and~\cite{[BS03]}, the non-determinism accompanied with the use of
  choice rules poses a further challenging problem. In this work, we realized that the update sets produced by non-deterministic DB-ASMs rules are 
  definable in a variant of second-order logic in which the second-order quantifiers are interpreted using a Henkin semantics, thus becoming part of the specification of a model rather than an invariant through all models as in the case of the classical second-order semantics. 
Base on these definitions, we use the  modal operator $[X]$ where $X$ is a second-order variable that represents an update set $U$ generated by a
 (possibly non-deterministic) DB-ASM rule $r$. By introducing $[X]$ into the logic for DB-ASMs, it is shown that nondeterministic database
  transformations can also be captured.

The use of a Henkin semantics in the definition of the logic ${\cal L}^{db}$ for DB-ASMs allowed us to show that ${\cal L}^{db}$ is actually a syntactic variant of a complete first-order theory of types. In turn, this allowed us to establish a sound and complete proof system for the logic for DB-ASMs, which can
be turned into a tool for reasoning about database transformations. However, this is restricted to reasoning about steps, not full runs, but no complete logic for reasoning about runs can be expected.
In the future we will continue to investigate how the logic for
DB-ASMs can be tailored towards different classes of database
transformations such as XML data transformations and used for
verifying the properties of database transformations in practice.

\bibliographystyle{plain}
\bibliography{DBTsLogic}

\begin{thebibliography}{10}

\bibitem{AbiteboulIQL89}
Serge Abiteboul and Paris~C. Kanellakis.
\newblock Object identity as a query language primitive.
\newblock In {\em Proceedings of the 1989 ACM SIGMOD international conference
  on Management of data}, pages 159--173. ACM Press, 1989.

\bibitem{AbiteboulUpdate87}
Serge Abiteboul and Victor Vianu.
\newblock A transcation language complete for database update and
  specification.
\newblock In Moshe~Y. Vardi, editor, {\em Proceedings of the Sixth ACM
  SIGACT-SIGMOD-SIGART Symposium on Principles of Database Systems}, pages
  260--268. ACM, 1987.

\bibitem{AbiteboulDatalogExtension}
Serge Abiteboul and Victor Vianu.
\newblock Datalog extensions for database queries and updates.
\newblock {\em J. Comput. Syst. Sci.}, 43(1):62--124, 1991.

\bibitem{beeri:fomlado1998}
Catriel Beeri and Bernhard Thalheim.
\newblock Identification as a primitive of data models.
\newblock In Torsten Polle, Torsten Ripke, and Klaus-Dieter Schewe, editors,
  {\em Fundamentals of Information Systems}, pages 19--36. Kluwer Academic
  Publishers, Boston Dordrecht London, 1999.

\bibitem{blass:tocl2003}
Andreas Blass and Yuri Gurevich.
\newblock Abstract state machines capture parallel algorithms.
\newblock {\em ACM Transactions on Computational Logic}, 4(4):578--651, October
  2003.

\bibitem{GurevichParallelCorrection08}
Andreas Blass and Yuri Gurevich.
\newblock Abstract state machines capture parallel algorithms: Correction and
  extension.
\newblock {\em ACM Transactions on Computation Logic}, 9(3):1--32, 06 2008.

\bibitem{Yuri02onpolynomial}
Andreas Blass, Yuri Gurevich, and Saharon Shelah.
\newblock On polynomial time computation over unordered structures.
\newblock {\em The Journal of Symbolic Logic}, 67(3):1093--1125, September
  2002.

\bibitem{Survey98}
Anthony~J. Bonner and Michael Kifer.
\newblock The state of change: A survey.
\newblock In {\em International Seminar on Logic Databases and the Meaning of
  Change, Transactions and Change in Logic Databases}, pages 1--36.
  Springer-Verlag, 1998.

\bibitem{boerger:2003}
E.~B{\"o}rger and Robert~F. St{\"a}rk.
\newblock {\em Abstract State Machines: A Method for High-Level System Design
  and Analysis}.
\newblock Springer-Verlag New York, Inc., 2003.

\bibitem{[BS03]}
Egon B{\"{o}}rger and Robert~F. St{\"{a}}rk.
\newblock {\em Abstract State Machines. {A} Method for High-Level System Design
  and Analysis}.
\newblock Springer, 2003.

\bibitem{CaiIFP+C89}
J.-Y. Cai, M.~F\"urer, and N.~Immerman.
\newblock An optimal lower bound on the number of variables for graph
  identification.
\newblock In {\em Proceedings of the 30th Annual Symposium on Foundations of
  Computer Science}, pages 612--617. IEEE Computer Society, 1989.

\bibitem{ChandraCompleteness}
Ashok~K. Chandra and David Harel.
\newblock Computable queries for relational data bases.
\newblock {\em Journal of Computer and System Sciences}, 21(2):156--178, 1980.

\bibitem{FenselMLPM96}
D~Fensel and R~Groenboom.
\newblock {MLPM}: Defining a semantics and axiomatization for specifying the
  reasoning process of knowledge-based systems.
\newblock In {\em Proceedings of the 12th European Conference on Artificial
  Intelligence (ECAI-96)}, Budapest, Hungary, 1996.

\bibitem{FerrarottiRT14}
Flavio Ferrarotti, Wei Ren, and Jose Maria~Turull Torres.
\newblock Expressing properties in second- and third-order logic: hypercube
  graphs and {SATQBF}.
\newblock {\em Logic Journal of the {IGPL}}, 22(2):355--386, 2014.

\bibitem{FerrarottiSTW16}
Flavio Ferrarotti, Klaus{-}Dieter Schewe, Loredana Tec, and Qing Wang.
\newblock A new thesis concerning synchronised parallel computing - simplified
  parallel {ASM} thesis.
\newblock {\em Theor. Comput. Sci.}, 649:25--53, 2016.

\bibitem{graedel:infcomp1998}
E.~Gr{\"a}del and Y.~Gurevich.
\newblock {Metafinite model theory}.
\newblock {\em Information and Computation}, 140(1):26--81, 1998.

\bibitem{GrCounting}
Erich Gr\"{a}del and Martin Otto.
\newblock Inductive definability with counting on finite structures.
\newblock In {\em Selected Papers from the Workshop on Computer Science Logic},
  pages 231--247. Springer-Verlag, 1993.

\bibitem{GroenboomMLCM94}
R.~Groenboom and G.~{Renardel de Lavalette}.
\newblock Reasoning about dynamic features in specification languages - a modal
  view on creation and modification.
\newblock In {\em Proceedings of the International Workshop on Semantics of
  Specification Languages (SoSL)}, pages 340--355. Springer-Verlag, 1994.

\bibitem{GroenboomFLEA95}
R.~Groenboom and G.~{Renardel de Lavalette}.
\newblock A formalization of evolving algebras.
\newblock In {\em Proceedings of Accolade95}. Dutch Research School in Logic,
  1995.

\bibitem{Gurevich-New-Thesis}
Yuri Gurevich.
\newblock A new thesis (abstracts).
\newblock {\em American Mathematical Society}, 6(4):317, August 1985.

\bibitem{gurevich:tocl2000}
Yuri Gurevich.
\newblock Sequential abstract state machines capture sequential algorithms.
\newblock {\em ACM Transactions on Computational Logic}, 1(1):77--111, July
  2000.

\bibitem{gurevich2004abstract}
Yuri Gurevich.
\newblock Abstract state machines: An overview of the project.
\newblock In {\em International Symposium on Foundations of Information and
  Knowledge Systems}, pages 6--13, 2004.

\bibitem{HellaAggregateLogics}
Lauri Hella, Leonid Libkin, Juha Nurmonen, and Limsoon Wong.
\newblock Logics with aggregate operators.
\newblock {\em Journal of the ACM}, 48(4):880--907, 2001.

\bibitem{henkin1950}
Leon Henkin.
\newblock Completeness in the theory of types.
\newblock {\em J. Symbolic Logic}, 15(2):81--91, 06 1950.

\bibitem{hughes:modallogic1996}
G.E. Hughes and MJ~Cresswell.
\newblock {\em {A new introduction to modal logic}}.
\newblock Burns \& Oates, 1996.

\bibitem{Immerman87}
N.~Immerman.
\newblock Expressibility as a complexity measure: Results and directions.
\newblock In {\em Proceedings of Second Conference on Structure in Complexity
  Theory}, pages 194--202, 1987.

\bibitem{Leivant94}
Daniel Leivant.
\newblock Higher order logic.
\newblock In Dov~M. Gabbay, Christopher~J. Hogger, J.~A. Robinson, and
  J{\"{o}}rg~H. Siekmann, editors, {\em Handbook of Logic in Artificial
  Intelligence and Logic Programming, Volume2, Deduction Methodologies}, pages
  229--322. Oxford University Press, 1994.

\bibitem{Nanchen07}
Stanislas Nanchen.
\newblock {\em Verifying abstract state machines}.
\newblock PhD thesis, ETH Z\"urich, 2007.

\bibitem{Otto95b}
M.~Otto.
\newblock {\em Bounded variable logics and counting -- {A} study in finite
  models}, volume~9.
\newblock Springer-Verlag, 1997.

\bibitem{Otto96}
Martin Otto.
\newblock The expressive power of fixed-point logic with counting.
\newblock {\em Journal of Symbolic Logic}, 61:147--176, 1996.

\bibitem{LavalettelogicMCL01}
G.~{Renardel de Lavalette}.
\newblock A logic of modification and creation.
\newblock In {\em Logical Perspectives on Language and Information. CSLI
  publications}, 2001.

\bibitem{schewe:actacyb1993}
Klaus-Dieter Schewe and Bernhard Thalheim.
\newblock Fundamental concepts of object oriented databases.
\newblock {\em Acta Cybernetica}, 11(1-2):49--84, 1993.

\bibitem{schewe:Axiomatization}
Klaus-Dieter Schewe and Qing Wang.
\newblock A customised {ASM} thesis for database transformations.
\newblock {\em Acta Cybernetica}, 19(4):765--805, 2010.

\bibitem{Schoenegge95}
A.~Sch\"onegge.
\newblock {Extending Dynamic Logic for Reasoning about Evolving Algebras}.
\newblock Technical Report 49/95, Universit{\"a}t Karlsruhe, Fakult{\"a}t
  f{\"u}r Informatik, 1995.

\bibitem{SpruitPhDThesis94}
P.A. Spruit.
\newblock {\em Logics of Database Updates}.
\newblock PhD thesis, Faculty of Mathematics and Computer Science, Vrije
  Universiteit, Amsterdam, 1994.

\bibitem{SpruitDDL92}
P.A. Spruit, R.J. Wieringa, and J.-J.Ch. Meyer.
\newblock Dynamic database logic: The first-order case.
\newblock In U.W. Lipeck and B.~Thalheim, editors, {\em Modelling Database
  Dynamics}, pages 103--120. Springer, 1993.

\bibitem{SpruitPDDL95}
Paul Spruit, Roel Wieringa, and John-Jules Meijer.
\newblock Axiomatization, declarative semantics and operational semantics of
  passive and active updates in logic databases.
\newblock {\em Journal of Logic and Computation}, 5:27--50, 1995.

\bibitem{SpruitFUL01}
Paul Spruit, Roel Wieringa, and John-Jules Meyer.
\newblock Regular database update logics.
\newblock {\em Theoretical Computer Science}, 254(1-2):591--661, 2001.

\bibitem{RobertLogicASM}
Robert St\"ark and Stanislas Nanchen.
\newblock A logic for abstract state machines.
\newblock {\em Journal of Universal Computer Science}, 7(11), 2001.

\bibitem{Torres01}
Jose~Maria {Turull Torres}.
\newblock On the expressibility and the computability of untyped queries.
\newblock {\em Annals of Pure and Applied Logic}, 108(1-3):345--371, 2001.

\bibitem{Torres02}
Jose~Maria {Turull Torres}.
\newblock Relational databases and homogeneity in logics with counting.
\newblock {\em Acta Cybernetica}, 17(3):485--511, 2006.

\bibitem{VandenBusscheThesis}
J.~{Van den Bussche}.
\newblock {\em Formal Aspects of Object Identity in Database Manipulation}.
\newblock PhD thesis, University of Antwerp, 1993.

\bibitem{VandenBusscheNondeterministic}
Jan {Van den Bussche} and Dirk {Van Gucht}.
\newblock Non-deterministic aspects of object-creating database
  transformations.
\newblock In {\em Selected Papers from the Fourth International Workshop on
  Foundations of Models and Languages for Data and Objects}, pages 3--16.
  Springer-Verlag, 1993.

\bibitem{VandenBusscheSemideterminism}
Jan {van den Bussche} and Dirk {van Gucht}.
\newblock A semideterministic approach to object creation and nondeterminism in
  database queries.
\newblock {\em J. Comput. Syst. Sci.}, 54(1):34--47, 1997.

\bibitem{VandenBusscheCompleteness}
Jan {Van Den Bussche}, Dirk {Van Gucht}, Marc Andries, and Marc Gyssens.
\newblock On the completeness of object-creating database transformation
  languages.
\newblock {\em Journal of the ACM}, 44(2):272--319, 1997.

\bibitem{LogicSurvey01}
Pascal van Eck, Joeri Engelfriet, Dieter Fensel, Frank van Harmelen, Yde
  Venema, and Mark Willems.
\newblock A survey of languages for specifying dynamics: A knowledge
  engineering perspective.
\newblock {\em IEEE Transactions on Knowledge and Data Engineering},
  13(3):462--496, 2001.

\bibitem{DBLP:wangphdbook}
Qing Wang.
\newblock {\em Logical Foundations of Database Transformations for
  Complex-Value Databases}.
\newblock Berlin, Germany: Logos-Verlag, 2010.

\end{thebibliography}

\end{document}